\def\draft{0}  

\documentclass[11pt]{article}
\usepackage[letterpaper,hmargin=1in,vmargin=1in]{geometry}

\usepackage{geometry, color}                
\usepackage{fullpage}
\usepackage{graphicx}
\usepackage{amssymb}
\usepackage{epstopdf}
\usepackage{amsfonts,latexsym,graphics}
\usepackage{epsfig}
\usepackage{amssymb}
\usepackage{amsmath}
\usepackage{amsfonts}
\usepackage{algorithm}
\usepackage{algorithmic}
\usepackage{url}
\usepackage{rotating}
\usepackage{hyperref}
\usepackage{eepic}
\usepackage{sectsty}

\ifnum\draft=1
\newcommand{\Knote}[1]{{[\bf Kai-Min's Note: #1]}}
\newcommand{\Wnote}[1]{{[\bf Xiaodi's Note: #1]}}
\newcommand{\Lnote}[1]{{[\bf Xin's Note: #1]}}
\newcommand{\cmt}[1]{{\color{red} #1}} 
\else
\newcommand{\Knote}[1]{{}}
\newcommand{\Wnote}[1]{{}}
\newcommand{\Lnote}[1]{{}}
\newcommand{\cmt}[1]{{#1}} 
\fi

\newtheorem{theorem}{Theorem}
\newtheorem{definition}[theorem]{Definition}

\newtheorem{lemma}[theorem]{Lemma}

\newtheorem{proposition}[theorem]{Proposition}

\newtheorem{corollary}[theorem]{Corollary}

\newtheorem{fact}[theorem]{Fact}

\newtheorem{remk}[theorem]{Remark}
\newenvironment{remark}{\begin{remk} \begin{normalfont}}{\end{normalfont}
\end{remk}}

\newenvironment{proof}{\noindent{\bf Proof. }}{\qed}

\newcommand{\SRExt}{\mathrm{SRExt}}
\newcommand{\AND}{\mathrm{AND}}

\newcommand{\Raz}{\mathrm{Raz}}
\newcommand{\IP}{\mathrm{IP}}
\newcommand{\QExt}{\mathrm{QExt}}

\def\FullBox{\hbox{\vrule width 8pt height 8pt depth 0pt}}

\def\qed{\ifmmode\qquad\FullBox\else{\unskip\nobreak\hfil
\penalty50\hskip1em\null\nobreak\hfil\FullBox
\parfillskip=0pt\finalhyphendemerits=0\endgraf}\fi}

\def\qedsketch{\ifmmode\Box\else{\unskip\nobreak\hfil
\penalty50\hskip1em\null\nobreak\hfil$\Box$
\parfillskip=0pt\finalhyphendemerits=0\endgraf}\fi}

\newcommand{\N}{{\mathbb{N}}}
\newcommand{\Z}{{\mathbb Z}}

\newcommand{\C}{\mathbb{C}}
\newcommand{\poly}{{\mathrm{poly}}}
\newcommand{\polylog}{{\mathrm{polylog}}}

\newcommand{\zo}{\{0,1\}}

\newcommand{\E}{\mathop{\mathrm E}\displaylimits}

\newcommand{\remove}[1]{}
\newcommand{\Exp}{\mathop{\mbox{\sc E}}\nolimits}

\newcommand{\eps}{\varepsilon}

\def\01{\{0,1\}}
\def\eps{\epsilon}

\newcommand{\Prob}{{\mathbf{Pr}}}
\newcommand{\tinyspace}{\mspace{1mu}}
\newcommand{\microspace}{\mspace{0.5mu}}

\newcommand{\norm}[1]{\left\lVert\tinyspace#1\tinyspace\right\rVert}

\newcommand{\defeq}{\stackrel{\smash{\text{\tiny def}}}{=}}
\newcommand{\tr}{\operatorname{tr}}

\newcommand{\ip}[2]{\left\langle #1 , #2\right\rangle}

\def\({\left(}
\def\){\right)}
\def\I{\mathsf{id}}

\newcommand{\setft}[1]{\mathrm{#1}}
\newcommand{\lin}[1]{\setft{L}\left(#1\right)}
\newcommand{\density}[1]{\setft{Dens}\left(#1\right)}

\newcommand{\ot}{\otimes}

\def\complex{\mathbb{C}}

\def\<{\langle}
\def\>{\rangle}
\def \lket {\left|}
\def \rket {\right\rangle}

\def \rbra {\right|}
\newcommand{\ket}[1]{\lket\microspace #1 \microspace\rket}

\newcommand{\ketbra}[1]{\lket\microspace #1 \rangle \langle #1 \microspace\rbra}

\def\X{\mathcal{X}}
\def\Y{\mathcal{Y}}
\def\Z{\mathcal{Z}}
\def\W{\mathcal{W}}
\def\A{\mathcal{A}}
\def\B{\mathcal{B}}

\def\C{\mathcal{C}}

\def\E{\mathcal{E}}

\newcommand{\trnorm}[1]{\norm{#1}_{\tr}}

\newcommand{\trdist}[1]{ \left | #1 \right |_{\rm tr}}
\newcommand{\uniform}[1]{\mathcal{U}_{#1}}

\def\defeq{\stackrel{\small \textrm{def}}{=}}

\newcommand{\commentout}[1]{}
\numberwithin{theorem}{section}
\numberwithin{equation}{section}

\newenvironment{step}
  {
    \begin{enumerate}

  }
  {\end{enumerate}}

\newenvironment{protocol*}[1]
  {
    \begin{center}
      \hrulefill\\
      \textbf{#1}
  }
  {
    \vspace{-1\baselineskip}
    \hrulefill
    \end{center}
  }

\newcommand{\Ext}{{\mathrm{Ext}}}
\newcommand{\Hmin}{H_{\mathrm{min}}}
\def \IExt {\mathrm{IExt}}
\def \TExt {\mathrm{2Ext}}
\def \QMExt {\mathrm{QMExt}}
\def \QTExt {\mathrm{QTExt}}
\def \Adv {\mathrm{Adv}}

\def \polylog {\mathrm{polylog}}

\def \QBSAExt {\mathrm{QBExt}}
\def \BExt {\mathrm{BExt}}

\def \Faulty {\mathsf{Faulty}}
\def \Good {\mathsf{Good}}
\def \Bad {\mathsf{Bad}}
\def \PubExt {\Ext_\mathrm{Pub}}
\def \NetExt {\Ext_\mathrm{Net}}
\def \PriExt {\Ext_\mathrm{Pri}}
\def \OAExt {\mathrm{OAExt}}

\def \AdvSI {\Adv_{\mathrm{SI}}}
\def \AdvNet {\Adv_{\mathrm{Net}}}

\def \NSA {NSA }

\title{Multi-Source Randomness Extractors Against Quantum Side Information, and their Applications}
\author{%
 Kai-Min Chung\footnote{Institute of Information Science, Academia Sinica, Taiwan.}  $\qquad$
 Xin Li\footnote{Department of Computer Science, Johns Hopkins University. } $\qquad$
 Xiaodi Wu\footnote{Center for Theoretical Physics,
Massachusetts Institute of Technology, Cambridge, MA 02139, USA. Part of research was conducted while the author was a Research Fellow at the Simons Institute for the Theory of Computing, University of California, Berkeley, CA 94720, USA. XW was funded by ARO contract W911NF-12-1-0486 and by the NSF Waterman Award of Scott Aaronson.}
}

\date{}                                           

\begin{document}

\begin{titlepage}
\maketitle

\begin{abstract}
We study the problem of constructing multi-source extractors in the quantum setting, which extract almost uniform random bits against an adversary who collects quantum side information from several initially independent classical random sources. This is a natural generalization of the two much studied problems of seeded randomness extraction against quantum side information, and classical independent source extractors. With new challenges such as potential entanglement in the side information, it is not a prior clear under what conditions do quantum multi-source extractors exist; the only previous work in this setting is \cite{KK12}, where the classical inner-product two-source extractors of \cite{CG88} and \cite{DEOR04} are shown to be quantum secure in the restricted \emph{Independent Adversary (IA) Model} and \emph{entangled Bounded Storage (BS) Model}.

In this paper we propose a new model called \emph{General Entangled (GE) Adversary Model}, which allows arbitrary entanglement in the side information and subsumes both the IA model and the BS model. We proceed to show how to construct GE-secure quantum multi-source extractors. 
To that end, we propose another model called \emph{One-sided Adversary (OA) Model}, which is weaker than all the above models. Somewhat surprisingly, we establish an equivalence between strong OA-security and strong GE-security. As a result,  all classical multi-source extractors can either directly work, or be modified to work in the GE model at the cost of one extra random source. Thus, our constructions essentially match the best known constructions of classical multi-source extractors. This answers several open questions in \cite{KK12,DPVR12}.  

We also apply our techniques to two important problems in cryptography and distributed computing --- \emph{privacy amplification} and \emph{network extractor}. Both problems deal with converting secret weak random sources into secret uniform random bits in a communicating environment, with the presence of a passive adversary who has unlimited computational power and can see every message transmitted. We show that as long as the sources have certain amounts of conditional min-entropy in our GE model (even with entangled quantum side information), we can design very efficient privacy amplification protocols and network extractors. 
\end{abstract}

\vfill
\thispagestyle{empty}
\end{titlepage}

\tableofcontents
\newpage


\section{Introduction} \label{sec:intro}
The enormous benefit of using randomness in computation has been witnessed by the vast number of applications in algorithms, distributed computing, cryptography and many more.  
However, often the random sources in nature are imperfect with various biases and dependence. In many applications these imperfect random sources need to be distilled before they can be used. Randomness extractors are tools for this distilling process --- they convert imperfect random sources into nearly uniform random bits.


A random source can be imperfect for two reasons. First, it can have natural biases. This occurs in for example thermal noises or computer mouse movements. Second, and more importantly in applications related to security and privacy, it becomes imperfect because an adversary manages to learn some side information about the source. Here, naturally we also require the output of the randomness extractor to be (almost) independent of the side information. In the classical setting, dealing with these two cases can often be unified by requiring the output of the extractor to be close to uniform whenever the imperfect random source has a certain amount of min-entropy:

\begin{definition}[Min-entropy] 
The \emph{min-entropy} of a random variable $X$ is given by 
\[
   \Hmin(X)=\min_{x \in \X} \log_2 (1/\Prob[X=x]).
\]
For $X\in \01^n$, we call $X$ an $(n, \Hmin(X))$-source with \emph{entropy rate} $\Hmin(X)/n$. 
\end{definition}

\begin{definition}[informal]
A (deterministic or randomized) function $\Ext: \{0,1\}^n \to \{0,1\}^m$ is an error $\eps$ extractor for a class $\cal C$ of sources with min-entropy $k$, if for any source $X \in \cal C$, we have   
\[
   \trdist{\Ext(X)-\uniform{m}} \leq \eps.
\] 
\end{definition}

The reason is that in most classical cases, we can fix the side information, and argue that conditioned on this fixing, the source still has enough min-entropy (as long as the adversary does not learn all information of the source).  Thus, the output of the extractor will be close to uniform even given the side information. This unified approach makes extractors the single tool to solve the above two different problems. The remaining question is to decide for what classes of sources we can construct extractors. For this purpose, it is not hard to show that no deterministic extractor can exist for general $(n, k)$ sources even when $k$ is as large as  $n-1$. Therefore, the study of randomness extractors has been pursued in two directions. One is to allow an extractor to use a short independent uniform random seed (i.e., $\Ext$ becomes a randomized function), and these extractors are known as \emph{seeded extractors}. The other is to construct extractors without seeds for random sources with special structures, where an important case is to extract random bits from multiple (independent) random sources. Both kinds of extractors have been studied extensively in the classical setting.

In many important problems related to cryptography and security, true (close to) uniform randomness is provably necessary. For example, Dodis et. al \cite{DodisOPS04} showed that many important cryptographic tasks, such as bit-commitment, encryption and zero-knowledge would become impossible even if the random bits used have entropy rate $0.99$.
Thus, it is important to use multi-source extractors to generate true (close to) uniform random bits for these applications. We note that in the classical setting, one can use the probabilistic method to show that very good extractors exist for just two independent weak sources with logarithmic min-entropy. This is a strict generalization of seeded extractors (where one can view the seed as another independent source) and only needs weaker requirements on the randomness used in applications. In fact, one natural and important question is what are the minimum requirements on randomness used in various applications; and in the classical setting, multi-source extractors provide an answer to this question in the case where independent weak sources can be obtained. This paper, on the other hand, can be viewed as a step towards answering the above question in the quantum setting.

Indeed, since our world is inherently non-classical, a more powerful adversary can use quantum processes to obtain the side information; and we need to define \emph{quantum conditional min-entropy} and \emph{quantum extractors} as follows.

\begin{definition}[Quantum conditional Min-entropy] 
Let $\rho_{XE} \in \density{\X\ot \E}$ be a classical-quantum state. The \emph{min-entropy} of $X$ conditioned on $E$ is defined~\footnote{This definition has a simple operational interpretation shown in~\cite{KRS09} that $\Hmin(X|E)_\rho = - \log(p_{\mathrm{guess}} (X|E)_\rho)$,
where $p_{\mathrm{guess}} (X|E)_\rho$ is the maximum probability of guessing $X$ by making arbitrary measurements on $E$ system. } as
  \begin{equation*}
    \Hmin({X|E})_\rho \defeq \max \{\lambda \geq 0 :  \exists \sigma_E \in \density{\E}, \mathrm{s.t.}\,\, 2^{-\lambda} \I_X \ot \sigma_E \geq \rho_{XE}\}.
  \end{equation*}
\end{definition}

\begin{definition}[informal]
A (deterministic or randomized) function $\Ext: \{0,1\}^n \to \{0,1\}^m$ is an error $\eps$ quantum extractor for a class $\cal C$ of sources with conditional min-entropy $k$, if for all cq states $\rho_{XE} \in \cal C$, we have
\[
   \trdist{\rho_{\Ext(X)E} - \uniform{m} \ot \rho_E} \leq \eps.
\]
\end{definition}

Quantum side information presents much more challenge than classical side information, since we do not know how to apply the technique of ``conditioning" on side information. Therefore a classical extractor is not necessarily an extractor secure against quantum side information. Indeed, Gavinsky et al. \cite{GKK+08} gave an example of a classical seeded extractor that is not secure even against a very small amount of quantum side information. As it turns out, to construct quantum seeded extractors is a non-trivial task; and today we only  have a few constructions of such extractors, with parameters much worse than the best known classical seeded extractors. For example, K$\ddot{o}$nig, Maurer, and Renner \cite{KMR05, Renner05, RK05} showed that seeded extractors based on the leftover hash lemma \cite{hill, ILL89} are quantum secure, and K$\ddot{o}$nig and Terhal \cite{KT08} showed that any one-bit output extractor is also quantum secure, with roughly the same parameters. Ta-Shma \cite{TaS11}, De and Vidick \cite{DV10}, and later De, Portmann, Vidick and Renner \cite{DPVR12} gave quantum seeded extractors with short seeds that can extract almost all of the min-entropy\footnote{Although the seed length is still much longer compared to the best known classical seeded extractor.}. All of these three constructions are based on Trevisan's extractor \cite{Tre01}. It remains an open problem to construct quantum seeded extractors that match the parameters of the best known classical seeded extractors.

In the multi-source case, the situation is even worse. This is because measuring each source's quantum side information might break the independence of the sources --- a condition that is needed in classical multi-source extractors. Moreover, the quantum side information of each source can have \emph{entanglement} --- a phenomenon that does not exist in the classical setting. Quantum entanglement yields several surprising effects that cannot happen in the classical world, such as non-local correlation \cite{Bell64} and superdense coding \cite{BW92}. These issues apparently make the task of constructing quantum multi-source extractors much harder than constructing classical multi-source extractors. Indeed, it is a prior not clear under what conditions do quantum multi-source extractors exist (this is in sharp contrast to the classical setting, where the existence of very good two-source extractors is guaranteed by the probabilistic method); and it was only very recently that \cite{KK12} gave a construction of two-source extractors in the independent adversary model (which roughly corresponds to independent sources in the classical setting), and the very restricted \emph{entangled bounded storage} model. 

However, the results of \cite{KK12} are still very limited and do not give us a clear picture of quantum multi-source extractors. The main reason is that in the case of \emph{independent adversary} model, it does not allow entangled side information; while in the case of \emph{entangled bounded storage} model, it uses a very special method to show that a particular function (namely the inner product function) is a two-source extractor. For this function to be a two-source extractor, we need to require that the two sources have large min-entropy (i.e., roughly have min-entropy rate $>1/2)$. 
On the other hand, this technique of showing a two-source extractor in the entangled bounded storage model seems hard to generalize to other functions (e.g., other classical two-source extractor constructions). Thus, given these results it is still not clear if two-source (or multi-source) extractors can exist for smaller min-entropy, in the entangled bounded storage model. 

\subsection{Sketch of Our Results}
In this paper we significantly improve the situation in the case of multi-source extractors. We show, somewhat surprisingly, that in a more general model, we can actually construct quantum multi-source extractors that essentially match the best constructions of classical multi-source extractors, even in the presence of entangled quantum side information. Our model is so general that it subsumes both the independent adversary model and the bounded storage model, and parallels what can be achieved in the classical setting. Indeed, our model is a strict generalization of the independent sources model in the classical setting, and we actually show that any classical multi-source extractor can either directly work, or be modified to work in our general model with roughly the same parameters. This not only establishes the existence of multi-source extractors (e.g., two-source extractors for logarithmic min-entropy) in the presence of (even entangled) quantum side information, but also gives us a general way to construct them. In particular, we answer several open questions in \cite{KK12,DPVR12} and give stronger results and simplified proofs. We view this new model as one of our main conceptual contributions.\ We then apply our techniques to two important problems in cryptography and distributed computing.


\paragraph{Privacy Amplification.} The most important application of seeded quantum extractors is privacy amplification with quantum side information. The setting is that two parties (Alice and Bob) share a secret weak random source $X$. They each also has local private random bits. The goal is to convert the shared weak source $X$ into a nearly uniform random string by having the two parties communicating with each other. However, the communication channel is watched by a (passive) adversary Eve, and we want to make sure that eventually the shared uniform random bits remain secret to Eve. In the quantum setting, Eve may also have quantum side information about the shared source $X$.

One can use strong (classical or quantum) seeded extractors to solve this problem in one round by having one party (say Alice) send a seed to Bob and they each apply the extractor to the shared source using the seed. The strong property of the extractor guarantees that even if seeing the seed, Eve has no information about the extracted uniform key. One advantage of this method is that if we have good strong seeded extractors, then we can just use a short seed to extract a long shared key. 

However, as mentioned before, it is not clear that we can simply assume that the two parties have local uniform random bits. They may well only have weak random seeds which may also be subject to (entangled) quantum side information. In this paper we show that as long as the two parties' local random seeds have arbitrarily constant min-entropy rate as measured in our general model, we can still achieve privacy amplification with asymptotically the same parameters. In particular, this keeps the nice property that we can use a short seed to extract a long uniform key. Note that in our model, the two parties' local random bits may be subject to entangled quantum side information with the shared weak source, and we show that even in this case privacy amplification can be achieved. 

As a by-product, we also give a general transformation that can convert any (classical or quantum) strong seeded extractor into another (classical or quantum) strong seeded extractor with roughly the same output size and error, and a constant factor larger size of seed, with the property that the new strong seeded extractor works as long as the entropy rate of the seed is at least $1/2+\delta$ for any constant $\delta>0$.\footnote{\cite{Raz05} also has a similar transformation that can convert any classical seeded extractor into another classical seeded extractor that works as long as the seed has entropy rate $1/2+\delta$. However, that transformation may not keep the property of strong extractors.} Other known constructions of strong quantum seeded extractors that can work with a weak random seed, such as that in \cite{DPVR12} requires the seed to have entropy rate at least $0.9$.

\paragraph{Network Extractor.} One of the main applications of multi-source extractors in the classical setting is in distributed computing and secure multi-party computation problems where multiple players each has an imperfect random source. The players then need to communicate with each other to convert their random sources into nearly uniform and private random bits. Therefore, we need to design a protocol, known as \emph{network extractor protocol}, as defined in \cite{KLRZ}. Here, the setting is that part of the players are corrupted by an adversary, who then manipulates these players to try to collapse the protocol. As in \cite{KLRZ}, we allow the adversary to have unlimited computational power, see every message transmitted in the protocol, and wait to transmit the faulty players' messages after seeing all the other players' messages (this is called \emph{rushing}). When each player leaks some side information, we require that a set of honest players end up with (almost) private and uniform random bits even given all the side information and the whole transcript of the protocol; and the goal is to sacrifice as few honest players as possible. We note that this problem can be viewed as a generalization of the multi-source extractor problem to the distributed and adversarial setting. A multi-source extractor can be thought as a network extractor with no faulty players. It is the existence of the network adversary that makes the construction of network extractors more challenging. 

Another important thing to notice here is that in the network extractor model, we essentially have \emph{two} adversaries. One adversary, which we call $\AdvSI$, obtains side information from the players' sources; while the other adversary,  which we call $\AdvNet$, controls the faulty players to try to collapse the protocol. These two adversaries may or may not collaborate. If they do not collaborate, then $\AdvNet$ only makes rushing choices based on the public messages. We call this strategy \emph{independent rushing}. 
On the other hand, if they do collaborate, then the adversary becomes more powerful --- he can use the quantum side information (in addition to the messages) to make the rushing choices. By doing this, the adversary can generate complicated correlations between different parts of the network source system, even if originally the side information is obtained in the \emph{independent adversary} model. This phenomenon is special in the quantum setting and we call this strategy \emph{quantum rushing}. It is conceivable that quantum rushing is much more difficult to handle than independent rushing, because of the potential entanglement the adversary can create. Nevertheless, we give network extractors in the presence of quantum side information (even entangled) in the case of both independent rushing and quantum rushing. In the former case, we can essentially match the performance of classical network extractors (in fact, our construction improves and simplifies existing construction of~\cite{KLRZ}); while in the latter case, we need to sacrifice a constant factor more of honest players.



\subsection{Our New Model}
Traditionally, extractors are designed to work whenever the class of sources satisfy a certain requirement on min-entropy (or quantum conditional min-entropy). An example in \cite{KK12} showed that the ``min-entropy requirement" may be problematic and this motivates \cite{KK12} to consider the more restricted bounded storage model. In this paper we rectify this problem and go back to the standard min-entropy requirement. To describe our new model, let us first revisit the example in \cite{KK12}.

First recall the following process where the adversary obtains quantum side information. Initially we have  $t$ non-communicating parties, each of which has a classical independent random source $X_i$.  The adversary $\Adv$ then prepares a quantum state $\rho_0$ on registers $A_1, \cdots A_t$  (independent of the $X_i$s, but could be arbitrary entangled) and sends each register $A_i$ to the $i$'th party who holds $X_i$. The $i$'th party then applies some operation on $X_i$ and $A_i$ to produce the leakage $E_i$. 
Finally, the adversary collects all $E_i$s as the side information of the sources $X_1, \cdots, X_t$.

The example in \cite{KK12} is classical but demonstrates the kind of problems that one may face when presented with entangled side information. Suppose Alice and Bob have two classically independent uniform $n$-bit sources $X$ and $Y$, and the adversary Eve prepares two identical copies of another uniform $n$-bit random string $R$, which is independent of $(X, Y)$. Eve then sends the two copies of $R$ to Alice and Bob, and obtains side information $E_a=X \oplus R$ and $E_b=Y \oplus R$ respectively. Note that conditioned on $(E_a, E_b)$, both $X$ and $Y$ have full entropy. 
Now suppose further that Eve  obtains $|X| \text{ mod } 4$ from Alice and $|Y|  \text{ mod } 4$ from Bob\footnote{$|X|$ and $|Y|$ are the hamming weights of $X$ and $Y$.}, which reduces the conditional min-entropy of $X$ and $Y$ by at most a constant.
However the classical inner-product two-source extractor $X \cdot Y$, which works if $X$ and $Y$ are two independent sources with min-entropy $>n/2$, completely fails in this case since one can compute $X \cdot Y=\frac{1}{2} ((|X|+|Y|-|X \oplus Y|)  \text{ mod } 4)$.  
\cite{KK12} thus argues that this model (requiring that each source has enough min-entropy given \emph{all} side information) may be problematic. 

Our crucial observation is that what this example tells us is not that the conditional min-entropy requirement is problematic, but that \emph{the way the conditional min-entropy is measured} is problematic. More specifically, once the adversary learns $E_a=X \oplus R$ and $E_b=Y \oplus R$, there is a bijection between $X$ and $R$, i.e., $X = R \oplus E_a$; and similarly, there is a bijection between $Y$ and $R$, i.e., $Y = R \oplus E_b$. Thus, given the side information $(E_a, E_b)$, there is a bijection between $X$ and $Y$, i.e., $X =Y \oplus (E_a \oplus E_b)$. This means that, although both $X$ and $Y$ have high conditional min-entropy, $X$'s entropy now comes from $Y$ and vice versa. In other words, this way of measuring conditional min-entropy creates \emph{interference} between the entropies of different sources, and causes \emph{double counting} of entropies. Thus the result that traditional extractors such as $X \cdot Y$ may fail should come as no surprise.

This problem is actually quite general in the case of entangled side information. Whenever one tries to measure a source's conditional min-entropy given \emph{all} side information, it is likely to create interference among the sources. To rectify this problem, we  choose an alternative way to measure the conditional min-entropy: for any source $X$, we imagine that the adversary first obtains some side information from $X$ without obtaining any side information from the other sources. We propose to measure $X$'s conditional min-entropy \emph{immediately after} this step, and \emph{right before} the adversary obtains any side information from the other sources. In this way we can ensure that the measured conditional min-entropy is specific to this particular source, and does not interfere with any other source. Our model now requires each source to have sufficient conditional min-entropy according to this way of measurement. We call this model the general entangled (GE for short) model. A formal definition is given in Section~\ref{sec:model}.

Going back to the above example, if we measure conditional min-entropy in our GE model, then we see that immediately after Eve obtains $E_a=X \oplus R$, Eve still has a copy of $R$ (which he has not sent to Bob yet). Thus, at this moment $X$'s conditional min-entropy is $0$ (since $X=E_a \oplus R$). Therefore, this example is not a counterexample in our model. 


We remark that our proposed GE model has a few nice and important properties. 
First, it is not hard to see that our GE model is a strict generalization of the no-side-information case, no matter in the way the side information is generated or the entropy is measured.
Second, the GE-entropy measure, similar to the classical min-entropy measure, captures the amount of uniform randomness that can be extracted from the source in the presence of GE-side information.
This is because all of the GE-entropy can be extracted and there exists sources with certain GE-side information, in which the GE-entropy also upper bounds the amount of uniform randomness that can be extracted. 
Finally, we argue that the one-round side-information-generating process in our GE model might be also appealing due to practical reasons. 
For example, if the side information is generated simultaneously at distant parties each holding one of the sources, then it can effectively be characterized by the one-round process. We refer curious readers to Section~\ref{sec:model} for details. 

\vspace{-1mm} \paragraph{Special cases.} We now briefly discuss some other models and their relations. In particular, \cite{KK12} considered the following two models: the \emph{Independent Adversarial (IA) Model} and the \emph{Bounded Storage (BS) Model}. The IA model poses one additional constraint on the GE model: that is the initial state $\rho_0$ is a product state over $A_1, \cdots, A_t$, i.e., $\rho_{A_1, \cdots, A_t}= \rho_{A_1} \ot \cdots \ot \rho_{A_t}$. Thus, by definition, $\rho_{E_1, \cdots E_t}$ is also a product state. The measurement of conditional min-entropy in our general GE model reduces exactly to $\Hmin(X_i|E_i)_\rho$ for each $X_i$.

The BS model poses a different constraint on the GE model: that is to bound the dimension of each register $E_i$ by $2^{b_i}, \forall i \in [t]$.  In this case,  the quality of the source $X_i$ is measured by its marginal min-entropy $k'_i=\Hmin(X_i)$ and the size bound $b_i$ on each register $E_i$. However, we can show that our measurement of conditional min-entropy in the GE model here is at least $k'_i-2b_i$, in which the factor two is due to the possibility of super-dense coding. Therefore, it should also be clear that our model subsumes both the IA model and the BS model. 

We now define another model, the \emph{One-sided Adversary (OA) Model}. Here the adversary is restricted to collect leakage information from only one source $X_i$ but has the freedom to choose which $i\in [t]$. Namely, only one $A_{i^*}$ is nonempty among all $A_i$s for some $i^*$.  
This is the weakest model of all.


\subsection{OA-GE Security Equivalence}

Somewhat surprisingly, we show an equivalence between strong security in the OA model (which is the weakest) and strong security in the GE model (which is the strongest). We then use this equivalence to give simple constructions of quantum multi-source extractors and network extractors in the GE model. This equivalence is one of our major results and another conceptual contribution of this paper.

Our security equivalence is established by a simulation argument, which we now illustrate in the context of strong two-source extractors. Consider a OA-secure \emph{$Y$-strong} two-source extractor $\Ext(X,Y)$ for min-entropy $k$ sources with error $\eps$. That is, for every sources $(X,Y)$ where both $X, Y$ have min-entropy $k$ in OA model, $\Ext(X,Y)$ is $\eps$-close to uniform given $Y$ and the side information. Consider a source $(X,Y)$ that both $X, Y$ have min-entropy $k$ w.r.t. GE side information adversary $\Adv_{GE}$, who sends registers $A_1$ and $A_2$ to $X$ and $Y$ respectively to collect side information $E_1$ and $E_2$. Consider a hybrid adversary $\Adv'$ who only sends $A_1$ to $X$ but keeps $A_2$ inside itself.\footnote{Technically, in our formal model, we do not allow $\Adv'$ to keep local register, so we instead let $\Adv'$ sends $A_2$ to $X$, and have $X$ send $A_2$ back.} Note that $\Adv'$ is a OA side information adversary, and $X$ has \emph{the same} amount of min-entropy w.r.t. $\Adv_{GE}$ and $\Adv'$ (since the entropy is measured \emph{immediately after} the adversary obtain the side information $E_1$ from $X$). Thus, $\Ext(X,Y)$ is $\eps$-close to uniform given $Y$ and the side information $(E_1,A_2)$ collected  by $\Adv'$. Now, note that given $Y$ and the side information collected by $\Adv'$, we can \emph{simulate} the side information of $\Adv_{GE}$ by internally applying leaking operation on $Y$ and $A_2$ to produce $E_2$, which can only decrease the trace distance. Therefore, $\Ext(X,Y)$ is also $\eps$-close to uniform given $Y$ and the side information $(E_1,E_2)$ collected  by $\Adv_{GE}$. Note that this simulation argument crucially relies on the \emph{strong} property of the extractors. 

The above simple yet powerful argument can be generalized to the setting of multi-source extractors that are strong on all-but-one sources (formally stated in Theorem~\ref{thm:SOA_GE} in Section~\ref{sec:OA_to_GEA}). Furthermore, it also extends to establishing equivalence of strong OA and GE security for honest players in network extractors with independent rushing (formally stated in Theorem~\ref{thm:CR_QR}), where strong security requires the player's output remains (close to) uniform even given all other players' inputs (and the transcript). The equivalence allows us to reduce the goal of achieving strong GE security in these settings to strong OA security, which is much simpler to achieve in general. We are able to develop several techniques for obtaining strong OA security, and thus provides strong GE-secure multi-source/network extractors that essentially match the best known parameters (without side information) for these settings.


\subsection{Multi-source Extractors with Quantum Side Information} \label{sec:intro:multi}
In the classical setting, using the probabilistic method one can show that an extractor exists for two independent $(n, k)$ sources with $k$ as small as $\log n+O(1)$. However constructing such extractors turn out to be a very hard problem. Historically, Chor and Goldreich \cite{CG88} were the first to formally study multi-source extractors, where they constructed explicit extractors for two independent $(n, k)$ sources with $k \geq (1/2+\delta)n$ for any constant $\delta>0$. After that there had been essentially no progress for two decades until Barak, Impagliazzo and Wigderson \cite{BarakIW04} showed how to extract from a constant number ($\poly(1/\delta))$ of independent $(n, \delta n)$ sources, for any constant $\delta>0$. Their work used advanced techniques from additive combinatorics. Since then, new techniques for this problem have emerged, resulting in a long line of research \cite{BarakKSSW05, Raz05, Bour05, Rao06, BarakRSW06, Li11, Li13a, Li13b} and culminating in Li's extractor for a constant number of independent $(n, k)$ sources with $k=\polylog(n)$ \cite{Li13b}. In the two-source setting, Bourgain's extractor \cite{Bour05} works for two independent $(n, k)$ sources with $k \geq (1/2-\delta)n$ for some universal constant $\delta>0$, which is the state of art.

In the quantum setting, the formal study of quantum multi-source extractors started with \cite{KK12}, who focused on analyzing a two-source extractor of Dodis, Elbaz, Oliveira, Raz~\cite{DEOR04} (which in turn based on the construction of Chor and Goldreich~\cite{CG88}) in the aforementioned independent adversary model and entangled bounded storage model.  \cite{KK12} showed that the DEOR extractor is secure in the BS model by a connection to communication complexity and establishing a communication complexity lower bound, and showed the security of the DEOR extractor in the IA model by first establishing security of its one-bit version (following~\cite{KT08}) and then appealing to a quantum version of XOR lemma. In both models, they established the (strong) security of the DEOR extractor with slightly degraded parameters; and this is currently the only known work about quantum multi-source extractors.

\paragraph{One-bit Argument.} We observe that, the argument of~\cite{KK12} for establishing IA security is in fact general, and can be used to establish strong OA security of best known two-source extractors~\cite{Raz05,Bour05,DEOR04}, or the existential two-source extractors for logarithmic min-entropy guaranteed by the probabilistic method with essentially matching parameters. Armed with our security equivalence result, it immediately implies that all known two-source extractors~\cite{Raz05,Bour05,DEOR04} are in fact strongly GE-secure.


\begin{theorem}[informal] There exist two-source extractors for logarithmic min-entropy that are strongly GE-secure.
\end{theorem}

\begin{theorem}[informal, refer to Theorem~\ref{thm:Raz-GE-secure}, \ref{thm:Bou-GE-secure}, and \ref{thm:DEOR-GE-secure}] The two-source extractors of Bourgain~\cite{Bour05}, Raz~\cite{Raz05}, and DEOR~\cite{DEOR04} are strongly GE-secure.
\end{theorem}

\paragraph{One-extra-source Argument.} In the multi-source setting, it turns out we could avoid the parameter loss in the quantum XOR lemma at the cost of an extra independent source. Our crucial observation is that for any marginally close-to-uniform distribution, one can add an independent quantum min-entropy source and make use of a quantum strong seeded extractor to lift its security from marginal to strong OA.
This observation is so powerful that it suffices to lift the security at the last step of the construction and work only with marginal security in all previous steps. 
Again, with our security equivalence result,  we can construct a strong GE-secure multi-source extractor from \emph{any} known classical independent source extractors.

\begin{theorem}[informal, refer to Theorem~\ref{thm:Ind_QIExt}]
From any independent source extractor $\IExt$ with $t$ sources, one can explicitly construct a GE-secure strong extractor  $\QMExt$ 
with $t+1$ sources. 
\end{theorem}


\begin{corollary}[informal, refer to Theorem~\ref{thm:Li_GE} and \ref{thm:BRSW_GE}] There exist explicit multi-source extractors based on the one of Li~\cite{Li13b}, or  BRSW~\cite{BarakRSW06,Rao06} that are strongly GE-secure. 
\end{corollary}

\paragraph{One-extra-block Argument.} In the context of block+general source extractors (e.g., \cite{BarakRSW06}), one can use an extra block to the existing block source and make use of one classical and one quantum strong seeded extractor to lift its security. 
Comparing to the one-extra-source technique, we only require to add one extra block that is not independent of existing sources.
Conceivably, this is a strictly more difficulty task, which is resolved by the technique called \emph{alternating extraction}. 

\begin{theorem}[informal, refer to Theorem~\ref{thm:one_extrac_block}]
From any strong block+general source extractor $\BExt$ with $C$ blocks, one can explicitly construct a GE-secure strong block+general extractor $\QBSAExt$ 
with $C+1$ blocks. 
\end{theorem}

\begin{corollary}[informal, refer to Theorem~\ref{thm:2-block-general-ext1} and Theorem~\ref{thm:2-block-general-ext}]
The block+general source extractors based on the one of BRSW~\cite{BarakRSW06}, or Raz~\cite{Raz05}  are strongly GE-secure. 
\end{corollary}

\subsection{Privacy Amplification with Weak Sources}
To show how to achieve privacy amplification with local weak random bits, 
we first give an extractor for a source $X=(X_1, X_2)$ and an independent $(n_3, k_3)$ source $X_3$, where $X_1$ is an $(n_1, k_1=\delta n_1)$ source for any constant $\delta>0$ and conditioned on $X_1$, $X_2$ is an $(n_2, k_2)$ source (i.e., $X$ is a block source). Our construction is simple. We first use the sum-product theorem based condenser in \cite{BarakKSSW05, Zuc07} to convert $X_1$ into a matrix of $D=O(1)$ rows such that one row is $2^{-\Omega(n_1)}$-close to having entropy rate $0.9$.\footnote{Strictly speaking, it is a convex combination of such matrices, but it does not make a difference to our analysis.} Then we use each row in this matrix and the strong two-source extractor $\Raz$ in \cite{Raz05} to extract an output from $X_3$ and concatenate the outputs to obtain a somewhere random source $W$. This step works because $\Raz$ works if one of the inputs has entropy rate $>0.5$ and indeed one row in our matrix has entropy rate $0.9$. Since $\Raz$ is strong,  even conditioned on $X_1$, $W$ is still somewhere random. We can also limit the size of each output in $W$ so that conditioned on $W$, $Y$ still has a lot of min-entropy. Now since $W$ only has a constant number of rows, we can use $W$ and a strong extractor from \cite{Rao06, BarakRSW06} to extract a uniform random string $V$ from $X_2$.  Conditioned on the fixing of $X_1$ and $W$, $V$ and $X_3$ are independent. We can take a classical strong seeded extractor $\Ext_c$ and use $V$ as a seed to extract a uniform random string from $X_3$, which gives us a classical $X$-strong extractor. 

The above argument naturally extends to the OA model, where we replace $\Ext_c$ by  a quantum strong seeded extractor $\Ext_q$ at the last step. The analysis turns out to be a special case of the ``one-extra-block" argument mentioned in the last section. Our OA-GE equivalence will then establish that the resulted extractor is a GE-secure $X$-strong extractor. 
See Section~\ref{sec:three_Ext} for details. 


We further observe that the above extractor gives us a general way to transform any classical or quantum strong seeded extractor into another strong seeded extractor that works as long as the seed has entropy rate $\geq 1/2+\delta$.\footnote{It is easy to see that one can divide the seed into two equal blocks and they form a block source with each block having entropy rate at least $\delta/2$}
In privacy amplification, if either party's local random source has entropy rate $1/2+\delta$, then we can just use this strong extractor. Otherwise, if the parties both have local sources with entropy rate $\delta$, then we can have both parties send their sources to each other and they just apply the original strong extractor (notice that the sources of the two parties form a block source $X=(X_1, X_2)$). Note that we can output a constant fraction of the entropy of $X$ in $V$, thus the size of $X$ only needs to be a constant factor larger than what is needed in the case when we have uniform random seeds. See Section~\ref{sec:PA} for details. 


\subsection{Network Extractor with Quantum Side Information} \label{sec:intro:network}
In the classical setting, network extractors are motivated by the problem of using imperfect randomness in distributed computing, a problem first studied by \cite{GoldwasserSV05}. Kalai, Rao, Li, and Zuckerman formally defined network extractors in \cite{KLRZ}, and gave several efficient constructions for both synchronous networks and asynchronous networks, and both the information-theoretic setting and the computational setting. For simplicity and to better illustrate our ideas, in this paper we will focus on synchronous networks and the information-theoretic setting. 

Following \cite{KLRZ}, we gave an informal definition of network extractors with quantum side information here. A formal definition is given in Section~\ref{sec:network}. We assume a set of $p$ players such that $t$ of them are corrupted by an adversary $\AdvNet$. Each (honest) player has an independent source $X_i$, and a side information adversary $\AdvSI$ collects side information $\rho$ from the sources $X = (X_1,\dots, X_p)$. We assume each $X_i$ has length $n$ and conditional min-entropy at least $k$ measured in our GE model. Depending on the case of independent rushing (IR) or quantum rushing (QR), $\AdvNet$ and $\AdvSI$ may or may not collaborate.

At the conclusion of protocol execution, let $T$ denote the transcript of protocol messages that are public, and $Z_i$ be the private output of (honest) player $i$.

\begin{definition} A protocol $\NetExt$ is a $(t, g ,\eps)$ network extractor for adversary $\Adv=(\AdvSI, \AdvNet)$ if at the end of the protocol, there exists a subset of honest players $S$ with $|S| \geq g$ such that
$$ \trdist{\rho_{Z_{S} Z_{-{S}} T \Adv}-  U \ot \rho_{Z_{-S} T \Adv}}\leq \eps,$$ where $Z_{S}$ and $Z_{-{S}}$ denote the outputs of the players in $S$ and the outputs of the players outside of $S$ respectively.
\end{definition}


We can now informally state our results.  For the case of independent rushing, we are able to tolerate close to $1/3$-fraction of faulty players, scarify only roughly $t$ honest players, and extract almost all entropy out even for low entropy $k = \polylog(n)$.

\begin{theorem}[IR-secure Network Extractor] \label{thm:CRnext} For every constants $\alpha < \gamma \in (0,1)$, $c > 0$, and sufficiently large $p,t,n,k$ s.t. $p \geq (3+\gamma) t$ and $k \geq \log^{10}n$, there exists a 3-round $(t, p - (2+\alpha) t, n^{-c})$ network extractor $\NetExt$ with output length $m =k - o(k)$ in the independent rushing case.
\end{theorem}

We note that even in the classical setting (with no side information), Theorem~\ref{thm:CRnext} is the best known. Essentially, this result matches the best known network extractor in the classical setting and improves the results in \cite{KLRZ}. The reasons are that (i) at the time of \cite{KLRZ}, they did not have Li's extractor for a constant number of weak sources with min-entropy $k=\polylog(n)$ \cite{Li13b}, and (ii) we use alternating extraction to extract almost all min-entropy out.

For the case of quantum rushing, we obtain slightly worse parameters, where we can tolerate a constant fraction of faulty players, and scarifice $O(t)$ honest players. Here we require the min-entropy $k$ to be sufficiently larger then $t$. We discuss at the end of the section how to relax this requirement.

\begin{theorem}[QR-secure Network Extractor] \label{thm:QRnext}  There exists a constant $\gamma \in (0,1)$ such that for every constant $c > 0$, and sufficiently large $p, t, n, k$ with $p > t / \gamma$ and $k \geq \max \{ \log^{10}n, t/\gamma \}$, there exists a 1-round $(t, p - t/\gamma, n^{-c})$ network extractor $\NetExt$ with output length $m = \Omega(k)$ in the quantum rushing case.
\end{theorem}

\begin{remark} Like in the classical setting, our network extractors can also be applied to distributed computing problems. For example, Theorem~\ref{thm:CRnext} implies that in the independent  rushing case, even with min-entropy as small as $k=\polylog(n)$, we can achieve synchronous Byzantine agreement while tolerating roughly $1/4$ fraction of faulty players. This is almost optimal since the optimal tolerance is roughly $1/3$. Similarly, Theorem~\ref{thm:QRnext} implies that in the quantum rushing case, even with min-entropy as small as $k=\polylog(n)$, we can achieve synchronous Byzantine agreement while still tolerating a constant fraction of faulty players.
\end{remark}

Our network extractor for the independent rushing case follows the same approach as our multi-source extractors. We establish OA-GE security equivalence, and use a simple security-lifting transformation to obtain OA security.

In contrast, achieving QR security is much more difficulty to handle. We first note that our simulation argument for OA-GE security equivalence breaks down in this setting, since we can no longer defer the collection of side information, which is used by $\AdvNet$ during the protocol execution. Also, even getting OA-QR security seems already challenging. To see why, consider that at some point of protocol execution, some public source $Y$ is used to extract private randomness from some honest player's source $X_i$. Suppose that $Y$ depends on some rushing information, which in turn can correlate with $X_i$ through side information. As such, it is hard to ensure that the extraction works.

To address the issue, we develop a security lifting technique from IR to QR security. Very informally, the idea is to break the correlation by guessing, which reduces QR attacks to IR ones, but at the cost of $2^{\mathsf{rushing-length}}$ blow-up in error (along with other limitations). We thus carefully design the protocol to restrict the (effective) length of rushing attacks, and this is the reason that we require that $k>t/\gamma$ in Theorem~\ref{thm:QRnext}. However, we also sketch an approach for the case of $k<t$ for quantum rushing setting towards the end of Section~\ref{sec:network}.


\subsection{Open Problems and Future Work} \label{sec:future}
Our results leave several open problems. First, although our GE model is quite general, it may not be the most general model. Thus, one can ask whether there is a more general model that also allows the construction of quantum multi-source extractors in the presence of even entangled quantum side information. Second, in our network extractor, we deal with quantum rushing using a naive ``guessing" technique, which results in a $2^{\mathsf{rushing-length}}$ blow-up in error. Is there a better way to tackle this problem?

For future work, it would be nice to see if our techniques can be applied to other related problems with quantum side information, such as privacy amplification with an active adversary.

\subsection*{Organization}
The rest of the paper is organized as follows: in Section~\ref{sec:prelim}, we summarize necessary background knowledge on quantum information, classical and quantum single/multi/block-source extractors. 
We then formally introduce our GE model in Section~\ref{sec:model} with detailed discussions. 
The strong OA-GE security equivalence is established in Section~\ref{sec:OA_to_GEA}.   Three arguments for obtaining strong OA-security are demonstrated in Section~\ref{sec:obtain_OA}. 
A new construction of a three-source extractor is illustrated in Section~\ref{sec:three_Ext} with its application to privacy amplification in Section~\ref{sec:PA}. We conclude with results about network extractors in Section~\ref{sec:network}.


%
%
%
%
%
%
%


\section{Preliminary} \label{sec:prelim}
We assume familiarity with the standard concepts from quantum
information and summarize our notation and useful facts in Section~\ref{sec:prelim_qi}. 
We also summarize necessary background about classical independent source extractors in Section~\ref{sec:prelim_indep} and quantum extractors in Section~\ref{sec:prelim_q_ext}.

\subsection{Quantum Information}  \label{sec:prelim_qi}
\begin{trivlist}

\item \textbf{Quantum States.} We only consider finite dimensional Hilbert spaces as quantum states in infinite dimensions can be truncated to be within a finite dimensional space with an arbitrarily small error. The state space $\A$ of  $m$-qubit is the complex Euclidean space $\complex^{2^m}$. An $m$-qubit quantum state is represented by a density operator $\rho$, i.e., a positive semidefinite operator over $\A$ with trace $1$. The set of all quantum states in $\A$ is denoted by $\density{\A}$. 

Let $\lin{\A}$ denote the set of all linear operators on space $\A$. The Hilbert-Schmidt inner product on $\lin{\A}$ is defined by $\ip{X}{Y}=\tr (X^*Y)$,  for all $X,Y \in \lin{\A}$, where $X^*$ is the adjoint conjugate of $X$. 
Let $\I_\X$ denote the identity operator over $\X$, which might be omitted
from the subscript if it is clear in the context. 

For a multi-partite state, e.g. $\rho_{ABC} \in \density{\A \ot \B \ot \C}$, its reduced state on some subsystem(s) is represented by the same state with the corresponding subscript(s). For example, the reduced state on $\A$ system of $\rho_{ABC}$
is $\rho_A=\tr_{\B\C}(\rho_{ABC})$, and $\rho_{AB}=\tr_{\C}(\rho_{ABC})$.  When all subscript letters are omitted, the notation
represents the original state (e.g., $\rho=\rho_{ABE}$).

A \emph{classical-quantum-}, or cq-state $\rho \in \density{\A \ot \B}$ indicates that the $\A$ subsystem is classical and $\B$ is quantum. Likewise for ccq-, etc., states. We use \emph{lower case}     letters to denote specific values assignment to the classical part of a state. For example, any cq-state $\rho_{AB}=\sum_a p_a \ketbra{a} \ot \rho^a_B$ in which $p_a=\Prob[A=a]$ and $\rho^a_B$ is a normalized state. 

\item \textbf{Distance Measures.} For any $X \in \lin{\A}$ with singular values $\sigma_1,\cdots, \sigma_d$, where $d=\dim(\A)$, the trace norm of $\A$ is $\trnorm{X}=\sum_{i=1}^d \sigma_i$.
The \emph{trace distance} between two quantum states $\rho_0$ and $\rho_1$ is defined to be \[\trdist{\rho_0- \rho_1} \defeq \frac{1}{2}\trnorm{\rho_0-\rho_1}.\] 
When $\rho_0$ and $\rho_1$ are \emph{classical} states,  the trace distance $\trdist{\rho_0-\rho_1}$ is equivalent to the \emph{statistical} distance between $\rho_0$ and $\rho_1$. 
It is also a well known fact that for two distributions $X_1, X_2$ over $\X$, let $p_x=\Prob[X_1=x]$ and $q_x=\Prob[X_2=x]$ and their statistical distance satisfies
\begin{equation} \label{eqn:fact:c_dist}
\trdist{X_1-X_2}= \frac{1}{2}\sum_{x} |p_x-q_x|=\sum_{x: p_x>q_x} (p_x-q_x).
\end{equation}
For simplicity, when both states are classical, we use $(X_1) \approx_\eps (X_2)$ to denote $\trdist{X_1-X_2}\leq \eps$. 

Moreover, the trace distance admits the following two simple facts. 

\begin{fact} \label{fact:trdist:add_prod}
For any state $\rho_1, \rho_2 \in \density{\A}$ and $\sigma \in \density{\B}$, we have 
\[
  \trdist{\rho_1-\rho_2}=\trdist{\rho_1\ot \sigma-\rho_2\ot \sigma}.
\]
\end{fact}

\begin{fact} \label{fact:trdist:c_decomp}
Let $\rho,\sigma \in \density{\A \ot \B}$ be any two cq-states where $\A$ is the classical part. Moreover, $\rho=\sum_a p_a \ketbra{a} \ot \rho^a_B$ and $\sigma=\sum_a q_a \ketbra{a} \ot \sigma^a_B$. Then we have
\[
 \trdist{\rho-\sigma}=\sum_a \trdist{p_a\rho^a_B -q_a \sigma^a_B}.
\]
\end{fact}

%

\item \textbf{The XOR-Lemma}. Vazirani's XOR-Lemma~\cite{Vaz87} relates the non-uniformity of a distribution to the non-uniformity of the XOR of certain bit positions. For our application, we need the following more general XOR-Lemma~\cite{KK12} which takes into account \emph{quantum} side information. 

\begin{lemma}[\cite{KK12}, Lemma 2.6]  \label{lem:KK_XOR}
Let $\rho_{ZE}$ be an arbitrary cq-state where $Z \in \01^m$ and the register $E$ is of dimension $2^d$.  Then we have 
\[
   \trdist{\rho_{ZE} -\uniform{m} \ot \rho_E}^2 \leq 2^{\min (d,m)} \sum_{\emptyset \neq S \subseteq \01^m} \trdist{ \rho_{Z_{\oplus S}E}-\uniform{1} \ot \rho_E}^2,
\]
where $Z_{\oplus S}= \bigoplus_{i \in S} z_i$. 
\end{lemma}

\item \textbf{Quantum Operations.} Let $\X$ and $\Y$ be state spaces.
A {\em super-operator} from $\X$ to $\Y$ is a linear map
\begin{equation}
  \Psi : \lin{\X} \rightarrow \lin{\Y}.
\end{equation}
Physically realizable \emph{quantum operations} are represented by \emph{admissible} super-operators,
which are completely positive and trace-preserving.
Thus any classical operation (such as extractors) can be viewed as an admissible super-operator. 
We shall use this abstraction in our analysis and make use of the following observation.

\begin{fact}[Monotonicity of trace distances] \label{prelim:fact:monotone_trace}
For any admissible super-operator $\Psi: \lin{\X}\rightarrow \lin{\Y}$ and $\rho_0,\rho_1\in \density{\X}$, we have
\begin{equation}
 \trdist{\Psi(\rho_0)-\Psi(\rho_1)}\leq \trdist{\rho_0-\rho_1}.
\end{equation}
\end{fact}

Moreover, we adopt the convention that when $\Psi$ is applied on a part of the quantum system,  we omit the identity operation applied on the rest part of the system when it is clear from the context.

Let $\{ \ket{i} : 1\le i\le \dim(\X)\}$ be the computational basis for $\X$.
An $\X$-controlled quantum operation on $\Y$ is an admissible operation 
 $\Phi: \lin{\X \ot \Y} \rightarrow \lin{\X \ot \Y'} $ such that 
for some admissible $\Phi_i :\lin{\Y} \rightarrow \lin{\Y'}$, $1\le i\le \dim(\X)$,
\begin{equation}
  \Phi =\sum_{1\le i\le\dim(\X)} \langle i|\cdot |i\rangle|i\rangle\langle i|\otimes \Phi_i(\cdot).
\end{equation}


\end{trivlist}

\subsection{Independent Source Extractors} \label{sec:prelim_indep}

\begin{trivlist}

\item \textbf{Random variables and min-entropy sources.} We use \emph{upper case} letters to denote random variables which take values over $\01^n$ for some $n$. 
Usually, the \emph{calligraphy} letter denotes the set of all possible values that this random variable can take. 
\emph{Lower case} letters are used to denote specific values of the random variables, such as random variable $A=a$ for some value $a\in \A$. This is consistent with our notation of quantum states when reduced to the classical cases. Moreover, if the whole system is classical, we will treat it as a classical random variable only and thus omit notation such as $\rho$ for clarity. 
For convenience, we denote the set $\{1, \cdots, t\}$ by $[t]$ for any positive integer $t$.

\begin{definition}[Min-entropy] \label{def:min-entropy}
The \emph{min-entropy} of a random variable $X$ is given by 
\[
   \Hmin(X)=\min_{x \in \X} \log_2 (1/\Prob[X=x]).
\]
For $X\in \01^n$, we call $X$ an $(n, \Hmin(X))$-source (or $\Hmin(X)$-source) with \emph{entropy rate} $\Hmin(X)/n$. 
\end{definition}

One useful property about min-entropy is the following lemma: 

\begin{lemma}[\cite{MW97}] \label{lem:condition} 
Let $X$ and $Y$ be random variables and let ${\cal Y}$ denote the
range of $Y$. Then for all $\epsilon>0$
\[\Pr_Y \left [ \Hmin(X|Y=y) \geq \Hmin(X)-\log|\Y|-\log \left( \frac{1}{\eps} \right )\right ] \geq 1-\epsilon\]
\end{lemma}

\begin{definition}[Block-source] \label{def:block_source}
A distribution $X=X^1 \circ X^2 \circ \cdots \circ X^C$ is called a $(k_1, k_2, \cdots, k_C)$ block-source if for any $i\in [C]$, we have that for any $x_1 \in \X^1, \cdots, x_{i-1} \in \X^{i-1}$, $\Hmin(X^i|X^1=x_1, \cdots, X^{i-1}=x_{i-1})\geq k_i$, i.e., each block contains high min-entropy even conditioned on every possible value of previous blocks. 
If $k_1=k_2=\cdots=k_C$, then $X$ is called a $k$-block-source. 
\end{definition}

\item \textbf{Two-source and Independent Sources Extractors.} Here we review two (or multi) independent source extractors, which turn two (or multi) independent min-entropy sources to a close-to-uniform distribution. 
At this moment, we don't consider the existence of adversaries and only look at the marginal distribution of the output of the extractors. Therefore, we refer this as the \emph{marginal} security throughout this paper. 

Let $\uniform{A}$ denote the completely mixed state on a space $\A$, i.e., $\uniform{A}=\frac{1}{\dim(\A)} \I_\A$. Let $\uniform{n}$ denote $\uniform{A}$ when $\A=\{0,1\}^n$. Moreover, for any given subset $S \subseteq \{1,\cdots, t\}$ and let $X_{ S}=\circ_{i \in S}X_i$.
 
\begin{definition}[Independent Source Extractor] \label{def:classical_IExt}
A function $\IExt :  (\01^n)^t \rightarrow \01^m$ is a $(t,n,k,m,\eps)$ 
 independent source extractor that uses $t$ sources and outputs $m$ bits with error $\eps$, if for
any $t$ independent $(n, k)$ sources $X_1, X_2, \cdots , X_t$, we have
\[
   \trdist{\IExt(X_1, X_2, \cdots, X_t)-\uniform{m}} \leq \eps.
\] 
For any subset $S \subseteq[t]$, $\IExt$ is called $S$-strong if 
\[
  \trdist{\IExt(X_1, X_2, \cdots, X_t)X_{ S}-\uniform{m}\ot X_{ S}} \leq \eps.
\]
\end{definition}

\begin{definition}[Two-source Extractor] \label{def:classical_TExt}
A function $\TExt :  \01^{n_1} \times \01^{n_2} \rightarrow \01^m$ is a $(n_1,k_1,n_2,k_2, m, \eps)$ two-source extractor 
if for
any independent $(n_1,k_1)$ source $X_1$ and $(n_2,k_2)$ source $X_2$,  we have
\[
   \trdist{\TExt(X_1, X_2)-\uniform{m}} \leq \eps.
\] 
Moreover, $\TExt$ is called $X_1$-strong, (and similarly for $X_2$-strong), if 
\[
  \trdist{\TExt(X_1, X_2)X_1-\uniform{m}\ot X_1} \leq \eps.
\]
\end{definition}

We say that an extractor is explicit if it can be computed in polynomial time.



%
%
%
%
%

\end{trivlist}

\subsection{Quantum Seeded Extractors} \label{sec:prelim_q_ext}

\begin{trivlist}
 \item{\textbf{Quantum Conditional Min-entropy}}. 
In the regime of quantum extractors, it is necessary to consider the existence of adversaries who are furthermore given quantum computational power.
In the seeded extractor setting, it suffices to model the adversary as \emph{quantum side information} which is stored in the system $\E$ as follows. 
For a cq state $\rho_{XE} \in \density{\X\ot \E}$, the amount of \emph{extractable} randomness (from $X$ against $E$) is characterized by its conditional min-entropy.
\begin{definition}[Conditional Min-entropy] \label{prelim:def:min-entropy}
Let $\rho_{XE} \in \density{\X\ot \E}$. The \emph{min-entropy} of $X$ conditioned on $E$ is defined as
  \begin{equation*}
    \Hmin({X|E})_\rho \defeq \max \{\lambda \geq 0 :  \exists \sigma_E \in \density{\E}, \mathrm{s.t.}\,\, 2^{-\lambda} \I_X \ot \sigma_E \geq \rho_{XE}\}.
  \end{equation*}
\end{definition}
This definition has a simple operational interpretation shown in~\cite{KRS09} that 
\[
  \Hmin(X|E)_\rho = - \log(p_{\mathrm{guess}} (X|E)_\rho),
\]
where $p_{\mathrm{guess}} (X|E)_\rho$ is the maximum probability of guessing $X$ by making arbitrary measurements on $E$ system. 
Similar to the classical min-entropy, the quantum conditional entropy also satisfies the following property. 

\begin{lemma}[\cite{KT08}] \label{lem:q_condition}
Given any ccq state $\rho_{XWE}$ in which $W \leftrightarrow X \leftrightarrow E$~\footnote{Namely, we have $\rho_{XWE}=\sum_{x,w} \Prob[X=x, W=w] \ketbra{x,w} \ot \rho^x_E$.}, we have 
\[\Pr_{w \sim W} \left [ \Hmin(X|W=w, E) \geq \Hmin(X)-\log\dim(\W)-\log (1/\eps ) \right] \geq 1-\epsilon\]
\end{lemma}

We can also consider the \emph{smooth} min-entropy that consists in maximizing the
min-entropy over all sub-normalized states that are $\eps$-close to the actual
state $\rho_{XE}$ in trace distance. Note that allowing an extra error $\eps$ can increase the min-entropy of a certain state very significantly.
\begin{definition}[smooth min-entropy] \label{prelim:def:smooth-min-entropy}
  Let $\eps \geq 0$ and $\rho_{XE} \in \density{\X \ot \E}$, then the
  \emph{$\eps$-smooth min-entropy} of $X$ conditioned on $E$ is defined as
  \begin{equation*}
    \Hmin^\eps(X|E)_\rho \defeq \max_{ \trdist{\sigma_{XE}
      -\rho_{XE}} \leq \eps} \Hmin(X|E)_\sigma,
  \end{equation*}
\end{definition}
Similarly, we call $\rho_{XE}$ a (smooth) $(n,k)$-source (or $k$-source) if $X \in \{0,1\}^n$ and $ \Hmin({X|E})_\rho\geq k$. ($ \Hmin^\eps({X|E})_\rho\geq k$)

\begin{definition}[Quantum block-source] \label{def:quantum_block_source}
Let $\rho_{X^1 \cdots  X^C E} \in \density{X^1 \ot \cdots \ot X^C \ot \E}$ is called a $(k_1, k_2, \cdots, k_C)$ quantum block-source if for any $i\in [C]$, we have that for any $x_1 \in \X^1, \cdots, x_{i-1} \in \X^{i-1}$, $\Hmin(X^i|X^1=x_1, \cdots, X^{i-1}=x_{i-1}, E)\geq k_i$, i.e., each block contains high min-entropy even conditioned on every possible value of previous blocks and the quantum system $E$. 
If $k_1=k_2=\cdots=k_C$, then $X$ is called a quantum $k$-block-source. 
\end{definition}

In the following survey a few useful lemmata about conditional quantum min-entropy. 

\begin{lemma}[Data Processing] \label{lem:data_processing}
Let $\rho_{XE}$ be a cq state, $\Phi: \lin{\E} \rightarrow \lin{\E'}$ be any admissible operation. Moreover, let $\sigma_{XE'}=\Phi(\rho_{XE})$. Then we have
\[
  \Hmin(X|E)_\rho \leq \Hmin(X|E')_\sigma. 
\]
\end{lemma}

\begin{lemma}[Chain-rule] \label{lem:chain_rule}
Let $\eps>0, \eps'>0, \eps''>0$ and $\rho \in \density{\A \ot \B \ot \C}$, then we have the following chain rule:
\[
  \Hmin^{\eps+2\eps'+\eps''}(AB|C)_\rho \geq \Hmin^{\eps'}(A|BC)_\rho + \Hmin^{\eps''}(B|C)_\rho-\log\frac{2}{\eps^2}. 
\]
\end{lemma}

\begin{lemma}[\cite{NS06}]  \label{lem:NS}
Let $X$ be an $n$-bit random variable with min-entropy $k$, and suppose Alice wishes to convey $X$ to Bob over a one-way quantum communication channel using $b$ qubits with shared entanglement. Let $Y$ be Bob's guess for $X$. Then we have $\Prob[Y=X] \leq 2^{-(k-2b)}.$
\end{lemma}

\item{\textbf{Quantum Seeded Extractors.}}
Here we review quantum seeded
randomness extractors, which turn a min-entropy source to a quantum-secure uniform output, with the help of a short seed.  Since now the system involves a quantum adversary, we refer this as the \emph{quantum} security.  

\begin{definition}[Quantum Strong Seeded Extractor] \label{prelim:def:q_strong_extractor}
  \label{def:extractor_q_proof}
  A function $\Ext: \{0,1\}^n \times \{0,1\}^d \to \{0,1\}^m$ is a
  \emph{quantum-secure} (or simply quantum) 
  \emph{$(k,\eps)$-strong seeded (randomness) extractor}, if for all cq
  states $\rho_{XE}$ with $\Hmin(X|E) \geq k$,
  and for a uniform seed $Y$ independent of $\rho_{XE}$, we
  have
\begin{equation} \label{eqn:extractor_q_proof} \trdist{
    \rho_{\Ext(X,Y)YE} - \uniform{m} \ot \rho_Y \ot \rho_E}
  \leq \eps. \end{equation}
 \end{definition}
We state the following quantum strong seeded extractor in~\cite{DPVR12} that will be useful for us to instantiate our multi-source and network extractors. 

\begin{theorem}[\cite{DPVR12}, Corollary 5.4] \label{thm:Ext1} For every $n, k \in \N$ and $\eps > 0$ with $k \geq 4 \log(1/\eps) + O(1)$, there exists a quantum $(k,\eps)$-strong seeded extractor $\Ext: \zo^n \times \zo^d \rightarrow \zo^m$ with $m = k - 4 \log(1/\eps) - O(1)$ and $d = O(\log^2(n/\eps) \log m)$.
\end{theorem}

\end{trivlist}

\section{Adversarial Model in Multi-source Extraction} \label{sec:model}
%
%
In this section, we formally define the adversarial model in the context of randomness extraction form multi-independent sources, in which the adversary could have access to \emph{quantum} resources.
As we discussed in the introduction,  the multi-source setting, contrasting to the single-source setting, offers a completely new aspect of the problem: the adversaries could potentially share \emph{entanglement} prior to tampering with the sources and 
the obtained leakage could be stored in entanglement. 
A preliminary discussion of such adversarial models can be found in~\cite{KK12}, which correspond to the \emph{independent adversary (IA)} model and the \emph{bounded storage (BS)} model in our later discussion. 
In the following, we identify a more general (powerful) adversarial model, which we called the \emph{general entangled (GE)} model that includes the IA and BS model as special cases (yet we show that randomness extraction is possible provided there are sufficient min-entropy in the sources). 
At the same time, we identify a much less powerful adversarial model, called the \emph{one-sided adversary (OA)} model, which is  a common special case of the GE and IA  model, and is a weaker model than the BS model with incomparable entropy measure.  

Given any extractor, if its output is uniform against any adversary in the GE model, then we call that extractor is \emph{secure} against the GE model (or \emph{GE-secure} for short). 
Similarly for the IA, BS and OA models.
One of the main results in this paper (which is presented in Section~\ref{sec:OA_to_GEA}) is to establish a surprising equivalence between the GE and the OA model in the following sense: any strong OA-secure extractor is \emph{automatically} a strong GE-secure extractor without any loss of parameters. 


\subsection*{General Entangled Adversarial Model}  

\paragraph{Generating Side Information.}First recall, in the classical independent source extraction setting, one is given $t$ \emph{independent} random variables $X_1, \cdots, X_t$ such that $X_i \in \01^n$ for $i\in [t]$. 
    Assume there are $t$ non-communicating parties, each of which receives one classical random variable $X_i$ for $i \in [t]$ each with 
We imagine an adversary who will generate the side information of the source $X_1, \cdots, X_t$ via the following procedure. 
 
   The adversary initially prepares a quantum state $\rho_0$ on registers $A_1, \cdots A_t$  (which is independent of $X_i$s) and sends each register $A_i$ to the $i$th party who holds $X_i$. 
Note that there could be arbitrary \emph{entanglement} among $A_1, \cdots, A_t$. 
Depending on the source $X_i$, the $i$th party then applies an arbitrary admissible \emph{leaking} operation
from $X_i$ to its own quantum register. 
Precisely, this leaking operation can be formulated as a  $X_i$-controlled operation denoted by $\Phi_i(\cdot): \lin{\X_i \ot \A_i} \rightarrow \lin{\X_i \ot \E_i}$. It is easy to see that $\Phi_i$ commutes with  $\Phi_j$ for any different $i,j$. 
Finally, the adversary collects all $E_i$s as the side information of the sources $X_1, \cdots, X_t$. 
Formally, 
the generated quantum side information together with the source $\rho_{X_1\cdots X_t\Adv}$ ($\Adv=E_1, \cdots, E_t$) is given by
  $\rho_{X_1\cdots X_t\Adv}=  \Phi_1 \ot \cdots \ot \Phi_t (X_1\cdots X_t \ot \rho_0)$. 


The above procedure is a generalization of the one discussed in~\cite{KK12} to the multi-source case. 
However,  what makes our model significantly different from theirs is the following \emph{crucial} observation on how to measure the quality (or entropy) of the sources, 
which allows us to directly deal with the more powerful GE model, rather than to work with much restricted IA or BS model like in~\cite{KK12}.

\paragraph{Entropy Measure and Properties.}For any $i \in [t]$, the quality of source $X_i$ is measured by the (in)ability of the adversary to guess it given the quantum side information that contains \emph{only} the leakage from $X_i$. 
Formally, this measure is captured by 
\begin{equation} \label{def:GE:entropy}
  k_i \defeq \Hmin(X_i|\Adv_i)_{\rho_i} \text{ and } \rho_i = \Phi_i (X_1\cdots X_t \ot \rho_0), \forall i \in [t]. 
\end{equation}
Let $X_{-i}$ denote all $X_j$ except $j=i$ and $A_{-i}$ denote all $A_j$ except $j=i$.  In this notation,   $\rho_i=X_{-i} \ot \Phi_i(X_i \ot \rho_0)$ and $\Adv_i=(A_{-i}, E_i)$. Thus, $\rho_{i_{X_i\Adv_i}}= \Phi_i(X_i \ot \rho_0)$.
Intuitively, $k_i$ measures the min-entropy of $X_i$ conditioned on $\Adv$ at an \emph{imaginary} step after the leaking operation $\Phi_i$ is performed, but before all the other leaking operations are performed. 
Such entropy measure enjoys the following natural expected properties. 

\vspace{1mm} \noindent \textbf{(1) Non-decreasing property}. Since $X_1, \cdots, X_t$ are independent and $\Phi_1, \cdots, \Phi_t$ commute with each other,
applying other leaking operations only increases the min-entropy of $X_i$ conditioned on the quantum side information, which is captured by the following proposition:

\begin{proposition} \label{prop:one_entropy}
For any $i\in [t]$, $k_i=\Hmin(X_i|\Adv_i)_{\rho_i} \leq \Hmin(X_i|X_{-i}\Adv)_\rho \leq \Hmin(X_i|\Adv)_\rho$. 
\end{proposition}

\begin{proof}
This is almost by definition and the data processing lemma of min-entropy (Lemma~\ref{lem:data_processing}). First note that $X_{-i}$ is independent of $X_i$ and can be locally generated on the $\Adv$ side. Thus, there is an admissible operation converting $\rho_i$ to $\rho$ only applying on $\Adv_i$ side by first generating $X_{-i}$ and then applying the leaking operation on them. This gives the first inequality. The second inequality follows because tracing out $X_{-i}$ system is an admissible quantum operation. 
\end{proof}

\vspace{1mm} \noindent \textbf{(2) Additivity property}. Our measure of entropy is genuine and can be added in the following sense: the smooth min-entropy of $X_1, \cdots, X_t$ conditioned on $\Adv$ in the final quantum side information $\rho$ is almost $\sum_{i=1}^t k_i$, the sum of entropies measured for each source $X_i, i \in [t]$. 

\begin{proposition} \label{prop:all_entropy}
For any $\eps>0$ and any $S \subseteq [t]$ (let $|S|=s$),  $\Hmin^{(s-1)\eps}(X_{S}|\Adv)_\rho \geq \sum_{i\in S} k_i -(s-1)\log (2/\eps^2).$
\end{proposition}

\begin{proof}
This proposition follows from a sequential application of the chain-rule for smooth min-entropy in Lemma~\ref{lem:chain_rule}. Without loss of generality, let us assume $|S|=s$ and $S=\{1, \cdots, s\}$. By Proposition~\ref{prop:one_entropy}, we have 
\[
  \Hmin(X_i|X_{-i}\Adv)_\rho \geq k_i, \forall i \in [t]. 
\]
Thus, we have $\Hmin(X_i|X_1\cdots X_{i-1}\Adv)_\rho \geq k_i, \forall i \in [t]$. The following comes from a sequential use of Lemma~\ref{lem:chain_rule}, 
\begin{eqnarray*}
  \Hmin^\eps (X_2X_1|\Adv)_\rho & \geq&  \Hmin(X_2|X_1\Adv)_\rho +\Hmin(X_1|\Adv) -\log \frac{2}{\eps^2} \\
  \Hmin^{2\eps} (X_3X_2X_1|\Adv)_\rho & \geq&  \Hmin(X_3|X_2X_1\Adv)_\rho +\Hmin^\eps(X_2X_1|\Adv) -\log \frac{2}{\eps^2} \\
  \cdots & & \cdots \\
    \Hmin^{(s-1)\eps} (X_sX_{s-1} \cdots X_1|\Adv)_\rho & \geq&  \Hmin(X_s|X_{s-1}\cdots X_1\Adv)_\rho +\Hmin^{(s-2)\eps}(X_{s-1}\cdots X_1|\Adv) -\log \frac{2}{\eps^2} 
\end{eqnarray*}
Therefore, by rearranging all the above inequalities, we have 
\[
  \Hmin^{(s-1)\eps}(X_S|\Adv) \geq \sum_{i \in S} k_i - (s-1)\log \frac{2}{\eps^2}.
\]
\end{proof}

\noindent \textbf{Remark.} Comparing to our model, the entropy measure in the model of~\cite{KK12} is only on the final quantum side information $\rho$ when all leaking operations have been performed (e.g., $\Hmin(X_i|\Adv)_\rho$), whereas our entropy measure is on the state $\rho_i$ for each $X_i$. 
By Proposition~\ref{prop:one_entropy}, the entropy measure on the final quantum side information could potentially be much higher than $k_i$ due to possible interference from other leaking 
operations,  which, hence, fails to characterize the right amount of entropy from each $X_i$. 
This is exactly our motivation to study our notion of entropy $k_i$ that is measured before any interference happens. As shown in Proposition~\ref{prop:all_entropy}, the total min-entropy of the source is lower bounded by the sum of $k_i$s. Thus, there is no double counting of entropy with our measure. 

\subsection*{Justification of GE model}
In this section, we further justify our proposed GE model by demonstrating a few nice properties about the model as follows.


First, we claim that our GE model is a strict generalization of the no-side-information case. Recall the no-side-information case, the sources are independent $X_1, \cdots, X_t \in \01^n$ each with min-entropy $k_i=\Hmin(X_i), \forall i \in [t]$. In the framework of GE model, this implies trivial space $A_1, \cdots, A_t, E_1, \cdots, E_t$ and trivial leaking operations $\Phi_i(\cdot), \forall i\in [t]$. 
By the entropy measure of GE model, we have the entropy for source $X_i$ is $k_i'=\Hmin(X_i|\Adv)_{\rho_i}=\Hmin(X_i|\Adv)_\rho=\Hmin(X_i)=k_i$. Namely, the GE-entropy exactly matches the original entropy measure in the no-side-information case. Thus, the GE model is a strict generalization.


Second, the GE-entropy measure, similar to the classical min-entropy measure, captures the amount of uniform randomness that can be extracted from the source in the presence of GE-side information.
We support the above statement with the following two points: 1) all of the GE-entropy can be extracted and 2) there exists sources with certain GE-side information, in which the GE-entropy also upper bounds the amount of uniform randomness that can be extracted. 
The first point is validated by the existence of strong GE-secure multi-source extractors in Section~\ref{sec:obtain_OA}.\footnote{Precisely, to extract all the GE-entropy, one first notice that there exist $t$-source GE-secure multi-source extractors $\QMExt$ that are strong to $t-1$ sources, which extracts the GE-entropy from one source. One can then apply a strong quantum-proof seeded extractor to the $t-1$ sources by using the output of $\QMExt$ as the seed. In this way, one can further extract all the GE-entropy within the $t-1$ sources guaranteed by Proposition~\ref{prop:all_entropy}.} 
The second point is due to the fact that classical independent flat $k$-sources\footnote{A distribution $\X$ over $\01^n$ is called a flat $k$-source if the support of $\X$ is $2^k$ and for each $x$ in its support, the probability $\Prob[X=x]=2^{-k}$.} are just special cases of GE-sources with GE-entropy also being $k$ and no side information. It is easy to see that in this case $k$ upper bounds the amount of uniform randomness that can be extracted from each source. 

%

Finally, we argue that the one-round side-information-generating process in our GE model (essentially from~\cite{KK12}) is appealing due to both theoretical and practical reasons. 
In the theoretical aspect, this one-round process together with our GE-entropy measure, for the first time, allows the randomness extraction in the presence of general entangled side information. 
Moreover, if one extends this one-round process to multi-rounds, then there will necessarily be interference between different sources. It is again not a prior clear whether the randomness extraction is possible. 
On the other side, the one-round process also characterizes several side-information generating scenarios in practice. For example, if the side information is generated simultaneously at distant parties each holding one of the sources, then it can effectively be characterized by the one-round process. 




\subsection*{Special cases: IA, BS and OA model}

Now we are ready to introduce special cases of the GE model when imposing various restrictions on the adversary and discuss the relation between the measure of entropies within each model. 

\begin{trivlist}

\item \emph{Independent Adversarial (IA) Model} imposes one additional constraint on the GE model: that is the initial state $\rho_0$ is a product state over $A_1, \cdots, A_t$, i.e., $\rho_{A_1, \cdots, A_t}= \rho_{A_1} \ot \cdots \ot \rho_{A_t}$. Thus, by definition, the side information state $\rho_{E_1, \cdots E_t}$ is also a product state.\footnote{This definition is slightly different from the one (called \emph{quantum knowledge}) of~\cite{KK12} which only requires the side information is a product state. However, it is a simple exercise to see that any product side information can be produced by a product initial state. Thus, two definitions are equivalent.}
In this case,  the entropy of each $X_i$ is measured by $k'_i= \Hmin(X_i|E_i)_\rho$, for $i \in [t]$, which matches exactly the definition of our more general entropy measure $k_i$ in (\ref{def:GE:entropy}) when reduced to the IA model. (i.e., $\Hmin(X_i|E_i)_\rho=\Hmin(X_i|E_iA_{-i})_{\rho_i}$) 


\item \emph{Bounded Storage (BS) Model} imposes a different constraint on the GE model: that is to
bound the dimension of each register $E_i$ by $2^{b_i}, \forall i \in [t]$ that are collected at the last step. 
In this case,  the quality of the source $X_i$ is measured by its marginal min-entropy $k'_i=\Hmin(X_i)$ and the size bound $b_i$ on each register $E_i$. 
By Lemma~\ref{lem:NS}, we can relate $k_i$ in (\ref{def:GE:entropy}) with $k'_i$ and $b_i$ by $k_i\geq k'_i-2b_i$, in which the factor two is due to the possibility of super-dense coding. 


\item \emph{One-sided Adversary (OA) Model} is the weakest model in which the adversary is restricted to collect leakage information from only one source $X_i$ but has the freedom to choose which $i\in [t]$. Let $i^*$ be the adversary's choice. Namely, only $A_{i^*}$ is nonempty among all $A_i$s.  
That is, $\Adv=\Adv_{i^*}=E_{i^*}$ and other $\Adv_j=\emptyset$ for $j\neq i^*$. 
The only non-trivial leaking operation is $\Phi_{i^*}$. 
It is easy to see that $\rho=\rho_{i^*}$ but different from $\rho_i, \forall i  \neq i^*$, which equals $\rho_0$. 
According to (\ref{def:GE:entropy}), the entropy of $X_{i^*}$ is measured by $k_{i^*}=\Hmin(X_{i^*}|\Adv_{i^*})_{\rho_{i^*}}=\Hmin(X_{i^*}|E_{i^*})_\rho$ and the entropy of other $X_i$s ($i\neq i^*$) is measured by $k_i=\Hmin(X_i|\Adv_i)_{\rho_i}=\Hmin(X_i)_{\rho_0}=\Hmin(X_i)$. 
It is easy to see that $\Hmin(X_i)=\Hmin(X_i|\Adv)_\rho, \forall i \neq i^*$ as $X_i$ is independent of $\rho$. 
By definition, the OA model is also a special case of the IA model. In terms of the adversary's power, it is also a special case of the BS model.
However, the measure of the quality of sources in the BS model is incomparable from the OA model.

\end{trivlist}

\subsection*{Quantum Multi-source Extractor}

Consider any $t$ independent sources $X_1, \cdots, X_t \in \01^n$ with the quantum side information generated in the GE model. For simplicity, we usually denote $(k_1, \cdots, k_t)$ from (\ref{def:GE:entropy}) by some $k$ such that $k\leq k_i, \forall i \in [t]$ unless explicitly specified 
 and denote all such sources together with the generated quantum side information by \emph{GE-$(t,n,k)$ sources}. Similarly for the IA, BS and OA model. Note that any IA, BS, or OA source is automatically a GE source by definition. Thus, if any extractor is GE-secure, it is automatically secure against IA, BS and OA. 
In the following, we only define extractors for GE and OA models for simplicity. 

\begin{definition}[Quantum Multi-source Extractor] \label{def:QMExt}

Any function $\QMExt :  (\01^n)^t \rightarrow \01^m$ is a $(t,n,k,m,\eps)$ $MM$-secure multi-source extractor 
if for any $MM$-$(t,n,k)$ source, the function $\QMExt$ outputs $m$ bits that are close to uniform with error $\eps$ against the side information in the $MM$ model, where $MM \in \{\mathrm{GE, OA}\}$. 
Namely, with $\Adv= (E_1, \cdots, E_t)$, 
\[
  \trdist{\rho_{\QMExt(X_1, X_2, \cdots, X_t)\Adv}-\uniform{m}\ot \rho_{\Adv}} \leq \eps.
\] 
Moreover, for any given subset $S \subseteq \{1,\cdots, t\}$ and let $X_{S}=\circ_{i \in S}X_i$
, $\QMExt$ is called $S$-strong if, 
\[
  \trdist{\rho_{\QMExt(X_1, X_2, \cdots, X_t)X_{S}\Adv}-\uniform{m}\ot \rho_{X_{S} \Adv}} \leq \eps.
\]
\end{definition}

For the convenience of illustrating parameters, we define formally a special case of the multi-source extractors when $t=2$, namely, two-source extractors, as follows. 

\begin{definition}[Quantum Two-source Extractor] \label{def:QTExt}
Any function $\QTExt :  \01^{n_1} \times \01^{n_2} \rightarrow \01^m$ is a $(n_1, k_1, n_2, k_2, m, \eps)$ $MM$-secure two-source extractor, where $MM \in \{\mathrm{GE, OA}\}$, if the following holds. Given two independent random variables $X_1 \in \01^{n_1} , X_2 \in \01^{n_2}$, 
let the side information $\rho_{\Adv}$ ($\Adv={E_1, E_2}$) be generated in the $MM$ model and let  $(k_1, k_2) $ be the entropy measure defined in (\ref{def:GE:entropy}). 
For any such source,  we have 
\[
  \trdist{\rho_{\QTExt(X_1, X_2)\Adv}-\uniform{m}\ot \rho_{\Adv}} \leq \eps.
\] 
Moreover, then $\QTExt$ is called $X_1$-strong, (and similarly for $X_2$-strong), if,
\[
  \trdist{\rho_{\QTExt(X_1, X_2)X_1\Adv}-\uniform{m}\ot \rho_{X_1\Adv}} \leq \eps.
\]
\end{definition}

\section{Equivalence between Strong OA Security and Strong GE Security}
\label{sec:OA_to_GEA}

In Section~\ref{sec:OA:equip}, we establish the equivalence between the strong one-sided adversary security and the 
general adversary security in the following sense: any strong OA-secure extractor is \emph{automatically} a strong GE-secure extractor without any loss of parameters. The reverse direction is straightforward by definition. 


\subsection*{Equivalence by a simulation argument} \label{sec:OA:equip}
The establishment of the equivalence is due to the following \emph{simulation} argument. Given any GE-$(t,n,k)$ source,  for some index $i^*$ chosen later, 
our first observation is that at the imaginary step when $k_{i^*}$ (from (\ref{def:GE:entropy})) is defined, the source and the side information $\rho_{i^*}$ actually forms a OA-$(t,n,k)$ source. 
Thus, by applying some OA-secure extractor, one can extract randomness from this source. 
The problem here is that the imaginary OA-$(t,n,k)$ source is different from the initial GE-$(t,n,k)$ source. 
Our second observation is to make use of the strong OA-security, which requires the OA-secure extractor to be strong for all $X_j$ except $j=i^*$. 
Then because of all leaking operations commute and commute with the extractor itself, one can safely convert the OA-$(t,n,k)$ source to the initial GE-$(t,n,k)$ source after applying the OA-secure extractor, without increasing the error. 
The above intuition is formally presented in the following theorem. 

\begin{theorem} \label{thm:SOA_GE}
Any $S$-strong $(t,n,k,m,\eps)$ OA-secure multi-source extractor $\QMExt$ is also a $S$-strong $(t,n,k,m,\eps)$ GE-secure multi-source extractor if the size of $S$ is $t-1$, i.e., $|S|=t-1$. 
\end{theorem}

\begin{proof}
Our proof follows from the two-step intuition illustrated above. Given any GE-$(t,n,k)$, let $X_1, \cdots, X_t$ be the source, $\rho_0 \in \density{\A_1\ot \cdots, \ot \A_t}$ the initial state, $\Phi_i: \lin{\X_i \ot \A_i} \rightarrow \lin{\X_i \ot \E_i}$ the leaking operation for the $i$th party, and $k_i \geq k, \rho_i$ defined as (\ref{def:GE:entropy}). Without loss of generality, let us assume $S=\{1, \cdots, t-1\}$. 

\vspace{1mm} \noindent \emph{(Step 1):} we prove that the source $X_1, \cdots, X_t$ and $\rho_t$ forms a specific OA-$(t,n,k)$ source, by describing a OA procedure to generate the side information $\rho_t$. For clarity, we denote the notations in the OA model with an extra prime.
Let the OA adversary choose to collect only the leakage from $X_t$. Choose $A'_t=(A_1, \cdots, A_t)$ and $E'_t=(A_1, \cdots, A_{t-1}, E_t)$ and $\Phi'_t(\cdot)=\Phi_t$. Thus, it is easy to see
that the side information $\rho'$ collected in the OA procedure is exactly $\rho_t$. 
Moreover, by definition, we have $k'_t=k_t \geq k $ and $k'_i\geq k_i\geq k$ for $i \in [t-1]$. Thus, it is a OA-$(t,n,k)$ source. 
By definition, for any $S$-strong $(t,n,k, m, \eps)$ OA-secure multi-source extractor $\QMExt$, we have, for $\Adv_t=(A_1, \cdots, A_{t-1}, E_t)$, 
\begin{equation} \label{eqn:OAE:step1}
  \trdist{\rho_{t_{\QMExt(X_1, \cdots, X_t)X_1\cdots X_{t-1}\Adv_t}}- \uniform{m} \ot \rho_{t_{X_1\cdots X_{t-1}\Adv_t}}}\leq \eps. 
\end{equation}

\vspace{1mm} \noindent \emph{(Step 2):} now we can apply $\Phi_i: i \in [t-1]$ to both states in (\ref{eqn:OAE:step1}). Since all $\Phi_i: i \in[t]$ commute, we have, for $\Adv=(E_1, \cdots, E_t)$,
\[ 
 \Phi_1 \ot \cdots \ot \Phi_{t-1} (\rho_t)=\rho_{X_1\cdots X_t \Adv}.
\]
Thus, by Fact~\ref{prelim:fact:monotone_trace}, we have 
\[
 \trdist{\rho_{\QMExt(X_1, \cdots, X_t)X_1\cdots X_{t-1}\Adv}- \uniform{m} \ot \rho_{X_1\cdots X_{t-1}\Adv}}\leq \eps,
\]
which, by definition, completes the proof. 
\end{proof}

%
%

\section{Obtaining Strong OA Security from Marginal Security} 
\label{sec:obtain_OA}
In this section, we demonstrate three different techniques to obtain strong OA security from marginal security, i.e., from the extractors that are only known to be marginal-secure. 
These techniques include the \emph{one-bit argument} (in Section~\ref{sec:bOA:one_bit}), the   \emph{one-extra-source argument} (in Section~\ref{sec:bOA:one_source}), and the \emph{one-extra-block argument} (in Section~\ref{sec:bOA:one_block}). 
Together with Theorem~\ref{thm:SOA_GE}, we shall obtain strong GE security for these extractors.  

\subsection{With one-bit argument and XOR lemma} \label{sec:bOA:one_bit}

This technique relies on the equivalence between the strong marginal security and the strong OA security for one-bit extractors demonstrated in~\cite{KK12, KT08}. 
Thus, our argument first shows any strong multi-bit extractor with marginal security is trivially a strong marginal-secure one-bit extractor. 
Then we make use of the aforementioned connection to upgrade the strong marginal security to the strong OA security. 
Finally, by making use of the XOR lemma, we can generalize the analysis to multi-bit extractors with a loss on the parameters. 

It is worth mentioning that this technique is so general that it could be applied to single-source, two-source, and multi-source settings. 
However, in the single-source setting, the parameter loss is so huge to afford, whereas in the multi-source settings, we can do  better by using one extra source (see Section~\ref{sec:bOA:one_source}). 

On the other side, our technique is particularly useful for two-source extractors, and implies that \emph{all} best-known two-source extractors,\footnote{There are several incomparable two-source extractors with different advantages. See below for details.} \emph{as they are}, are in fact strongly quantum-secure with essentially the same parameters. 


\begin{lemma} \label{lem:M_multi_one}
For any $\emptyset \neq S\subseteq [m]$,  any $m$-bit extractor with output $Z \in \01^m$ is also a one-bit extractor with the same set of parameters and properties by outputting $Z_S= \bigoplus_{i\in S} z_i$. 
\end{lemma}

\begin{proof}
The lemma simply follows by definition and Fact~\ref{prelim:fact:monotone_trace}. 
\end{proof}

Then by a corollary\footnote{This corollary was originally stated for the IA security, which implies the OA security automatically.} from~\cite{KK12} (which is a simple application of techniques from~\cite{KT08}), we have 

\begin{lemma}[\cite{KK12}, Corollary 5.5] \label{lem:bOA:KK_single_bit}
If $\TExt$ is a classical $(n_1, k_1, n_2, k_2, 1, \eps)$ $X_1$-strong two-source extractor, then it is also a OA-secure $(n_1, k_1, n_2, k_2+\log(1/\eps), 1, \sqrt{\eps})$ $X_1$-strong two-source extractor. Similarly for $\TExt$ being $X_2$-strong. 
As a consequence, if $\TExt$ is a classical $(n_1, k_1, n_2, k_2, 1, \eps)$ two-source extractor that is both $X_1$-strong and $X_2$-strong, then it is also a OA-secure $(n_1, k_1+\log(1/\eps), n_2, k_2+\log(1/\eps), 1, \sqrt{\eps})$  two-source extractor that is both $X_1$-strong and $X_2$-strong. 
\end{lemma}

Thus, by making use of the quantum version of the XOR Lemma (Lemma~\ref{lem:KK_XOR}), we have 

\Knote{Future todo: May be able to save the $\log(1/\eps)$ loss in $k$ using Salil's technique}

\begin{theorem}  \label{thm:one_bit_OA}
If $\TExt$ is a classical $(n_1, k_1, n_2, k_2, m, \eps)$ $X_1$-strong two-source extractor, then it is also a OA-secure $(n_1, k_1, n_2, k_2+\log(1/\eps), m, 2^m\sqrt{\eps})$ $X_1$-strong two-source extractor. Similarly for $\TExt$ being $X_2$-strong. As a consequence, if $\TExt$ is a classical $(n_1, k_1, n_2, k_2, m, \eps)$ two-source extractor that is both $X_1$-strong and $X_2$-strong, then it is also a OA-secure $(n_1, k_1+\log(1/\eps), n_2, k_2+\log(1/\eps), m, 2^m\sqrt{\eps})$  two-source extractor that is both $X_1$-strong and $X_2$-strong. 
\end{theorem}

\begin{proof}
We only prove the theorem for the extractor being $X_1$-strong. A similar argument proves when the extractor is $X_2$-strong. By Lemma~\ref{lem:KK_XOR},  and for any subset $\emptyset \neq \tau \subseteq [t]$, let $\TExt_\tau(\cdot)=\TExt(\cdot)_\tau$ as defined in Lemma~\ref{lem:M_multi_one} , then we have, 
\begin{eqnarray*}
  \trdist{\rho_{\TExt(X_1,X_2)X_1\Adv}- \uniform{m} \ot \rho_{X_1\Adv}} & \leq &  \sqrt{2^m \sum_{\tau \neq \emptyset} \trdist{\rho_{\TExt_\tau(X_1,X_2)X_1 \Adv}- \uniform{1} \ot \rho_{X_1\Adv}}^2} \\
  &\leq & \sqrt{2^m \cdot 2^m \eps } =2^m\sqrt{\eps},
\end{eqnarray*}
where the second inequality is due to Lemma~\ref{lem:M_multi_one} and Lemma~\ref{lem:bOA:KK_single_bit}.
\end{proof}

\subsection*{Instantiations}

Here we apply Theorem~\ref{thm:one_bit_OA} and Theorem~\ref{thm:SOA_GE} to lift the security of existing (marginally secure) two-sources extractors to obtain GE-secure extractors with essentially the same parameters (up to a constant factor loss). We first consider Raz's extractor, which has the advantage to apply to two unequal length sources but one of them needs to have $> 1/2$ entropy rate.

\begin{theorem}[Raz's Extractor~\cite{Raz05}] \label{thm:Raz} For any $n_1,n_2,k_1,k_2,m$, and any $0<\delta<1/2$ with
\begin{itemize}
\item $n_1 \geq 6 \log n_1 +2\log n_2$
\item $k_1 \geq (0.5+\delta) n_1 + 3\log n_1 + \log n_2$
\item $k_2 \geq 5 \log(n_1-k_1)$
\item $m \leq \delta \min \{ n_1/8, k_2/40 \} - 1 $
\end{itemize}
There is an explicit $(n_1,k_1,n_2,k_2,m,\eps)$ two-source extractor 
with error $\eps = 2^{-1.5m}$. Furthermore, the extractor is both $X_1$-strong and $X_2$-strong.
\end{theorem}

\begin{theorem}[GE-secure Raz's Extractor] \label{thm:Raz-GE-secure} For any $n_1,n_2,k_1,k_2,m$, and any $0<\delta<1/2$ with
\begin{itemize}
\item $n_1 \geq 6 \log n_1 +2\log n_2$
\item $k_1 \geq (0.5+\delta) n_1 + 3\log n_1 + \log n_2$
\item $k_2 \geq 6 \log(n_1-k_1)$
\item $m \leq (\delta/16) \min \{ n_1/8, k_2/40 \} - 1 $
\end{itemize}
There is an explicit $(n_1,k_1,n_2,k_2,m,\eps)$ GE-secure two-source extractor 
with error $\eps = 2^{-1.5m}$. Furthermore, the extractor is both $X_1$-strong and $X_2$-strong.
\end{theorem}
\begin{proof}
Let $k'_1 = k_1 - 5 m$, $k'_2 = k_2 - 5 m$, $\delta' = \delta/2$, and $m' = \delta' \min \{ n_1/8, k'_2/40 \} - 1$. Note that $k'_1 \geq (0.5+\delta') n_1 + 3\log n_1 + \log n_2$, $k'_2 \geq 5 \log(n_1-k_1)$. By Theorem~\ref{thm:Raz}, there exists a $(n_1,k'_1,n_2,k'_2,m,\eps')$ two-source extractor $\TExt$ with $\eps' = 2^{-5m} \geq 2^{-1.5m'}$ that is both $X_1$-strong and $X_2$-strong.  By Theorem~\ref{thm:one_bit_OA}, $\TExt$ is also a OA-secure $(n_1, k'_1+\log(1/\eps'), n_2, k'_2+\log(1/\eps'), m, 2^{m}\sqrt{\eps'})$ two-source extractor that is both $X_1$-strong and $X_2$-strong. Note that $k_1 \geq k'_1+\log(1/\eps')$, $k_2 \geq k'_2+\log(1/\eps')$, and $2^{m}\sqrt{\eps'}\leq \eps$. By Theorem~\ref{thm:SOA_GE}, $\TExt$ is also a $(n_1,k_1,n_2,k_2,m,\eps)$ GE-secure two-source extractor that is both $X_1$-strong and $X_2$-strong.
\end{proof}

We next consider Bourgain's extractor, which breaks the ``1/2-barrier''. That is, the extractor works even if both sources have entropy rate (slightly) below $1/2$. 

\begin{theorem}[Bourgain's Extractor~\cite{Bour05}] \label{thm:Bou} There exists a universal constant $\alpha$ such that for any $n \in \N$, there is an explicit $(n,k,n,k,m,\eps)$ two source extractor with $k = (0.5-\alpha) n$, $m = \alpha n$ and $\eps = 2^{-\alpha n}$. Furthermore, the extractor is both $X_1$-strong and $X_2$-strong.
\end{theorem}

\begin{theorem}[GE-secure Bourgain's Extractor] \label{thm:Bou-GE-secure}  There exists a universal constant $\beta$ such that for any $n \in \N$, there is an explicit $(n,k,n,k,m,\eps)$ GE-secure two source extractor with $k = (0.5-\beta) n$, $m = \beta n$ and $\eps = 2^{-\beta n}$. Furthermore, the extractor is both $X_1$-strong and $X_2$-strong.
\end{theorem}
\begin{proof} Let $\beta = \alpha / 5$, where $\alpha$ is the universal constant in Theorem~\ref{thm:Bou}. Let $\TExt$ be the $(n,k',n,k',m',\eps')$ two-source extractor in Theorem~\ref{thm:Bou} that is both $X_1$-strong and $X_2$-strong, where $k' = (0.5-\alpha)n$, $m'=\alpha n$, and $\eps' = 2^{-\alpha n}$. Let $\eps'' = 2^{-4 \beta n} \geq \eps'.$ By Theorem~\ref{thm:one_bit_OA}, $\TExt$ is also a OA-secure $(n, k'+\log(1/\eps''), n, k'+\log(1/\eps''), m, 2^{m}\sqrt{\eps''})$ two-source extractor that is both $X_1$-strong and $X_2$-strong. Note that $k \geq k' + \log(1/\eps'')$, and $2^m \sqrt{\eps''} \geq \eps$. By Theorem~\ref{thm:SOA_GE}, $\TExt$ is also a $(n,k,n,k,m,\eps)$ GE-secure two-source extractor that is both $X_1$-strong and $X_2$-strong.
\end{proof}

Finally, we consider the DEOR extractor~\cite{DEOR04}, which has the advantage that the extractor works as long as the sum of the entropy from two sources is greater than $n$. 

\begin{theorem}[DEOR Extractors~\cite{DEOR04}] \label{thm:DEOR} For any $n, k_1, k_2, m$, there is an explicit $(n,k_1,n,k_2,m,\eps)$ two-source extractor with error $\eps = 2^{-(k_1+k_2+1-n-m)/2}$. Furthermore, the extractor is both $X_1$-strong and $X_2$-strong.
\end{theorem}

\begin{theorem}[GE-secure DEOR Extractors] \label{thm:DEOR-GE-secure} For any $n, k_1, k_2$ with $k_1+k_2 > n - 1$, there is an explicit $(n,k_1,n,k_2,m,\eps)$ two-source extractor with $m = \min\{ (k_1+k_2+1-n)/ 20, k_1/4, k_2/4\} $ and $\eps = 2^{-m}$. Furthermore, the extractor is both $X_1$-strong and $X_2$-strong.
\end{theorem}
\begin{proof}
Let $k'_1 = k_1 - 4m$ and $k'_2 = k_2-4m$. Let $\TExt$ be the $(n,k'_1,n,k'_2,m,\eps')$ two source extractor in Theorem~\ref{thm:DEOR} that is $X_1$-strong and $X_2$-strong, where $\eps' = 2^{-4m}$. By Theorem~\ref{thm:one_bit_OA}, $\TExt$ is also a OA-secure $(n, k'+\log(1/\eps'), n, k'+\log(1/\eps'), m, 2^{m}\sqrt{\eps'})$ two-source extractor that is both $X_1$-strong and $X_2$-strong. Note that $k_a \geq k'_a + \log(1/\eps')$ for both $a \in \{1,2\}$, and $2^m \sqrt{\eps'} \geq \eps$. By Theorem~\ref{thm:SOA_GE}, $\TExt$ is also a $(n,k_1,n,k_1,m,\eps)$ GE-secure two-source extractor that is both $X_1$-strong and $X_2$-strong.
\end{proof}

\paragraph{Remark.} With similar arguments, it is not hard to show that the best existential two-source extractor with logarithmic min-entropy (guaranteed by the probabilistic method) is also GE-secure with almost the same parameters. 

\subsection{With one extra independent source} \label{sec:bOA:one_source}

Our second technique is a transformation that uses one extra source to obtain strong OA security, and is particularly useful for the multi-source setting. In fact, it additionally offers several extra advantages: the original multi-source extractor does not need to be strong, yet the resulting extractor is strong for all but the last source, and extracts almost all min-entropy out from the last block. This in turn, allows us to use the output to extract from all but the last source, and extract all min-entropy out from all sources!

This technique relies on the following observation: for any marginally close-to-uniform distribution, one can add an independent quantum min-entropy source and make use of a quantum strong seeded extractor to lift its security. 
Precisely,  the marginal distribution will be used as the seed to extract from the additional independent quantum min-entropy source. Because the extractor is a strong seeded extractor, 
any quantum system that is associated with the marginal distribution can be added back without destroying its security. The following lemma formalizes the above idea. 


\begin{lemma} \label{lem:ind_lift}
Consider two \emph{independent}  cq states $\rho_{X_1E_1}$ and $\rho_{X_2E_2}$ such that $X_1=\01^{n_1}$ and $X_2=\01^{n_2}$ (i.e., the global system $\rho_{X_1X_2E_1E_2}=\rho_{X_1E_1} \ot \rho_{X_2E_2}$). 
Let $f : \01^{n_2} \rightarrow \01^d$ be any classical deterministic function.
If $\Hmin(X_1|E_1)_\rho \geq k$ and the marginal  of  $f(X_2)$ satisfies $\trdist{f(X_2)- \uniform{d}}\leq \delta$, then for any quantum strong $(k,\eps)$ extractor $\Ext : \01^{n_1} \times \01^{d} \rightarrow \01^m$, we have 
\[
  \trdist{\rho_{\Ext(X_1,f(X_2))X_2E_1E_2} -\uniform{m} \ot \rho_{X_2E_1E_2}} \leq \eps+\delta.
\]
Note that $\rho_{X_2E_1E_2}=\rho_{E_1} \ot \rho_{X_2E_2}$.
\end{lemma}

\begin{proof}
First note that since $\Ext$ is a quantum strong $(k,\eps)$ extractor and $\Hmin(X_1|E_1) \geq k$, we have that for an independent seed $\rho_{Y}=\uniform{d}$, 
\begin{equation*}
  \trdist{\rho_{\Ext(X_1,Y)YE_1} - \uniform{m} \ot \rho_{Y} \ot \rho_{E_1}} \leq \eps,
\end{equation*}
which is equivalent to
\begin{equation} \label{eqn:ext_decomp}
  w_y \defeq \trdist{\rho_{\Ext(X,y)yE_1} -\uniform{m} \ot \ketbra{y}\ot  \rho_{E_1}}
  \text{ and } \sum_{y \in \01^{d}} \frac{1}{2^{d}} w_y \leq \eps.
\end{equation}
Moreover, let $\rho_{f(X_2)X_2E_2}=\sum_{x_2 \in \01^{n_2}} p_{x_2} \ketbra{f(x_2), x_2} \ot \rho^{x_2}_{E_2}$. For each $x_2 \in \01^{n_2}$, define
\begin{equation}
  u_{x_2}\defeq \trdist{ \rho_{\Ext(X_1, f(x_2))f(x_2)x_2E_1E_2}-\uniform{m} \ot \ketbra{f(x_2),x_2} \ot \rho_{E_1E_2}^{x_2}}.
\end{equation}
Note that $\rho_{\Ext(X_1, f(x_2))f(x_2)x_2E_1E_2}=\rho_{\Ext(X_1,f(x_2))f(x_2)E_1} \ot \ketbra{x_2} \ot \rho^{x_2}_{E_2}$ and $\rho_{E_1E_2}^{x_2}=\rho_{E_1} \ot \rho_{E_2}^{x_2}$. 
Hence, by Fact~\ref{fact:trdist:add_prod}, for each $x_2 \in \01^{n_2}$,
\begin{eqnarray}
  \nonumber u_{x_2} & = & \trdist{ \rho_{\Ext(X_1, f(x_2))f(x_2)x_2E_1E_2}-\uniform{m} \ot \ketbra{f(x_2),x_2} \ot \rho_{E_1E_2}^{x_2}} \\ 
  & = & \trdist{\rho_{\Ext(X_1, f(x_2))f(x_2)E_1}-\uniform{m} \ot \ketbra{f(x_2)} \ot \rho_{E_1}}=w_{f(x_2)}. \label{eqn:u_x_w_f}
\end{eqnarray}
By Fact~\ref{fact:trdist:c_decomp},  observe that 
\begin{eqnarray*}
 & & \trdist{\rho_{\Ext(X_1,f(X_2))X_2E_1E_2} -\uniform{m} \ot  \rho_{X_2E_1E_2}} \\
 & = &\trdist{\rho_{\Ext(X_1,f(X_2))f(X_2)X_2E_1E_2} -\uniform{m} \ot  \rho_{f(X_2)X_2E_1E_2}} =\sum_{x_2 \in \01^{n_2}} p_{x_2} u_{x_2}.
 \end{eqnarray*}
By (\ref{eqn:u_x_w_f}), it is easy to see that 
\[
  \sum_{x_2 \in \01^{n_2}} p_{x_2} u_{x_2}=\sum_{z \in \01^d} \sum_{f(x_2)=z} p_{x_2}u_{x_2}
  =\sum_{z \in \01^d} w_z \sum_{f(x_2)=z} p_{x_2}=\sum_{z \in \01^d} p_z w_z,
\]
in which $p_z$ denotes the marginal distribution of $f(X_2)$. Finally, we have 
 \begin{eqnarray*}
 \sum_{x_2 \in \01^{n_2}} p_{x_2} u_{x_2} & = & \sum_{z \in \01^d} p_z w_z =  \sum_{z \in \01^d} \frac{1}{2^d}w_z + \sum_{z \in \01^d} (p_z -\frac{1}{2^d})w_z \\
 & \leq & \eps+ \sum_{z: p_{z}> \frac{1}{2^{d}}} (p_{z}- \frac{1}{2^{d}}) \leq \eps+\delta,
 \end{eqnarray*}
where the first inequality is because of (\ref{eqn:ext_decomp}) and $0\leq w_z \leq 1$ and the second inequality is due to $\trdist{Z=f(X_2)- \uniform{d}}\leq \delta$ and (\ref{eqn:fact:c_dist}). 
\end{proof}

Now we present a general construction of strong OA-secure multi-source extractors 
from a classical independent source extractor and a quantum strong seeded extractor, which requires one more independent source but can match almost all parameters of classical independent source extractors.

\begin{figure}
\begin{protocol*}{Multi-source Extractor $\QMExt$}
\begin{description}
\item Let $\IExt: (\zo^n)^t \rightarrow \zo^m$ be a classical $(t, n,k, m,\eps_1)$ independent source extractor. 
\item Let $\Ext_q: \zo^n \times \zo^m \rightarrow \zo^l$ be a quantum strong $(k, \eps_2)$ seeded extractor.
\item Construct $\QMExt: (\zo^n)^{t+1} \rightarrow \zo^l$ as follows:
\end{description}
\begin{step}
\item Let $Z=\IExt(X_1,\cdots, X_t)$. 
\item $\QMExt(X_1, \cdots, X_t, X_{t+1})\defeq \Ext_q(X_{t+1}, Z)= \Ext_q(X_{t+1}, \IExt(X_1, \cdots, X_t))$. 
\end{step}
\end{protocol*}
\caption{Construction of $\QMExt$ from any classical independent source extractor $\IExt$.}
\label{fig:QIndExt}
\end{figure}

\begin{theorem}   \label{thm:Ind_QIExt}
Let $\IExt: (\01^n)^t \rightarrow \01^m$ be any classical $(t, n,k,m,\eps_1)$ independent source extractor. Let $\Ext_q: \zo^n \times \zo^m \rightarrow \zo^l$ be a quantum strong $(k, \eps_2)$ randomness extractor. Then $\QMExt$ (constructed in Fig.~\ref{fig:QIndExt}) is an OA-secure $(t+1, n,k, l, \eps_1+\eps_2)$ multi-source extractor. Moreover, $\QMExt$ is $X_S$-strong for $S=\{1, \cdots, t\}$.  
\end{theorem}

\begin{proof}
Consider any $t+1$ independent sources $X_1, \cdots, X_{t+1} \in \01^n$ and the quantum side information $\rho_{\Adv}$  that is generated in the OA model. Note that in this case $\Adv=E_{i^*}$ for a single $i^* \in [t+1]$. By definition, 
 $\Hmin(X_i)\geq \Hmin(X_i|E_i)_{\rho_i}=\Hmin(X_i|E_i)_\rho \geq k$ for each $i \in [t+1]$, and moreover $X_1, \cdots X_t$ are independent. 
Because $\IExt$ is a $(t, n,k,m,\eps_1)$ independent source extractor, by definition, we have 
\[
  \trdist{\IExt(X_1,\cdots, X_t)-\uniform{m}} \leq \eps_1. 
\]
Let $Z=\IExt(X_1,\cdots, X_t)$. Hence, we have $\trdist{Z- \uniform{m}}\leq \eps_1$.  If $i^* \in [t]$, then we have $\rho_{X_1\cdots X_t\Adv}$ and $\rho_{X_{t+1}}$ are independent cq states and $\Hmin(X_{t+1}) \geq k$. Otherwise, we have $i^*=t+1$, and $\rho_{X_1\cdots X_t}$ and $\rho_{X_{t+1}\Adv}$ are independent cq states and $\Hmin(X_{t+1}|\Adv)_\rho \geq k$.   
In either case,  by Lemma~\ref{lem:ind_lift}, we have 
\[ 
   \trdist{\rho_{\Ext(Z,X_{t+1})X_1\cdots X_t\Adv} -\uniform{l} \ot \rho_{X_1\cdots X_t \Adv}} \leq \eps_1+\eps_2,
\]
which, by definition, completes the proof. 
\end{proof}


\subsection*{Instantiations}

Here we apply Theorem~\ref{thm:Ind_QIExt} and Theorem~\ref{thm:SOA_GE} to lift the security of existing (marginally secure) multi-sources extractors to obtain strong GE-secure extractors that extract all min-entropy out. 

The best known multi-source extractor is Li's extractor~\cite{Li13b}, which can extract randomness from a constant number of sources with entropy as low as $\polylog(n)$. 

\begin{theorem}[Li's Extractor \cite{Li13b}] \label{thm:Li_IExt}
For every constant $\eta > 0$ and all $n,k \in \N$ with $k \geq \log^{2+\eta} n$, there exists an explicit $(t,n,k,m,\eps)$ independent source extractor with $m = \Omega(k)$, $t = O(1/\eta) + O(1)$, and $\eps = 1/\poly(n) + 2^{-k^{\Omega(1)}}$. 
\end{theorem}

By using the quantum strong seeded extractor from Theorem~\ref{thm:Ext1} in Theorem~\ref{thm:Ind_QIExt} (and applying Theorem~\ref{thm:SOA_GE}), we obtain strong GE-secure version of Li's extractor with improved output length.  

\begin{theorem}[GE-secure Li's Extractor]  \label{thm:Li_GE} For every constant $\eta > 0$ and all $n,k \in \N$ with $k \geq \log^{2+\eta} n$, there exists an explicit $(t+1,n,k,m,\eps)$ independent source extractor with $m = k- o(k)$, $t = O(1/\eta) + O(1)$, and $\eps = 1/\poly(n) + 2^{-k^{\Omega(1)}}$. Furthermore, the extractor is $X_{S}$-strong for $S = \{1,\dots,t\}$.
\end{theorem}

The downside of  Li's extractor is that it has at least $1/\poly(n)$ error. The following extractor from~\cite{BarakRSW06,Rao06} achieves (sub-)exponentially small error, but uses $O(\log n / \log k)$ sources.

\begin{theorem}[BRSW Extractor\cite{BarakRSW06,Rao06}] For any $n,k \in \N$ with $k \geq \log^{10}n$, there exists an explicit $(t,n,k,m,\eps)$ independent source extractor with $m = \Omega(k)$, $t = O(\log n / \log k)$, and $\eps = 2^{-k^{\Omega(1)}}$.
\end{theorem}

As before, using the quantum strong seeded extractor from Theorem~\ref{thm:Ext1} in Theorem~\ref{thm:Ind_QIExt} (and applying Theorem~\ref{thm:SOA_GE}), we obtain strong GE-secure version of BRSW extractor with improved output length.

\begin{theorem}[GE-secure BRSW Extractor\cite{BarakRSW06,Rao06}] \label{thm:BRSW_GE} For any $n,k \in \N$ with $k \geq \log^{10}n$, there exists an explicit $(t+1,n,k,m,\eps)$ independent source extractor with $m = k- o(k)$, $t = O(\log n / \log k)$, and $\eps = 2^{-k^{\Omega(1)}}$. Furthermore, the extractor is $X_{S}$-strong for $S = \{1,\dots,t\}$.
\end{theorem}


\subsection{With one extra block in block-sources} 
\label{sec:bOA:one_block}



In the context of block+general source extractors (e.g., \cite{BarakRSW06}), our third technique is to add one extra block to the existing block source and make use of one classical and one quantum strong seeded extractor to lift its security. 
Comparing to the technique in Section~\ref{sec:bOA:one_source}, we only require to add one extra block that is not independent of existing sources.
Thus, it is conceivable that we need more complicated techniques to obtain strong OA security in this case. 
To that end, we make use of the technique called \emph{alternating extraction}, and moreover, a quantum strong seeded extractor at the last step to achieve this goal. As in Section~\ref{sec:bOA:one_source}, we are able to improve the output length to extract all min-entropy out, but this time, we need to start from a strong block+general source extractor. 


\begin{figure}
\begin{protocol*}{Block + General Source Extractor $\QBSAExt$}
\begin{description}
\item Let $\BExt: \zo^{n_1} \times \zo^{n_3} \rightarrow \zo^{m_1}$ be a classical block+general source extractor that works with a $k_1$-block-source and an independent $(n_3, k_3)$ source. Moreover, the extractor is \emph{strong} in the block-source side and with error $\eps_1$. 
\item Let $\Ext_c: \zo^{n_2} \times \zo^{m_1'} \rightarrow \zo^{m_2}$ be a classical strong $(k_2, \eps_2)$ seeded extractor.
\item Let $\Ext_q: \zo^{n_3} \times \zo^{m_2} \rightarrow \zo^l$ be a quantum strong $(0.9k_3, \eps_3)$ seeded extractor.
\item Let $(X_1, X_2) \in \01^{n_1} \times \01^{n_2}$ be a block-source with $C+1$ blocks in which $X_1$ contains $k_1$-block-source with $C$ blocks and $X_2$ is the last block and an $(n_2, k_2)$ source conditioned on $X_1$. 
Let $X_3 \in \01^{n_3}$ be an independent min-entropy source.  Both $(X_1,X_2)$ and $X_3$ could have quantum side information. 
\item Construct $\QBSAExt: \zo^{n_1+n_2} \times \zo^{n_3} \rightarrow \zo^l$ as follows:
\end{description}
\begin{step}
\item Apply $\BExt$  to  obtain $R$ that is the first $0.05k_3$ bits of $\BExt(X_1, X_3)$. (i.e., $m_1'<m_1$. )
\item Alternating Extraction: let $T=\Ext_c(X_2, R)$ and $Z=\Ext_q(X_3, T)$.
\item $\QBSAExt( (X_1,X_2), X_3) \defeq Z$. 
\end{step}
\end{protocol*}
\caption{Construction of $\QBSAExt$ from any classical block + general source extractor $\BExt$.}
\label{fig:QBSAExt}
\end{figure}


\begin{theorem}\label{thm:one_extrac_block}
Let $\BExt: \zo^{n_1} \times \zo^{n_3} \rightarrow \zo^{m_1}$ be a classical block+general source extractor that makes use of a $k_1$-block source with $C$ blocks and one extra $k_3$-source to output $m_1$ uniform bits with error $\eps_1$. Moreover, $\BExt$ is \emph{strong} in the block-source side. Then $\QBSAExt: \zo^{n_1+n_2} \times \zo^{n_3} \rightarrow \zo^l$ constructed in Fig.~\ref{fig:QBSAExt} is an OA-secure block+general source extractor that makes use of a $(k_2, k_1, \cdots, k_1)$-block source with $C+1$ blocks and one extra $k_3$-source, both entropy measured in the OA model, to output $l$ uniform bits with error  $4\eps_1+2\eps_2 + \eps_3 + 2^{-\Omega(k_3)}$. Moreover, $\QBSAExt$ is \emph{strong} in the block source $(X_1, X_2)$. 
\end{theorem}

\begin{proof}
Let us first argue that such alternating extraction works without considering the OA side information. For now, let us imagine all the entropies are measured in the no side information case and lie in the range for the extractors to work. Thus, by definition of $\BExt$, we have 
\[
   \trdist{(X_1, R) - X_1 \ot \uniform{m_1}} \leq \eps_1. 
\]
Fix $X_1=x_1$ and let $w_{x_1}=\trdist{(x_1, R)-x_1 \ot \uniform{m_1}}$. Then we have $\Exp_{x_1\sim X_1} w_{x_1} \leq \eps_1$. Because $X_2$ is independent of $R$ conditioned on $x_1$, thus we have 
\[
  \trdist{(X_2,x_1, R) - (X_2, x_1) \ot \uniform{m_1}} =w_{x_1}. 
\]
Now by definition of $\Ext_c$ and the fact that $(X_2,X_1)$ is a $k$-block source, we have 
\begin{equation} \label{eqn:extc}
  \trdist{(T, x_1, R) - \uniform{m_2} \ot x_1 \ot \uniform{m_1}} \leq w_{x_1}+\eps_2.
\end{equation}
By an triangle inequality and the definition of $w_{x_1}$, it is easy to see that 
\[
  \trdist{(T,x_1, R) - \uniform{m_2} \ot (x_1, R)} \leq 2w_{x_1}+\eps_2. 
\]
Because $T$ and $X_3$ are independent conditioned on $x_1$ and $R$, then we have 
\[
\trdist{(T,x_1, R, X_3) - \uniform{m_2} \ot (x_1, R, X_3)} \leq 2w_{x_1}+\eps_2.
\]
Let us further condition on $R=r$, and let 
\[
w_{x_1,r}=\trdist{(T,x_1, r, X_3) - \uniform{m_2} \ot (x_1, r, X_3)},
\]
Namely, we have $\Exp_{r \sim R|_{x_1}} w_{x_1,r} \leq 2w_{x_1}+\eps_2$.
By Lemma~\ref{lem:condition} (resp. in the presence of quantum side information, we invoke Lemma~\ref{lem:q_condition}), with probability $1-2^{-0.05k_3}$ over $r$,  $X_3$ still has min-entropy (resp. quantum conditional min-entropy) at least $k_3-0.05k_3-0.05k_3\geq 0.9k_3$.
By the definition of $\Ext_q$  we have 
\begin{equation} \label{eqn:extq}
\trdist{(T,x_1, r, Z) - \uniform{m_2} \ot (x_1,r) \ot \uniform{l}} \leq w_{x_1,r} + \eps_3+2^{-\Omega(k_3)}.
\end{equation}
By triangle inequalities, and take average over $x_1,r$, then we have 
\[
\trdist{(T, X_1, R, Z) - (T, X_1, R) \ot \uniform{l}} \leq 4\eps_1+2\eps_2 + \eps_3 + 2^{-\Omega(k_3)}.
\]
Because $X_2$ and $Z$ are independent conditioned on $T$ and $R$, thus we have 
\begin{equation} \label{eqn:strong_block}
 \trdist{(X_2, X_1, Z) - (X_2, X_1) \ot \uniform{l}} \leq 4\eps_1+2\eps_2 + \eps_3 + 2^{-\Omega(k_3)}.
\end{equation}
Namely, we prove our claim in the case of no side information. 

Now let us proceed to see what happens with the OA side information. First notice that no matter which source the OA adversary gets side information from, the OA entropy is a lower bound on the marginal entropy. Moreover, we will invoke Lemma~\ref{lem:q_condition} instead of Lemma~\ref{lem:condition} in the presence of quantum side information.  Thus all the sources have sufficient entropy to guarantee the success of the above argument. 

In the case in which the OA adversary gets side information from the block-source $(X_1, X_2)$, we are already done because the side information can be generated after obtaining (\ref{eqn:strong_block}). In the case in which the side information is from $X_3$, because we use a quantum-proof strong extractor at the last step, we still have  the side information version of (\ref{eqn:extq}). All of the rest arguments still apply. 
Thus, let $\Adv$ denote the OA side information, we always have 
\[
  \trdist{\rho_{X_1X_2 \QBSAExt(X_1,X_2,X_3) \Adv} - \uniform{l} \ot \rho_{X_1X_2\Adv}} \leq 4\eps_1+2\eps_2 + \eps_3 + 2^{-\Omega(k_3)},
\]
which completes the proof. 
\end{proof}

%

\subsection*{Instantiations}



Here we apply Theorem~\ref{thm:one_extrac_block} and Theorem~\ref{thm:SOA_GE} to lift the security of a block+general source extractor of~\cite{BarakRSW06} to obtain a strong GE-secure version extractor that extracts all min-entropy out.

\begin{theorem}[BRSW Block+general Source Extractor~\cite{BarakRSW06}] There exists a constant $c$ such that for every $n, k \in \N$, and $C = c \cdot \log n / \log k$ there exists a classical block+general source extractor $\BExt: \zo^{Cn} \times \zo^n \rightarrow \zo^m$ that make use of a $k$-block source with $C$ blocks and one general $k$-source to output $m = \Omega(k)$ uniform bits with error $\eps = n^{-\Omega(1)}$. Moreover, $\BExt$ is strong in the block-source side.
\end{theorem}

Note that the extractor we cite above is strong in the block-source side but has inverse polynomial error. \cite{BarakRSW06} also showed their construction has exponentially small error, but it is no longer clear if it is strong in the block-source side. It is an interesting question to see whether the extractor is strong with exponentially small error.

By using the quantum strong seeded extractor from Theorem~\ref{thm:Ext1} in Theorem~\ref{thm:one_extrac_block} (as both $\Ext_c$ and $\Ext_q$) and applying Theorem~\ref{thm:SOA_GE}, we obtain GE-secure version of this extractor. 

\begin{theorem}[GE-secure BRSW Block+general Source Extractor]\label{thm:2-block-general-ext1}
There exists a constant $c$ such that for every $n, k \in \N$ with $k \geq \log^3 n$, and $C = c \cdot \log n / \log k$ there exists a GE-secure block+general source extractor $\QBSAExt: \zo^{(C+1)n} \times \zo^n \rightarrow \zo^m$ that make use of a $k$-block source with $C+1$ blocks and one general $k$-source to output $m = k-o(k)$ uniform bits with error $\eps = n^{-\Omega(1)}$. Moreover, $\QBSAExt$ is strong in the block-source side.
\end{theorem}

We also apply this technique to Raz's two-source extractor to obtain a strong  GE-secure two-block+general source extractor that extracts all entropy out. We will later use this extractor in Section~\ref{sec:network}.\footnote{We mention that we do not make attempt to optimize the parameters of this extractor. For example, the entropy rate of the second block do not need to be $\geq 1/2$. We state the extractor in a way that it is sufficient to be used in Section~\ref{sec:network}.}

\begin{theorem}[GE-secure Two-block+general Source Extractor] \label{thm:2-block-general-ext}
For any $n_1, n_2, k_1, k_2 \in \N$ and any $0<\delta<1/2$ with $k_1,k_2 \geq \log^5(n_1+n_2)$ and $k_1 \geq (0.5+\delta) n_1$, there exists a GE-secure block+general source extractor $\QBSAExt: \zo^{2n_1} \times \zo^{n_2} \rightarrow \zo^m$ that make use of a $k_1$-block source with $2$ blocks and one general $k_2$-source to output $m = k-o(k)$ uniform bits with error $\eps = 2^{-k_2^{\Omega(1)}}$. Moreover, $\QBSAExt$ is strong in the block-source side.
\end{theorem}

\def \bits {\{0, 1\}}
\def \e {\epsilon}
\def \cond {\mathrm{Cond}}
\def \bext {\mathrm{BExt}}
\def \supp {\mathrm{Supp}}

\section{A New Three-source Extractor and its GE-security} \label{sec:three_Ext}
In this section we construct a \emph{new} strong three-source extractor for sources of uneven lengths. 
Moreover, we prove the strong OA-security of the newly constructed extractor (which is essentially due to our technique in Section~\ref{sec:bOA:one_block}) and then make use our OA-GE equivalence to convert it into a GE-secure strong three source extractor.
We also make use of the newly constructed three source extractor to obtain a strong GE-secure seeded extractor that works even if the seed only has min-entropy rate bigger than a half. 
We will demonstrate its application to privacy amplification and quantum key distribution in Section~\ref{sec:PA}. 


We start with the construction of a strong classical three source extractor. We will first list some of the previous work that we use.

\subsection{Somewhere Random Sources, Extractors and Condensers}

\begin{definition} [Somewhere Random sources] \label{def:SR} A source $X=(X_1, \cdots, X_t)$ is $(t \times r)$
  \emph{somewhere-random} (SR-source for short) if each $X_i$ takes values in $\bits^r$ and there is an $i$ such that $X_i$ is uniformly distributed.
\end{definition}

\begin{definition}
An elementary somewhere-k-source is a 	vector of sources $(X_1, \cdots, X_t)$, such that some $X_i$ is a $k$-source. A somewhere $k$-source is a convex combination of elementary somewhere-k-sources.
\end{definition}

\begin{definition}
A function $C: \bits^n \times \bits^d \to \bits^m$ is a $(k \to l, \epsilon)$-condenser if for every $k$-source $X$, $C(X, U_d)$ is $\epsilon$-close to some $l$-source. When convenient, we call $C$ a rate-$(k/n \to l/m, \epsilon)$-condenser.   
\end{definition}

\begin{definition}
A function $C: \bits^n \times \bits^d \to \bits^m$ is a $(k \to l, \e)$-somewhere-condenser if for every $k$-source $X$, the vector $(C(X, y)_{y \in \bits^d})$ is $\e$-close to a somewhere-$l$-source. When convenient, we call $C$ a rate-$(k/n \to l/m, \epsilon)$-somewhere-condenser.   
\end{definition}

\begin{theorem} [\cite{BarakKSSW05, Zuc07}] \label{thm:swcondenser}
For any constant $\delta>0$, there is an efficient family of rate-$(\delta \to 0.9, \epsilon=2^{-\Omega(n)})$-somewhere condensers $\cond: \bits^n \to (\bits^m)^D$ where $D=\poly(1/\delta)=O(1)$ and $m=\Omega(n)$. 
\end{theorem}

\begin{theorem} [\cite{Rao06, BarakRSW06}] \label{thm:srgeneral_a} 
For any constant $C>1$ and every $n,k(n)$ with $k > \log^{2}n$, there is a polynomial time computable function $\SRExt:\{0,1\}^{n} \times \{0,1\}^{Ck} \rightarrow \{0,1\}^m $ s.t. if $X$ is an $(n,k)$ source and $Y$ is a $(C \times k)$-SR-source,

\[ | (Y , \SRExt(X,Y)) - (Y , U_m) | < \epsilon \]
and
\[ | (X , \SRExt(X,Y)) - (X , U_m) | < \epsilon \]

where $U_m$ is independent of $X,Y$, $m = \Omega(k)$ and $\epsilon = 2^{-\Omega(k)}$.
\end{theorem}

\subsection{Extractor Construction and its Marginal Security}

Given a $(k_1, k_2)$ block source $X=(X_1, X_2) \in \zo^{n_1} \times \zo^{n_2}$ and an independent source $(n_3, k_3)$ source $X_3$ such that $k_1 \geq \delta n_1$ for some constant $\delta>0$, our block source extractor is given in Figure~\ref{fig:BExt}.

We note the proof here share a lot of similarity with the one for Theorem~\ref{thm:one_extrac_block}, however, with concrete instantiation and parameters. 

\begin{figure}
\begin{protocol*}{Block-source Extractor $\bext$ ($\QBSAExt$)}
\begin{description}
\item Let $\cond$ be the somewhere condenser in Theorem~\ref{thm:swcondenser}. 
\item Let $\Raz$ be the strong two-source extractor from Theorem~\ref{thm:Raz}.
\item Let $\SRExt$ be the extractor from Theorem~\ref{thm:srgeneral_a}.
\item Let $\Ext_c$ be a strong seeded extractor that uses $\Omega(k)$ bits to extract $m$ bits from an $(n_3, 0.9k_3)$ source with error $\e$. 
In the \emph{quantum} case,  we will use a quantum strong seeded extractor $\Ext_q$. 
\item Construct $\bext: \zo^{n_1} \times \zo^{n_2} \times \zo^{n_3} \to \zo^m$ as follows:
\end{description}
\begin{step}
\item $Y=\cond(X_1)$ such that $Y$ has $D=O(1)$ rows and each row has length $\Omega(n_1)$.
\item For each row $i$ of $Y$, apply $\Raz$ to $Y_i$ and $X_3$ and output $\ell=\Omega(k)$ bits such that $D\ell \leq 0.05k_3$. Concatenate these outputs to get $W_3$.
\item Let $V=\SRExt( X_2, W_3)$
\item $\bext(X_1, X_2, X_3) \defeq Z=\Ext_c(X_3, V)$. \\  In the quantum case, $\QBSAExt(X_1,X_2,X_3) \defeq Z =\Ext_q(X_3, v)$. 
\end{step}
\end{protocol*}
\caption{Construction of $\bext$ ($\QBSAExt$) for a weak source and a block source of two blocks.}
\label{fig:BExt}
\end{figure}

\begin{theorem}\label{thm:bext}
For all $n_1, k_1, n_2, k_2, n_3, k_3, k \in \N$ and constant $\delta>0$ such that $k_1 \geq \delta n_1$, $\min(k_1, k_2, k_3) \geq k \geq \log^3 (\max(n_1, n_2, n_3))$, the function $\bext: \bits^{n_1} \times \bits^{n_2} \times \bits^{n_3} \to \bits^m$ described in Figure~\ref{fig:BExt} is a block source extractor such that if $X=(X_1, X_2)$ is a $(k_1, k_2)$ block source on $n_1+n_2$ bits and $X_3$ is an independent $(n_3, k_3)$ source, then 

\[\trdist{(\bext(X, X_3), X)-\uniform{m} \ot X} \leq 2^{-\Omega(k)}+\e.\]

\end{theorem}

\begin{proof}
By Theorem~\ref{thm:swcondenser}, $Y$ is $2^{-\Omega(n_1)}$-close to a somewhere entropy rate $0.9$ source. Without loss of generality we can assume that it is an elementary somewhere-rate-$0.9$ source. Ignoring the error, now by Theorem~\ref{thm:Raz}, $W_3$ is $2^{-\Omega(k)}$-close to a somewhere random source with $D=O(1)$ rows and each row has length $\ell=\Omega(k)$. Note that since $\Raz$ is a strong two-source extractor, thus the previous statement is true even if we condition on the fixing of the source $X_1$. Note that after this fixing, $W_3$ is a deterministic function of $X_3$, and is thus independent of $X_2$.

Note that since $X=(X_1, X_2)$ is a $(k_1, k_2)$ block source, we have that conditioned on the fixing of $X_1$, $X_2$ is an $(n_2, k_2)$ source. Now by Theorem~\ref{thm:srgeneral}, we have

\[\trdist{(V, W_3)- \uniform{} \ot W_3} \leq 2^{-\Omega(\ell)}=2^{-\Omega(k)}.\]

Thus we can further fix $W_3$, and condition on this fixing, $V$ is $2^{-\Omega(k)}$-close to uniform. Note that after this conditioning, $V$ is a deterministic function of $X_2$, and is thus independent of $X_3$. Furthermore by Lemma~\ref{lem:condition} we know that with probability $1-2^{-0.05k}$ over this fixing, $X_3$ still has min-entropy at least $k_3-0.05k-D\ell \geq 0.9k_3$. Since $\Ext$ is a strong $(0.9k_3, \e)$ extractor, we have that

\[\trdist{(Z, V)-\uniform{m} \ot V}\leq \e.\]

Note that $Z=\Ext(X_3, V)$. Thus conditioned on $V$, $Z$ is a deterministic function of $X_3$, which is independent of $X_2$. Thus we also have that

\[\trdist{(Z, X_2)-\uniform{m} \ot X_2} \leq \e.\]

Note that we have already fixed $X_1$. Thus adding back all the errors we get

\[\trdist{(Z, X_1, X_2)-\uniform{m} \ot ( X_1, X_2)} \leq \e+2^{-\Omega(n_1)}+2^{-\Omega(k)}+2^{-\Omega(k)}+2^{-0.05k}=2^{-\Omega(k)}+\e.\]
\end{proof}

One corollary of this theorem is as follows.

\begin{corollary}\label{cor:weakseed}
For any constant $\delta>0$ there exists a constant $C=\poly(1/\delta)$ such that if there is a classical strong $(k, \e)$ extractor $\Ext_c: \bits^n \times \bits^d \to \bits^m$ with $d \leq k/C$, then there is another strong $(1.2k, \e+2^{-\Omega(d)})$ extractor $\Ext'_c: \bits^n \times \bits^{d'} \to \bits^m$ where $d'=O(d)$ and $\Ext'_c$ works even if the seed only has min-entropy $(1/2+\delta)d'$.
\end{corollary}

\begin{proof}
We first show that for any weak source $R$ on $d'$ bits with min-entropy $(1/2+\delta)d'$, if we divide it into two equal blocks $R=(R_1, R_2)$, then it is $2^{-\Omega(d')}$-close to a $(\delta d', \delta d'/2)$ block source. Indeed, we have that for any $r \in \supp(R_1)$, $\Pr[R_1=r] \leq 2^{d'/2} 2^{-(1/2+\delta)d'} =2^{-\delta d'}$. Thus $R_1$ is a $\delta d'$ source. Now by Lemma~\ref{lem:condition}, we have that with probability $1-2^{-\delta d'/2}$ over the fixing of $R_1$, $R_2$ has min-entropy at least $(1/2+\delta)d'-d'/2-\delta d'/2=\delta d'/2$. Thus $R=(R_1, R_2)$ is $2^{-\Omega(d')}$-close to a $(\delta d', \delta d'/2)$ block source.

Now we can apply Theorem~\ref{thm:bext} where $R=(R_1, R_2)$ is the block source and $X$ is an independent $(n, k)$ source to construct $\Ext'_c$, where the $k$ in that theorem will be $\delta d'/2=O(d)$. We can choose $C=\poly(1/\delta)$ large enough so that in step 3 we can output $d$ bits while still satisfying that $D\ell \leq 0.05k_3$. Note that $0.9 \cdot 1.2k > k$, so we can use the strong extractor $\Ext_c$ to compute the final output $Z=\Ext_c(X, V)$, and the error is $\e+2^{-\Omega(d)}$. 
\end{proof}

\paragraph{Instantiations.}We can instantiate with the following classical extractor of best-known parameters and get two corollaries.

\begin{theorem} [\cite{GuruswamiUV09}] \label{thm:optext} 
For every constant $\alpha>0$, and all positive integers $n,k$ and $\e>0$, there is an explicit construction of a strong $(k,\e)$ extractor $\Ext: \bits^n \times \bits^d \to \bits^m$ with $d=O(\log n +\log (1/\e))$ and $m \geq (1-\alpha) k$.
\end{theorem}

\begin{corollary}\label{cor:weakseed2}
For any constant $\delta>0$ there exist constants $C>1$ and $\alpha>0$ such that for any $n, k \in \N$ and $2^{-\alpha k} \leq \e \leq 2^{-C\log^3 n}$ there is an efficient strong $(k, \e)$ extractor $\Ext: \bits^n \times \bits^d \to \bits^m$ with $d =O(\log (n/\e))$ and $m=0.9k$ that works even if the seed only has min-entropy $(1/2+\delta)d$.
\end{corollary}
\begin{proof}
This follows directly from Corollary~\ref{cor:weakseed} and Theorem~\ref{thm:optext}. 
\end{proof}


\begin{corollary}\label{cor:3source}
For any constant $\delta>0$ there exist constants $C>1$ and $\alpha>0$ such that for any $n, k \in \N$ and $2^{-\alpha k} \leq \e \leq 2^{-C\log^3 n}$ there is an efficient function $\Ext: \bits^d \times \bits^d \times \bits^n \to \bits^m$ with $d =O(\log (n/\e))$ and $m=0.9k$, such that if $X_1, X_2$ are two independent $(d, \delta d)$ sources and $X_3$ is an independent $(n, k)$ source then 

\[\trdist{(\Ext(X_1, X_2, X_3), X_1, X_2)-(U_m, X_1, X_2)} \leq \e.\]
\end{corollary}

\begin{proof}
Note that since $X_1, X_2$ are independent, they form a $(\delta d, \delta d)$ block source. Note that $\delta d=\Omega(\log^3 n) >\log^2 n$, so we can apply Theorem~\ref{thm:bext} such that the final error is at most $\e$.
\end{proof}

\subsection{Strong OA-security and Instantiations} 

The strong OA-security of $\bext$ in Figure~\ref{fig:BExt} is quite straightforward from the proof of Theorem~\ref{thm:bext} and Theorem~\ref{thm:one_extrac_block}.

\begin{theorem}\label{thm:qbext}
For all $n_1, k_1, n_2, k_2, n_3, k_3, k \in \N$ and constant $\delta>0$ such that $k_1 \geq \delta n_1$, $\min(k_1, k_2, k_3) \geq k \geq \log^3 (\max(n_1, n_2, n_3))$, the function $\QBSAExt: \bits^{n_1} \times \bits^{n_2} \times \bits^{n_3} \to \bits^m$ described in Figure~\ref{fig:BExt} is an OA-secure block source extractor such that if $X=(X_1, X_2)$ is a $(k_1, k_2)$ block source on $n_1+n_2$ bits and $X_3$ is an independent $(n_3, k_3)$ source, and let $\Adv$ denote the side information, then  
\[\trdist{\rho_{\QBSAExt(X, X_3) X\Adv}-\uniform{m} \ot \rho_{X\Adv}} \leq 2^{-\Omega(k)}+\e.\]
\end{theorem}

\begin{proof}
(Sketch): the proof of Theorem~\ref{thm:bext} demonstrates the parameters are correct. Then we can make use of the same argument in the proof of Theorem~\ref{thm:one_extrac_block} to lift its security to strong OA. 
\end{proof}

Similar to the classical case, we could turn this OA-secure extractor $\QBSAExt$ into an OA-secure strong seeded extractor that works even if the seed only has entropy rate $>1/2$.
However, different from the classical case, we will consider side information generated in the OA model. (and later in the instantiations, the side information could be generated in the GE model). 
In such models, the (weak) seed for the extractor could have quantum side information (in the OA model) that could even be entangled with the side information of the source (in the GE model). 

\begin{corollary}\label{cor:weakseed_q}
For any constant $\delta>0$ there exists a constant $C=\poly(1/\delta)$ such that if there is a quantum strong $(k, \e)$ extractor $\Ext_q: \bits^n \times \bits^d \to \bits^m$ with $d \leq k/C$, then there is another OA-secure strong $(1.2k, \e+2^{-\Omega(d)})$ extractor $\Ext'_q: \bits^n \times \bits^{d'} \to \bits^m$ where $d'=O(d)$ and $\Ext'_q$ works even if the seed only has min-entropy $(1/2+\delta)d'$.
\end{corollary}

\begin{proof}
We note the proof here resembles the one of Corollary~\ref{cor:weakseed}. It suffices to prove the arguments therein extend to the quantum case. 

Firstly, given any quantum weak source (i.e., cq state) $\rho_{RE}, R \in \01^{d'}$ such that $\Hmin(R|E)_\rho\geq (1/2+\delta)d'$, if we divide it into two equal blocks $R=(R_1, R_2)$, then it is $2^{-\Omega(d')}$-close to a quantum $(\delta d', \delta d'/2)$ block source. 
By definition, there exists a $\sigma \in \density{\E}$, such that 
\[
  \rho_{RE} =\sum_r \Prob[R=r] \ketbra{r} \ot \rho^E_r \leq 2^{-(1/2+\delta)d'} \I_R \ot \sigma.
\]
Thus, by taking a partial trace over $R_2$, we have 
\[
  \rho_{R_1E} =\sum_{r_1} \Prob[R_1=r_1] \ketbra{r_1} \ot \rho^E_{r_1} \leq 2^{d'/2}2^{-(1/2+\delta)d'} \I_{R_1} \ot \sigma. 
\]
By definition, we have $\Hmin(R_1|E)_\rho \geq \delta d'$. 
Morever by Lemma~\ref{lem:q_condition}, we have that with probability $1-2^{-\delta d'/2}$ over the fixing of $R_1=r_1$, $\Hmin(R_2|R_1=r_1, E)\geq (1/2+\delta)d'-d'/2-\delta d'/2=\delta d'/2$. Thus $R=(R_1, R_2)$ is $2^{-\Omega(d')}$-close to a quantum $(\delta d', \delta d'/2)$ block source.

Now we can apply Theorem~\ref{thm:qbext} (instead of Theorem~\ref{thm:bext}) to construct $\Ext'_q$. The rest argument remains the same. 
%
\end{proof}

\paragraph{Instantiations.} We can have the following two instantiations of GE-secure extractors, similar to the classical setting. Two points are worth noticing. First, there is no quantum strong seeded extractors like the one of Theorem~\ref{thm:optext}. Instead, we make use of the Trevisan's extractor from Theorem~\ref{thm:Ext1}. 
Second, after obtaining the strong OA-security, we apply Theorem~\ref{thm:SOA_GE} to lift its security to strong GE. 
 
\begin{corollary}\label{cor:weakseed2_q}
For any constant $\delta>0$ there exist constants $C>1$ and $\alpha>0$ such that for any $n, k \in \N$ and $2^{-\alpha k} \leq \e \leq 2^{-C\log^3 n}$ there is an efficient GE-secure strong $(k, \e)$ extractor $\Ext_q: \bits^n \times \bits^d \to \bits^m$ with $d =O(\log^3 (n/\e))$ and $m=0.9k$ that works even if the seed only has min-entropy $(1/2+\delta)d$.
\end{corollary}
\begin{proof}
We instantiate the OA-secure extractor in Corollary~\ref{cor:weakseed_q} with the one from Theorem~\ref{thm:Ext1}. Then we apply Theorem~\ref{thm:SOA_GE} to obtain the GE-security. 
\end{proof}


\begin{corollary}\label{cor:3source_q}
For any constant $\delta>0$ there exist constants $C>1$ and $\alpha>0$ such that for any $n, k \in \N$ and $2^{-\alpha k} \leq \e \leq 2^{-C\log^3 n}$ there is an efficient GE-secure strong extractor $\Ext_q: \bits^d \times \bits^d \times \bits^n \to \bits^m$ with $d =O(\log^3 (n/\e))$ and $m=0.9k$.
Namely,  if $X_1, X_2$ are two independent $(d, \delta d)$ sources and $X_3$ is an independent $(n, k)$ source (all entropy are measured in the GE model), and let $\Adv$ denote the side information in the GE model, then 
\[\trdist{\rho_{\Ext_q(X_1, X_2, X_3) X_1 X_2\Adv}- \uniform{m} \ot \rho_{ X_1 X_2\Adv}} \leq \e.\]
\end{corollary}

\begin{proof}
Let us prove its strong OA security first. (i.e., assume for now all the entropies are measured in the OA model).
Then we claim that $(X_1, X_2)$ forms a $(\delta d, \delta d)$ block source. The reason is that $X_1, X_2$ are independent and the side information can only be from one (or none) of them. 
Note that $\delta d=\Omega(\log^3 n) >\log^2 n$, so we can apply Theorem~\ref{thm:qbext} to construct a strong OA-secure $\Ext_q$ such that the final error is at most $\e$. 
Then we apply Theorem~\ref{thm:SOA_GE} to obtain the GE-security. 
\end{proof}

\section{Application to Privacy Amplification} \label{sec:PA}
Privacy amplification is a basic and important task in cryptography and an important ingredient in quantum key distribution. The setting is that two parties, Alice and Bob share a secret weak random source $X$. Alice and Bob each also has local private random bits. The goal is to convert the shared weak source $X$ into a nearly uniform random string by having the two parties communicating with each other. However, the communication channel is watched by a (passive) adversary Eve, and we want to make sure that eventually the shared uniform random bits remain secret to Eve. In the quantum setting, Eve can also have quantum side information to the shared source $X$.

Strong seeded extractors (and quantum secure strong seeded extractors) can be used to solve this problem in one round by having one party (say Alice) send a seed to Bob and they each apply the extractor to the shared source using the seed. The strong property of the extractor guarantees that even if seeing the seed, Eve has no information about the extracted uniform key. One advantage of this method is that if we have good strong seeded extractors, then we can just use a small seed to extract a long shared key. 

However, as we stated before, it is not clear that we can simply assume that the two parties have local uniform random bits. There may well only have weak sources which may  be subject to (entangled) quantum side information. 
Here we show that as long as the local random sources have arbitrary constant min-entropy rate as measured by our GE model, we can still achieve privacy amplification with asymptotically the same parameters. In particular, this keeps the nice property that we can use a small (weak) seed to extract a long uniform key. 

\subsection*{Privacy Amplification with Local Weak Sources}
We present two scenarios in which we can perform privacy amplification with weak sources. In the first case, only one local random source with entropy rate $>1/2$ is needed. In the second one, two local random sources are needed, however, can be of any constant entropy rate. See Figure~\ref{fig:PA} for details. The correctness of such protocols follow directly from Corollary~\ref{cor:weakseed2_q} and Corollary~\ref{cor:3source_q}. 

\begin{figure}
\begin{protocol*}{Privacy Amplification with One Local Random Source}
\begin{description}
\item Alice and Bob share a weak random source $X$ with entropy at least $k$.  Moreover, Alice has a local random source $Y$ that is independent of $X$ with entropy rate $>1/2$. Both entropies are measured in the GE model. 
\item Let $\Ext_q$ be the extractor from Corollary~\ref{cor:weakseed2_q}.
\end{description}
\begin{step}
\item Alice sends $Y$ to Bob. 
\item Then both parties compute $Z=\Ext_q(X, Y)$, which is their shared randomness. 
\end{step}
\end{protocol*}

\begin{protocol*}{Privacy Amplification with Two Local Random Sources}
\begin{description}
\item Alice and Bob share a weak random source $X$ with entropy at least $k$.  Moreover, Alice has a local independent random source $Y_1$ and 
Bob has a local independent random source $Y_2$. Both are of entropy rate $\delta$ for any constant $\delta>0$. 
All entropies are measured in the GE model. 
\item Let $\Ext_q$ be the extractor from Corollary~\ref{cor:3source_q}.
\end{description}
\begin{step}
\item Alice sends $Y_1$ to Bob. And Bob sends $Y_2$ to Alice. 
\item Then both parties compute $Z=\Ext_q(Y_1, Y_2, X)$, which is their shared randomness. 
\end{step}
\end{protocol*}

\caption{Privacy Amplification with Weak Sources}
\label{fig:PA}
\end{figure}

%
%
%

We remark that in our GE model, the two parties local randomness may even have \emph{entangled} quantum side information with the shared weak source, and we show that even in this case privacy amplification can still be achieved.

%
%

\section{Network Extractor} \label{sec:network}

In the classical setting, network extractors are motivated by the problem of using imperfect randomness in distributed computing, a problem first studied by \cite{GoldwasserSV05}. Kalai, Rao, Li, and Zuckerman formally defined network extractors in \cite{KLRZ}, and gave several efficient constructions for both synchronous networks and asynchronous networks, and both the information-theoretic setting and the computational setting. For simplicity and to better illustrate our ideas, in this paper we will focus on synchronous networks and the information-theoretic setting. We start with formal definitions.

\subsection{Model Definition}

We consider a set $P = [p]$ of $p$ players execute a classical protocol in a synchronized network. Each (honest) player receives an independent source $X_i$, and a side information adversary $\AdvSI$ collects side information $\rho$ from the sources $X = (X_1,\dots, X_p)$ (in certain adversarial models such as OA and GE). We assume each $X_i$ has length $n$ and min-entropy at least $k$ measured in the same way as (\ref{def:GE:entropy}). We call $(X, \AdvSI)$ a $(n,k)$-source for $P$. Formally, such a source is represented by a state
$\rho \in \density{\X_1\ot \cdots \ot \X_p \ot \AdvSI}$. 
We assume $k > C \log p$ for some constant $C>1$ (This is because in distributed computing problems such as Byzantine agreement or leader election, each player needs at least $C \log p$ random bits).


We consider \emph{adaptive} corruption in a full information model, where an all powerful adversary $\AdvNet$ may decide to corrupt a set $\Faulty \subset P$ of up to $t$ players at  any time during the protocol execution, and can perform \emph{rushing} attack to determine the messages of the corrupted players after seeing all communication messages from honest players at each round, potentially with the help of quantum side information generated by $\AdvSI$. At the conclusion of protocol execution, let $T$ denote the transcript of protocol messages that are public, and $Z_i$ be the private output of (honest) player $i$. 

We call $(X,\AdvSI,\AdvNet)$ a $(p,t,n,k)$ network-source-adversary (\NSA) system, and $\AdvSI, \AdvNet$ the adversary for the system. Let $\Adv$ denote all the space that is used by $\AdvSI$ and $\AdvNet$.

The goal of a \emph{network extractor} protocol $\NetExt$ is to let (as many) honest players to extract private uniform randomness  at the conclusion of the protocol when executed on any $(p,t,n,k)$ \NSA system. 
To formally define the security, we need to specify the way that $\AdvSI$ collects side information from $X$ as well as the way that $\AdvNet$ perform rushing attacks. For $\AdvSI$, as before, we consider only the \emph{one-sided adversary} (OA) and the \emph{general entangled} (GE) adversary. For $\AdvNet$, we consider \emph{independent rushing} (IR) adversary and \emph{quantum rushing} (QR) adversary.
\begin{itemize}
\item Independent rushing (IR) adversary:  The rushing messages of the corrupted players depends only on the protocol messages that $\AdvNet$ sees, but \emph{not} on the (quantum) side information $\rho$ collected by $\AdvSI$. This models the situation where the side information $\rho$ is not available during the protocol execution, or the scenario that $\AdvNet$ is classical and cannot process the quantum side information (which can be later used by a quantum distinguisher to distinguish the (private) outputs of the honest players from uniform.
\item Quantum rushing (QR) adversary: The rushing messages can depend on both the protocol messages and the side information collected by both $\AdvNet$ and $\AdvSI$. Moreover, $\AdvNet$ could \emph{simultaneously} manipulate the rushing message and the quantum side information, creating complicated correlations among the protocol messages and the side information. 
This models the situation that the side information $\rho$ is available to $\AdvNet$ at the beginning of protocol execution.
\end{itemize}

Clearly, quantum rushing adversaries are more general and characterize the general power of a fully quantum adversary. On the other hand, the scenario of independent rushing adversary  seems also quite natural and reasonable when the adversary is semi-quantum. Therefore, we consider both settings. 

We note that handling quantum rushing is much more challenging, since it allows protocol messages to depend on the whole side information $\rho$, which in turn depends on all sources $X$. As such, it introduces global correlation among all sources and protocol messages, and destroys the structure of side information. For example, consider that at some point of the protocol, a public seed (or high entropy string) $Y$ is generated and used by a honest player $i$ to extract private uniform randomness from $X_i$ (which is the case for existing protocols). If $Y$ depends on the rushing messages (which is hard to prevent), then $Y$ is correlated with $X_i$ through rushing messages, and thus, extraction has no guarantee to work.

We proceed to define various security notion for network extractors against side information, parametrized by the type of adversaries $\AdvSI$ and $\AdvNet$. Our definition is slightly stronger than the definition in~\cite{KLRZ}, where we guarantees security for a fixed set of players (if they are honest).

\begin{definition} A network extractor $\NetExt$ for $(p,t,n,k)$ \NSA system is \emph{XX-YY secure} for a player set $S \subset P$ with error $\eps$ if for every $(p,t,n,k)$ \NSA system $(X, \Adv_{SI}, \Adv_{Net})$ with XX $\Adv_{SI}$ and YY $\Adv_{Net}$, let $S' = S \backslash \Faulty$, 
$$ \trdist{\rho_{Z_{S'} Z_{-{S'}} T \Adv}-  U \ot \rho_{Z_{-S'} T \Adv}}\leq \eps,$$
where XX $\in \{$ M, OA, GEA $\}$, and YY $\in \{$ IR, QR $\}$. 
\end{definition}

We note that when there is no side information, it suffices to require that a honest player's output $Z_i$ is close to uniform given the transcript $T$, as defined in~\cite{KLRZ}, since conditioned on $T$, $Z_i$ is independent of $X_{-i}$. In the presence of side information, we need to explicitly require $Z_{S'}$ to be close to uniform given $Z_{-{S'}}, T,$ and $\Adv$.

The above definition implies that at the conclusion of the protocol, at least $g = |S|-t$ players obtain secure private uniform randomness. Thus, our definition implies the $(t, g= |S|-t, \eps)$ notion in~\cite{KLRZ}, and has the additional property that the set of successful honest players is fixed before the protocol execution. The KLRZ construction actually satisfies this property. 

To reason about security for a set of players, we also define strong security for an individual player $i$, where we require $Z_i$ to be close to uniform even given other players' input $X_{-{i}}$.

\begin{definition} A network extractor $\NetExt$ for $(p,t,n,k)$ \NSA  system is \emph{strongly XX-YY secure} for a player $i \in P$ with error $\eps$ if for every $(p,t,n,k)$ \NSA  system $(X, \AdvSI, \AdvNet)$ with XX $\AdvSI$ and YY $\AdvNet$ such that $i \notin \Faulty$,  for some uniform distribution $U$
$$ \trdist{\rho_{Z_i X_{-i} T \Adv} - U \ot \rho_{X_{-i} T\Adv}} \leq \eps,$$
where XX $\in \{$ M, OA, GE $\}$, and YY $\in \{$ IR, QR $\}$. 
\end{definition}

The following lemma says that if $\NetExt$ is strongly secure for every $i\in S$, then $\NetExt$ is secure for $S$.

\begin{lemma} \label{lem:net_ext-single-to-set} If $\NetExt$ is a network extractor for $(p,t,n,k)$ \NSA system with strong XX-YY security for every $i \in S$ for some set $S \subset P$, then $\NetExt$ is also XX-YY secure for $S$.
\end{lemma}
\begin{proof}
Let $S' = S \backslash \Faulty = \{i_1,\dots,i_s\}$ and $S'_j = \{i_1,\dots,i_j\}$. We have for every $i\in S'$, 
$$\trdist{\rho_{Z_i X_{-i} T \Adv} -  U \ot \rho_{X_{-i}T \Adv}}\leq \eps.$$ 
We show 
$$\trdist{\rho_{Z_{S'} X_{-S'}T \Adv} - U \ot \rho_{X_{-S'}T\Adv}} \leq |S'|\eps, $$ 
by induction on $j$ for the following statement:
$$\trdist{\rho_{Z_{S'_j} X_{-S'_j}T \Adv} - U \ot \rho_{X_{-S'_j}T\Adv}} \leq j\eps. $$
The base case $S_1$ is trivial. Now, suppose induction holds for $j-1$. That is,
$$\trdist{\rho_{Z_{S'_{j-1}} X_{-S'_{j-1}}T \Adv} - U \ot \rho_{X_{-S'_{j-1}}T\Adv}} \leq (j-1)\eps. $$
Note that $Z_{i_j}$ is a deterministic function of $T$ and $X_j$. We have
$$\trdist{\rho_{Z_{S'_{j-1}} Z_{i_j} X_{-S'_j}T \Adv} - U \ot \rho_{Z_{i_j}X_{-S'_j}T\Adv}} \leq (j-1)\eps. $$
We also have
$$\trdist{\rho_{Z_{i_j} X_{-i_j} T \Adv} -U \ot \rho_{X_{-i_j}T \Adv}} \leq \eps,$$
which implies 
$$\trdist{\rho_{Z_{i_j} X_{-S'_j} T \Adv} -U \ot \rho_{X_{-S'_j}T \Adv}} \leq \eps.$$
Therefore, by triangle inequality, we have 
$$\trdist{\rho_{Z_{S'_j} X_{-S'_j}T \Adv} - U \ot \rho_{X_{-S'_j}T\Adv}} \leq j\eps. $$
\end{proof}

\subsection{Our Results} 

Here we formally state our results for network extractors against side information. For the case of independent rushing, we are able to tolerate close to $1/3$-fraction of faulty players, scarify only roughly $t$ honest players, and extract almost all entropy out even for low entropy $k = \polylog(n)$.





\begin{theorem}[GE-IR-secure Network Extractors] \label{thm:CR-GEA-net-ext} For every constants $\alpha < \gamma \in (0,1)$ and $c > 0$, for sufficiently large $p,t,n,k$ such that $p \geq (3+\gamma) t$ and $k \geq \log^{10}n$, there exists a 3-round network extractor $\NetExt$ for $(p,t,n,k)$ \NSA system with output length $m = k - o(k)$ and a set $S \subset [p]$ of size $|S| \geq p - (1+\alpha) t$ such that $\NetExt$ is GE-IR secure for set $S$ with error $\eps = n^{-c}$.
\end{theorem}

We note that even without side information, Theorem~\ref{thm:CR-GEA-net-ext} is the best known and improves the result of~\cite{KLRZ}. The reasons are that (i) at the time of \cite{KLRZ}, they did not have Li's extractor for a constant number of weak sources with min-entropy $k=\polylog(n)$ \cite{Li13b}, and (ii) we additionally use alternating extraction to extract almost all entropy out.




For the case of quantum rushing, we obtain slightly worse parameters, where we can tolerate a constant fraction of faulty players, and scarify $O(t)$ honest players. Here we require the min-entropy $k$ to be sufficiently larger then $t$. We discuss at the end of the section how to relax this requirement.

\begin{theorem}[GE-QR-secure Network Extractor] \label{thm:QR-GEA-net-ext}  There exists a constant $\gamma \in (0,1)$ such that for every constant $c > 0$, for sufficiently large $p, t, n, k$ with $p > t / \gamma$ and $k \geq \max \{ \log^{10}n, t/\gamma \}$, there exists a network extractor $\NetExt$ for $(p,t,n,k)$ 
\NSA system with output length $m = \Omega(k)$ and a set $S \subset [p]$ of size $|S| \geq p - t/\gamma$ such that $\NetExt$ is GE-QR secure for set $S$ with error $\eps = n^{-c}$.
\end{theorem}

\subsection{Security Lifting Lemmas for Network Extractors} \label{sec:CR_QR}
In this section we present two security equivalence/lifting tools in the context of network extractors, which we consider as one of our main contributions in this paper. 
The first one is about the equivalence between the strong OA security and the strong GE security in the context of network extraction, which is an analogue of their equivalence in the multi-source extraction. 
However, to make it work, our argument relies on the fact that protocols only have independent rushing but no quantum rushing. 
The second one is a new tool to connect the IR security to the QR security, by another simulation argument which suffers a certain amount of loss of parameters. We start with the OA-GE equivalence as follows. 


\begin{theorem} \label{thm:OA-GE-eq-for-net-ext}
If $\NetExt$ is strongly OA-IR-secure for a player $i \in [p]$ with error $\eps$, then $\NetExt$ is GE-IR-secure for $i$  with error $\eps$. 
\end{theorem}

\begin{proof} 
The proof of theorem is quite similar to the one of Theorem~\ref{thm:SOA_GE}. Thus, we only provide a sketch here and highlight the difference. 
Given any GE source in the network extraction context, for each $i \in S$, one can perform exactly the same first step in the proof of Theorem~\ref{thm:SOA_GE} by working at an imaginary step after the leakage from the source $X_i$ but before any leakage from $X_{-i}$ happens. At that step, one obtains an OA source, and can apply $\NetExt$ because it is strongly OA-IR-secure. 

The second step is slightly different, where we need to make crucial use of the fact that the protocol only allows IR, which makes the operation $\NetExt$ commute with all leaking operations $\Phi_i$ on the source. 
In contrast, if there were QR, then such commutativity is violated and we cannot proceed with our current technique. 
Then we can follow the original argument to make use of the fact that $\NetExt$ is strongly OA-IR-secure and safely convert the OA source at the imaginary step to the final GE source. 
%
%
%
\end{proof}

\vspace{1mm} Now we switch to dealing with quantum rushing. To that end, we need to formally define the possible correlations that could be generated between classical and quantum systems during the execution of the protocol. 
However, our simulation idea is so general that we don't want to restrict to a very specific protocol design in discussion. 
Thus, we formulate a relatively general model as below which will fit our use in the later analysis for specific protocols, and at the same time serve as an intuitive model to understand independent rushing, quantum rushing and our idea to bridge them. 


Imagine a ccq state $\rho_{XY\Adv} \in \density{\X \ot \Y \ot\Adv}$ where $\X, \Y$ are classical. In a real protocol execution, this state $\rho$ could represent the system at some point. Moreover, let $Y$ be the public message and $X$ be some private information. 
Now imagine a rushing message $Y_R$ and a function $E: \X \times \Y \times \Y_R \rightarrow \Z$ that could be the output of the protocol.  
The difference between independent rushing and quantum rushing can be formulated as
\begin{itemize}
  \item (IR): the rushing message $Y_R$ is only a function of the public $Y$, i.e., $Y_R=Y_R(Y)$. The correlation between $X, Y$ and the quantum part $\Adv$ remains the same. Only some new purely classical correlation is established between $X,Y$ and $Y_R$. 
  \item (QR): the rushing message $Y_R$ is generated by an admissible operation on both $\Y$ and $\Adv$. Precisely, let $\Phi_q: \lin{\Y \ot \Adv} \rightarrow \lin{\Y \ot \Y_R \ot \Adv}$ be a $Y$-controlled admissible operation that captures the quantum rushing strategy. Thus, after the quantum rushing, the whole system becomes, 
  \[
    \rho_{XYY_R \Adv} =\Phi_q (\rho_{XY\Adv}).
  \]
  As a result, the correlation between $X,Y$ and the quantum part $\Adv$ could be completely changed. 
\end{itemize}

In the context of randomness extraction, we care about whether the output $Z=E(X,Y, Y_R)$ is close to uniform against $\Adv$. Let us denote its distance by $\eps$. 
By the following simulation argument, if $Z$ is $\eps$ close to uniform against $\Adv$ when subject to any IR attack, then it is $2^m\eps$ close to uniform against $\Adv$ when subject to any QR attack, where $m=|Y_R|$. 

\begin{theorem}[IR to QR] \label{thm:CR_QR}
For any $\rho_{XY\Adv}$ system described above, let $Z \in \01^l$ and $Y_R \in \01^m$. If for any IR attack,
\[
   \trdist{\rho_{E(X,Y,Y_R)Y\Adv}-\uniform{l} \ot \rho_{Y\Adv}} \leq \eps,
\]
then for any QR attack, we have 
\[
   \trdist{\rho_{E(X,Y,Y_R)Y\Adv}-\uniform{l} \ot \rho_{Y\Adv}} \leq 2^m\eps.
\]
\end{theorem}

\begin{proof} Let $\Phi_q$ be a quantum rushing operation that defines a QR attack.

For any $r \in \01^m$, we can consider an IR attack that set $Y_R=r$ deterministically. By the premise of the theorem, we have
\[
   w_r\defeq \trdist{\rho_{E(X,Y, r)Y\Adv}-\uniform{l} \ot \rho_{Y\Adv}} \leq \eps.
\]
We can apply $\Phi_q$ on both sides:
\[
\trdist{\rho_{E(X,Y, r)Y_RY\Adv}-\uniform{l} \ot \rho_{Y_RY\Adv}}\leq \eps. 
\]
Define 
\[
  u_{rr'} \defeq \trdist{\rho_{E(X,Y, r){y_R=r'}Y\Adv}-\uniform{l} \ot \rho_{{y_R=r'}Y\Adv}} .
\]
Note that $\rho_{E(X,Y, r){y_R=r'}Y\Adv}$ is a sub-normalized state and $\sum_{r'} \rho_{E(X,Y, r){y_R=r'}Y\Adv}=\rho_{E(X,Y, r)Y_RY\Adv}$. Thus, it is easy to see that $\sum_{r' \in \01^m} u_{rr'} \leq w_r \leq \eps, \forall r \in \01^m$. Finally, observe that when $r'=r$, the classical part and the quantum part have the correct correlation after the QR attack, and thus, 
\[
    \trdist{\rho_{E(X,Y,Y_R)Y\Adv}-\uniform{l} \ot \rho_{Y\Adv}}\leq \sum_{r \in \01^m} u_{rr} \leq \sum_{r,r' \in \01^m} u_{rr'} \leq 2^m \eps. 
\]
\end{proof}

The above theorem provides an important tool to handle QR attacks. However, this technique incurs a significant loss in parameters and using this technique alone would fail to handle the QR setting for known protocols. We shall address the additional issues and provide our solutions in Section~\ref{sec:QR}

As a final remark of the two theorems in this section, we shall first apply Theorem~\ref{thm:OA-GE-eq-for-net-ext} to lift the OA-IR security to the GE-IR security as our simulation technique there does not handle QR. Then we apply Theorem~\ref{thm:CR_QR} together with the ideas from Section~\ref{sec:QR} to lift the GE-IR security to the GE-QR security. 


\subsection{Combinatorial and Extractor Tools}
Before moving to the construction and the analysis of our protocol, we briefly review a few combinatorial tools that will be used later. 
First, we shall need the concept of an \emph{$\AND$-disperser} defined in \cite{KLRZ}:

\begin{definition}
[$\AND$-disperser] An $(l,r,d,\delta,\gamma)$ $\AND$-disperser is a bipartite graph with left vertex set $[l]$, right vertex set $[r]$, left degree $d$ s.t. for every set $V \subset [r]$ with $|V|=\delta r$, there exists a set $U \subset [l]$ with $|U| \geq \gamma l$ whose neighborhood is contained in $V$.
\end{definition}

The following lemma is proved in \cite{KLRZ}.

\begin{lemma}[$\AND$-disperser] \label{lem:ANDdisperser} 
There exists a constant $c>0$ such that if $D=o(\log M)$ then for every constant $0<\alpha<1$ and large enough $M$, there exists an explicit construction of an $(N,M,D,\alpha,\beta)$ AND-disperser $G$ such that $M < N \leq Md^D$ and $\beta >\mu^D$. Here $d = c\alpha^{-8}, \mu =\alpha^2/3.$
\end{lemma}

Another well studied object that we need is a construction of a bipartite expander.

\begin{definition}
[Bipartite Expander] A $(l,r,d,\beta)$ bipartite expander is a bipartite graph with left vertex set $[l]$, right vertex set $[r]$, left degree $d$ and the property that for any two sets $U \subset [l], |U| = \beta l$ and $V \subset [r], |V| = \beta r$, there is an edge from $U$ to $V$.
\end{definition}

Pippenger proved the following theorem:

\begin{theorem}
[Explicit Bipartite Expander \cite{Pippenger87, LubotzkyPS88}]\label{thm:expander} For every $\beta > 0$, there exists a constant $d(\beta) < O(1/\beta^2)$ and a family of polynomial time constructible $(l,l,d(\beta),\beta)$ bipartite expanders.
\end{theorem}

We will also need to use the following extractor for a special type of sources. 

\begin{theorem} [General Source vs Somewhere random source with few rows Extractor \cite{BarakRSW06}]
There exist constants \label{thm:srgeneral}
  $\alpha, \beta < 1$ such that for every $n,k(n)$ with $k >
\log^{10}n,$ and constant $0< \gamma < 1/2$, there is a polynomial
time computable function $\SRExt:\{0,1\}^{n} \times
\{0,1\}^{k^{\gamma +1}} \rightarrow \{0,1\}^m $ s.t. if $X$ is an
$(n,k)$ source and $Y$ is a $(k^{\gamma} \times k)$-SR-source,\footnote{Here, we view $Y$ as $k^{\gamma}$ rows of strings of length $k$, and $Y$ is a $(k^{\gamma} \times k)$-SR-source if there exist a marginally uniform row in $Y$.}

\[ | (Y , \SRExt(X,Y)) - (Y , U_m) | < \epsilon \]
and
\[ | (X , \SRExt(X,Y)) - (X , U_m) | < \epsilon \]
where $U_m$ is independent of $X,Y$, $m = k-k^{O(1)}$ and
$\epsilon = 2^{-k^{\alpha}}$.
\end{theorem}

\subsection{Our Network Extractor for the Independent Rushing Case} \label{subsec:CR-net-ext}

We construct our network extractors for the independent rushing case and prove Theorem~\ref{thm:CR-GEA-net-ext} in this section. Note that by Theorem~\ref{thm:OA-GE-eq-for-net-ext} and Lemma~\ref{lem:net_ext-single-to-set}, it suffices to construct a strongly OA-IR secure network extractor. Our construction follows the construction in~\cite{KLRZ}, but lift the marginal security to OA security. Along the way, we obtain a simpler construction that improves several aspects of the KLRZ network extractors by using improved independent source extractor of Li~\cite{Li13b}, and an alternate extraction idea.

\begin{lemma}[Strong OA-IR Network Extractors] \label{lem:CR-OA-net-ext} For every constants $\alpha < \gamma \in (0,1)$ and $c > 0$, for sufficiently large $p,t,n,k$ such that $p \geq (3+\gamma) t$ and $k \geq \log^{10}n$, there exists a network extractor $\NetExt$ for $(p,t,n,k)$ \NSA system with output length $m = k - o(k)$ and a set $S \subset [p]$ of size $|S| \geq p - (1+\alpha) t$ such that $\NetExt$ is strongly OA-IR secure for every $i \in S$ with error $\eps = n^{-c}$.
\end{lemma}

At a high level, to lift the security, we simply replace the extractor in the last step of~\cite{KLRZ} by a strongly OA-secure one with a similar idea appeared in Section~\ref{sec:bOA:one_source}. More precisely, the construction of~\cite{KLRZ} can be viewed in two steps, where the first step generates a public high min-entropy source $Y$, which is used by each honest player $i$ in the second step to extract private uniform randomness $Z_i = \Ext(X_i,Y)$ using some $Y$-strong randomness extractor $\Ext$. We show that if $\Ext$ is strongly OA-secure for $X_i$, then the network extractor is strongly OA-IR secure for player $i$.
We proceed to present our (simplified) construction in steps as follows. 
\begin{itemize}
\item In step 1, we construct a three-round sub-protocol $\PubExt$ that outputs a public two-block source $y = (y^1, y^2)$ with marginal security using $\AND$-dispersers (Lemma~\ref{lem:ANDdisperser}), expanders (Theorem~\ref{thm:expander}), BRSW extractors (Theorem~\ref{thm:srgeneral}), and improved independent source extractors (Theorem~\ref{thm:Li_IExt}). 
\item In step 2, each honest player $i$ uses $y$ to extract uniform randomness from $x_i$ using a $y$-strong OA-secure randomness extractor. 
\end{itemize}

Let $\delta = (\gamma - \alpha)/4$. Throughout the protocol, we partition players into three disjoint sets $P = A \cup B \cup C$ of size $|A| = (1+\alpha)\cdot t$, $|B| = 2\cdot (1+2\delta) \cdot t$, and $|C| = p - |A| - |B|$.



\paragraph{Step 1. Obtain a public block source with marginal security.} In this step, we construct a $\PubExt$ sub-protocol that outputs a public two-block source $y = (y^1, y^2)$ with (marginal) entropy rate $> 0.5$ in both blocks. A formal description of the $\PubExt$ protocol can be found in Figure~\ref{fig:step1-net-ext}. Note that only marginal security is required here. We prove the following lemma by adapting the analysis of~\cite{KLRZ}. The proof explains the intuition behind the construction.

\begin{lemma} \label{lem:step-1-net-ext} For every $(p,t,n,k)$ \NSA system $(X,\AdvSI, \AdvNet)$ with OA $\AdvSI$ and IR $\AdvNet$, there exists a set $B_{\Good} \subset B \backslash \Faulty$ of size at least $|B_{\Good}| \geq (1/2 + \delta/4) \cdot |B| +1$ such that at the conclusion of $\PubExt$, for every $j \in B_{\Good}$, we have $(Y_j, T_1) \approx_{\eps_1+\eps_2} (U_{m_2}, T_1)$, where $T_1$ denotes the transcript of the first round.
\end{lemma}







\begin{proof} 
Let $A_{\Faulty}$ and $B_{\Faulty}$ denotes the sets of faulty players in $A$ and $B$, respectively. Since $|A_{\Faulty}| \leq t$, by the property of the $\AND$-disperser, there exists a good set $V \subset [N]$ of size $|V| \geq \beta_1 N$ such that the neighbors of $V$ in $G$ are contained in $A \backslash A_{\Faulty}$. Thus, for every $v\in V$ with neighbors $i_1,\dots,i_{d_1}$, $S_v = \IExt(X_{i_1},\dots, X_{i_{d_1}})$ is $\eps_1$-close to uniform. Let $B_{\Bad}$ be the set of left vertices $j \in H$ such that all neighbors of $j$ are outside $V$. By the property of the expander, we have $|B_{\Bad}| \leq \beta_2 N \leq \delta t$. Let $B_{\Good} = B \backslash (B_{\Faulty} \cup B_{\Bad})$. We have $|B_{\Good}| \geq |B| - t - \delta t \geq (1/2+\delta/4)|B|$. By definition, for every $j \in B_{\Good}$, $j$ is a honest player and $j$ has a neighbor in $V$. Thus, $S^j$ is $\eps_1$-close to a somewhere random source, and $Y_j = \SRExt(X_j,S^j)$ is $(\eps_1+\eps_2)$-close to uniform given $S^j$. Finally, note that given $S^j$, $Y_j$ is independent of $T_1$. Therefore, 
$(Y_j,T_1) \approx_{\eps_1+\eps_2} (U_{m_2}, T_1)$.
\end{proof}

\begin{figure}
\begin{protocol*}{Protocol $\PubExt$: Obtain a public block source with marginal security.}
\begin{description}
\item Protocol Input: Private weak sources $x_i$'s of players $i$ in sets $A$ and $B$.
\item Protocol Output: A public block source $y = (y^1, y^2) \in \zo^{|B|\cdot \sqrt{k}+|B|\cdot \sqrt{k}}$.
\item Sub-Routines and Parameters:
\begin{step}
\item Let $\IExt$ be a $(d_1,n,k,m_1,\eps_1)$ independent source extractor with some constant $d_1$ and $\eps_1 \leq n^{-2c}$ from Theorem~\ref{thm:Li_IExt}.
\item Let $G$ be an explicit $\AND$-disperser with parameters $(N, M=|A|, d_1, \alpha_1 = (\alpha/(1+\alpha)), \beta_1)$ from Theorem~\ref{lem:ANDdisperser}, where $M \leq N \leq M \cdot \poly(\alpha_1^{-d_1})$ and  $\beta_1 \geq \poly(\alpha_1^{d_1})$.
\item Let $H$ be an explicit bipartite expander with parameters $(N,N, d_2, \beta_2)$ from Theorem~\ref{thm:expander}, where $\beta_2 = \min \{ \beta_1, \delta t/ N\}$, and $d_2 = O(1/\beta_2^2)$.
\item Let $\SRExt$ be the BRSW extractor from Theorem~\ref{thm:srgeneral} with error parameter $2^{-k^{\Omega(1)}}$ and output length $m_2 \geq \sqrt{k}$. 
\end{step}
\item Round $1$.
\begin{step}
\item Every player $i \in A$ sends his source $x_i$ to all the players in $B$.
\end{step}
\item Round $2$ and $3$.
\begin{step}
\item Identify $A$ with the right vertex set of $G$. Identify $B$ with (arbitrary subset of) left vertex set of $H$. Identify right vertex set of $H$ with left vertex set of $G$. 
\item For each left vertex $v \in [N]$ in $G$, let $i_1,\dots, i_{d_1}$ be its neighbors. Define $s_v = \IExt( x_{i_1},\dots, x_{i_{d_1}})$.
\item For $j \in B$, let $v_1, \dots, v_{d_2}$ be his neighbors in $H$. Let $s^j = (s_{v_1},\dots, s_{v_{d_2}})$. Player $j$ computes $y_j = \SRExt(x_j, s^j)$ and output the first $\sqrt{k}$ bits as $y^1_j$ in round $2$ and and the next $\sqrt{k}$ bits as $y^2_j$ in round $3$.
\item The public outputs $y^1$ and $y^2$ are concatenation of $y^1_j$ and $y^2_j$ for $j\in B$, respectively. 
\end{step}
\end{description}
\end{protocol*}
\caption{Step 1 of our GE-IR secure network extractor protocol.}
\label{fig:step1-net-ext}
\end{figure}

The above lemma readily implies the following technical statement, which says that the output $(Y^1,Y^2)$ forms a block source even given the transcript $T_1$. 

\begin{lemma} \label{lem:step-1-net-ext-conclusion} For every $(p,t,n,k)$ \NSA system $(X,\AdvSI, \AdvNet)$ with OA $\AdvSI$ and IR $\AdvNet$, at the conclusion of $\PubExt$, the output $(T_1, Y^1, Y^2)$ is $\eps'$-close to a block source with entropy rate at least $1/2 + \delta/4$ in second and third blocks (i.e., $Y = (Y^1, Y^2)$ is a two-block source even conditioned on $T_1$), where $\eps' = |B| \cdot (\eps_1+\eps_2)$.

Furthermore, for every $j\in B \backslash \Faulty$, let $Y_{-j} = (Y^1_{-j}, Y^2_{-j})$ be the two-block string $Y$ with the $j$-th component removed. $(T_1, Y^1_j, Y^2_j, Y^1_{-j}, Y^2_{-j})$ is a block source with entropy rate at least $(1/2 + \delta/4)$ for the last two blocks. 
\end{lemma}
\begin{proof}By Lemma~\ref{lem:step-1-net-ext}, there exists a set $B_{\Good} \subset B \backslash \Faulty$ of size at least $|B_{\Good}| \geq (1/2 + \delta/4) \cdot |B| + 1$ such that at the conclusion of $\PubExt$, for every $j \in B_{\Good}$, we have $(Y_j, T_1) \approx_{\eps_1+\eps_2} (U_{m_2}, T_1)$.  Note that conditioned on $T_1$, $\{Y_j\}_{j\in B_{\Good}}$ are mutually independent. For notational convenience, let $\bar{B} = B \backslash B_{\Good}$, and let $Y_{B_{\Good}}$ denote $\{Y_j\}_{j\in B_{\Good}}$. By a standard hybrid argument, we have $(Y_{B_{\Good}}, T_1) \approx_{\eps'} (U_{|B_{\Good}|\cdot m_2}, T_1)$. Therefore, up to a $\eps'$ statistical error, we can switch to a hybrid where $Y_{B_{\Good}}$ is uniform given $T_1$. (Moe precisely, we can define a hybrid experiment where $Y_{B_{\Good}}$ is perfectly uniform and independent of $T_1$, and the real experiment is $\eps'$ close to the hybrid experiment.)

In this hybrid, since $|B_{\Good}| \geq (1/2+\delta/4)\cdot |B|+1$, $Y^1$ has entropy rate at least $1/2+\delta/4$ given $T_1$. Also, note that $Y^2_{B_{\Good}}$ is uniform given both $T_1$ and $Y^1_{B_{\Good}}$, and since $Y^1$ and $Y^2$ are released in different round, $Y^1_{\bar{B}}$ can only depend on $Y^1_{B_{\Good}}$ and $T_1$, but independent of $Y^2_{B_{\Good}}$. Thus,  $Y^2$ has entropy rate at least $1/2+\delta/4$ given $Y^1$ and $T_1$. It follows that in this hybrid, $(T_1, Y^1, Y^2)$ is $\eps'$-close to a block source with entropy rate at least $1/2 + \delta/4$ in second and third blocks, which proves the first statement of the lemma.

The ``furthermore'' part of the lemma follows by the same argument and noting that the $B_{\Good} \backslash \{j\}$ components already provide sufficient entropy.
\end{proof}

\paragraph{Step 2. Extract OA-secure private uniform randomness using $y$.} In this step, each honest player in $B \cup C$ simply uses a $Y$-strong OA-secure two-block+general extractor from Theorem~\ref{thm:2-block-general-ext} to extract private uniform randomness (there is no interaction). A formal description of the $\PriExt$ protocol can be found in Figure~\ref{fig:step1-net-ext}.


We show that the output is OA-IR secure for every player $i \in B \cup C$. To see this, let us consider a honest player $i\in C$. Note that an OA $\AdvSI$ can only get side information from one source. Let us first consider the case that $\AdvSI$ gets side information $\rho_i$ from $X_i$. In this case, by the $Y$-strong OA-security of $\OAExt$, $Z_i$ is close to uniform given both $Y$ and $\rho_i$. Now, note that conditioned on $Y$, $Z_i$ is independent of $X_{-i}$ and transcript $T$. Therefore, $Z_i$ is close to uniform even given $(X_{-i}, T, \rho_i)$. Similarly, for the case that $\AdvSI$ gets side information $\rho_j$ from some $X_j$ for $j \neq i$, $Z_i$ is close to uniform given $Y$, and contitioned on $Y$, $Z_i$ is independent of $X_{-i}, T$, and $\rho_j$. Thus, $Z_i$ is close to uniform given $(X_{-i}, T, \rho_i)$. The analysis generalizes to handle players $j\in B$ by additionally condition on $T_1$ and $(Y^1_{j}, Y^2_j)$. 

\vspace{1mm} 
\begin{proof}(of Lemma~\ref{lem:CR-OA-net-ext}) We consider $\NetExt$ that execute $\PubExt$ and $\PriExt$ sub-protocols in order, and the set $S = B \cup C$. We show that $\NetExt$ is OA-IR secure for every $i\in S$ with error $\eps \leq n^{-c}$. 


Let $(X,\AdvSI, \AdvNet)$ be a $(p,t,n,k)$ \NSA system with OA $\AdvSI$ and IR $\AdvNet$. Let us first consider a honest player $i \in C$. 
By Lemma~\ref{lem:step-1-net-ext-conclusion}, at the conclusion of $\PubExt$, $(T_1, Y^1, Y^2)$ is $\eps'$-close to a block source with entropy rate at least $1/2 + \delta/4$ in second and third blocks.
Thus, up to a $\eps'$ error in the trace distance, we can switch to a hybrid where the condition holds with no error.

Suppose the OA $\AdvSI$ chooses to only get side information $\rho_i$ from $X_i$. Note that $Y$ is independent of $(X_i,\rho_i)$, and $X_i$ has $k$-bits of entropy given $\rho_i$. By strong OA-security of $\OAExt$, we have $\trdist{\rho_{Z_iY\Adv}-\uniform{m} \ot \rho_{Y\Adv}}\leq \eps_3$
(note that $\Adv$ denotes the space of $(\AdvSI, \AdvNet)$, and here it refers to the side information space)
. Also note that given $Y$, $Z_i$ is independent of $X_{-i}$ and transcript $T$. Therefore,
$$\trdist{\rho_{Z_i X_{-i}T \Adv} - \uniform{m} \ot \rho_{X_{-i}T\Adv}} \leq \eps_3.$$

Similarly, suppose the OA $\AdvSI$ chooses to get side information for $\rho_{i'}$ from $X_{i'}$ for some $i' \neq i$. Note that $Y$ is independent of $X_i$, and $X_i$ has $k$-bits of entropy. By strong OA-security of $\OAExt$, we have 
$\trdist{\rho_{Z_iY\Adv}-\uniform{m} \ot \rho_{Y\Adv}}\leq \eps_3$. Also note that given $Y$, $Z_i$ is independent of $X_{-i}$, transcript $T$, and side information $\rho_{i'}$. Therefore,
$$\trdist{\rho_{Z_i X_{-i}T \Adv} - \uniform{m} \ot \rho_{X_{-i}T\Adv}} \leq \eps_3.$$

Now, let us consider a honest player $j\in B$. Again by Lemma~\ref{lem:step-1-net-ext-conclusion}, at the conclusion of $\PubExt$, $(T_1, Y^1_j, Y^2_j, Y^1_{-j}, Y^2_{-j})$ is a block source with entropy rate at least $(1/2 + \delta/4)$ for the last two blocks. Thus, up to an $\eps'$ error in trace distance, we can switch to a hybrid where the condition holds with no error. In what follows, we perform our analysis conditioned on $H = (T_1,Y^1_j, Y^2_j)$. 


Suppose the OA $\AdvSI$ gets side information $\rho_j$ from $X_j$. Note that given $H$ and $\rho_j$, $X_j$ has at least $k - 2\sqrt{k} = k - o(k)$ bits of min-entropy, and is independent of $Y_{-j} = (Y^1_{-j}, Y^2_{-j})$, which is a two block source with at least $1/2+\delta/4$ entropy rate per block. By OA-security of $\OAExt$, we have 
$|\rho_{Z_jY_{-j}H\Adv} - U_{m} \ot \rho_{Y_{-j} H \Adv} | \leq \eps_3$. 
Also note that given $Y_{-j}, H$, $Z_j$ is independent of $X_{-j}$ and $T$. Thus, 
$$\trdist{\rho_{Z_jX_{-j}T\Adv}-\uniform{m} \ot \rho_{X_{-j}T\Adv}} \leq \eps_3. $$

For the final case that the OA $\AdvSI$ gets side information $\rho_{j'}$ from $X_{j'}$ for some $j' \neq j$, by the same argument and OA-security of $\OAExt$, we have 
$\rho_{Z_j Y_{-j} H} -  U_{m} \ot \rho_{Y_{-j} H} | \leq \eps_3$. 
Again note that given $Y_{-j}, H$, $Z_j$ is independent of $X_{-j}$, $T$, and $\rho_{j'}$. Thus, 
$$\trdist{\rho_{Z_jX_{-j}T\Adv}-\uniform{m} \ot \rho_{X_{-j}T\Adv}} \leq \eps_3. $$
\end{proof}



\begin{figure}
\begin{protocol*}{Protocol $\PriExt$: Extract OA-secure Private Uniform Randomness.}
\begin{description}
\item Protocol Input: Private weak sources $x_i \in \zo^n$ of players $i$ in sets $B$ and $C$. Public two-block source $y = (y^1,y^2) \in \zo^{|B|\sqrt{k}+|B|\sqrt{k}}$. 
\item Protocol Output: A private string $z_i \in \zo^m$ for each player $i \in B\cup C$.
\item Sub-Routines and Parameters:
\begin{step}
\item Let $\OAExt(X,Y)$ be a $Y$-strong OA-secure two-block+general source extractor from Theorem \ref{thm:2-block-general-ext} with output length $m = k-o(k)$ and error $\eps_3 \leq 2^{-\Omega(k^{\Omega(1)})}$. 
\end{step}
\item The protocol has no interaction. Each player $i\in B \cup C$ generates a private output $z_i$.
\begin{step}
\item For every $i \in C$, player $i$ computes $z_i = \OAExt(x_i,y)$ and output $z_i$.
\item For every $j \in B$, let $y_{-j} = (y^1_{-j}, y^2_{-j})$ be the two-block string $y$ with the $j$-th component removed. Player $j$ computes $z_j = \OAExt(x_j,y_{-j})$ and output $z_j$.
\end{step}
\end{description}
\end{protocol*}
\caption{Step 2 of our GE-IR secure network extractor protocol.}
\label{fig:step2-net-ext}
\end{figure}

\subsection{Our Network Extractor for the Quantum Rushing Case} \label{sec:QR}

In this section, we discuss how to deal with quantum rushing (QR) adversaries and present our GE-QR secure network extractor. Recall that it means the protocol adversary $\AdvNet$ is allowed to operate on the quantum side information collected by $\AdvSI$ to produce rushing messages for faulty players. This is clearly more general, and it turns out that this setting is very different from the IR adversary setting, and much more challenging to handle, as explained as follows.

We first note that whether OA and GE security are equivalent is no longer clear in the QR setting, and even if it's true, it seems unlikely to be proven by existing techniques. Recall that in the proof of the equivalence in the IR setting, we crucially rely on the fact that the side information can be collected \emph{after} the protocol execution. This is no longer true in the QR setting, since the side information is used by $\AdvNet$ during the protocol execution and the operations are not commute in general. 

Secondly, even getting OA-QR security seems already challenging. To see the issues, for example, consider our network extractor in Figure~\ref{fig:step1-net-ext} and~\ref{fig:step2-net-ext}. There, a public high min-entropy source $Y$ is generated in $\PubExt$ protocol, which is used by each honest player $i$ in the second step to extract private randomness from his source $X_i$ using a $Y$-strong OA-secure randomness extractor $\OAExt$. Now, suppose that $Y$ depends on some rushing information, which in turn can be correlated with the side information $\rho_i$ of $X_i$ collected by $\AdvSI$. As such, it can create correlation between $Y$ and $X_i$ and the extractor $\OAExt$ cannot be guaranteed to work. Indeed, by corrupting different set of players, $\AdvNet$ can create such correlation for every bit of $Y$.

Thus, a natural approach is to avoid such correlation. Note that if a message $y$ depends only on honest players, then it is not subject to rushing attack. For example, consider a simple solution that we group players into $s = p/d$ groups of size $d$, and apply a quantum-secure multi-source extractor $\QMExt$ to extract private randomness for each group. Since there are only $t$ faulty players, at least $s-t$ groups contains only honest players. It can be shown that if we use GE-secure multi-source extractor, then the outputs of honest groups are GE-QR secure. However, note that to compute $\QMExt(x_1,\dots,x_d)$, $d-1$ players need to send their inputs to the remaining player, so these $s\cdot (d-1) = (1-1/d)\cdot p$ players cannot hope to obtain private randomness. 

One can do better by letting $s < p/d$ groups publish uniform seeds extracted by $\QMExt$, and let the remaining players choose one seed to extract their private randomness (distributed evenly since $t$ out of $s$ groups can be faulty). It can be shown that if the seeded extractor in use is OA-secure, then the output is GE-QR secure when both the group and the player are honest. By setting $s = \Theta(\sqrt{pt})$, we can ensure that at least $p - O(\sqrt{pt})$ players obtain  uniform private output. However, we still lose $O(\sqrt{pt})$ players, which is much worse than losing a small $O(t)$ players in the IR setting when $t = o(p)$. Furthermore, for players using seeds from faulty groups, they may generate far from uniform output without knowing their failure, which can be devastated for cryptographic applications. It is not clear to us if we can get around these issues if we only rely on non-rushing messages.

\paragraph{Our Approach.} Our key idea here is a simulation-based security lifting technique (from IR to QR) that allows us to handle a limited amount of quantum rushing correlation, which we already elaborate on in Section~\ref{sec:CR_QR}. 
However, as we explained before, this technique alone fails to resolve the quantum rushing issue. 

 To illustrate the idea, let us consider the following construction. Let us again have $s$ groups publish $s$ uniform seeds $y_1,\dots, y_s$ extracted from $\QMExt$, and concatenate short slices from each seed to obtain a public source $y$ (as in $\PubExt$ in Fig.~\ref{fig:step1-net-ext}), which is used to extract randomness for the remaining player $i$ from their private $x_i$ using a OA-secure two-source extractor $\QTExt$.

Let $y = (y_{\Good}, y_{\Bad})$ where $y_{\Good}$ (resp., $y_{\Bad}$) are the components from honest (resp., faulty) groups. Since $y_{\Good}$ is from honest group, it is uniform and independent of $x_i$, which also implies $y$ has good amount of min-entropy. However, $y_{\Bad}$ is subject to quantum rushing and can depend on both $y_{\Good}$ and the side information $\rho$ that depends on $x_i$, and thus, $x_i$ and $y$ are not independent. Nevertheless, such quantum rushing correlation is limited to the $y_{\Bad}$ part, which can be a small fraction of $y$ if $s$ is sufficiently larger than $t$. Also, note that if only independent rushing is allowed (i.e., $y_{\Bad}$ can only depend on $y_{\Good}$, but not the side information), then $y$ remains independent of $x_i$ and thus the extraction works as long as $\QTExt$ is $y$-strong and OA-secure.

Our idea now is to break the quantum rush correlation by the simulation idea in Theorem~\ref{thm:CR_QR}, which guesses the value of $y_{\Bad}$ and only looks at the situation when the guess value matches the real value. 
As a result, it occurs a $2^{|y_{\Bad}|}$ factor loss in the error, however, reduces any correlation generated by quantum rushing to a correlation generated only by independent rushing. 
One still needs to prove the IR security of the protocol, which is a simpler task than directly proving the QR security.  The caveat is, however, that one needs to be able to afford 
the $2^{|y_{\Bad}|}$ blow-up in the error parameter.


In the above construction, we have errors from both $\QMExt$ and $\QTExt$ extractors, where $\QMExt$ has large error $1/\poly(n)$, which we cannot afford. Fortunately, note that $\QMExt$ is used to generate $y_{\Good}$ from honest groups, which is not subject to rushing. Thus, we can switch to a hybrid where $y_{\Good}$ is actually uniform, and avoid paying the $2^{|y_{\Bad}|}$ blow-up for the $\QMExt$ error. On the other hand, we have two-source extractors $\QTExt$ with exponentially small error in the smaller entropy of the two sources. If $k$ is sufficiently large (compared to $t$), then we can set $s$ to be a sufficiently large $O(t)$ so that $y_{\Bad}$ is a sufficiently small fraction of $y$ and the blow-up is affordable. This leads to a GE-QR secure network extractor that lose only $O(t)$ honest players and ensure private uniform randomness for every players with outputs, resolving the issues from the above naive approach.

On the other hand, for the $k < t$ case, we cannot afford the blow-up since $|y_{\Bad}|$ is at least $t$ but the extractor error is at least $2^{-k}$. For clarity of exposition, we defer discussion about how to handle $k<t$ case in later sections. In what follows, we formalize the above construction to give a GE-QR secure network extractor for the case where $k$ is sufficiently larger than $t$.    

\paragraph{Our GE-QR Secure Network Extractor for Sufficiently Large $k$} We present a formal description of the above protocol in Fig.~\ref{fig:GE-QR-net-ext}. Note that in the actual protocol, we only require marginal security from the multi-source extractors. We use the construction to prove Theorem~\ref{thm:QR-GEA-net-ext}.


\begin{figure}
\begin{protocol*}{Protocol $\NetExt$: GE-QR secure network extractor.}
\begin{description}
\item Protocol Input: A private weak sources $x_i$ for each $i\in P$.
\item Protocol Output: A private output string $z_i$ for each $i\in P$.
\item Sub-Routines and Parameters:
\begin{step}
\item Let $\IExt$ be a $(d,n,k,m,\eps_1)$ independent source extractor with some constant $d$ and $\eps_1 \leq n^{-2c}$ from Theorem~\ref{thm:Li_IExt}.
\item Let $\QTExt_{\Raz}(X,Y)$ be the $Y$-strong quantum-secure two-source extractor from Theorem~\ref{thm:Raz-GE-secure} for for sources with min-entropy at least $0.9k$ and error $\eps_2 \leq 2^{-\alpha k}$ and output length $m_2 = \Omega(k)$. 
\end{step}
\item Round $1$.
\begin{step}
\item Let $s = t/2\alpha$ (where $\alpha$ is the constant in the exponent of the error of $\QTExt$). For each $i \in [s]$, let $A_i = \{(i-1)\cdot d + 1,\dots, i\cdot d\}$. Let $B = P \backslash (A_1 \cup \dots \cup A_s)$.
\item All players $i$ in $A_1,\dots, A_s$ publish their input $x_i$ and output $z_i = \bot$. For each $i\in [s]$, let $y_i$ be the first $k/s$ bits of $\IExt(x_{(i-1)d+1},\dots, x_{i\cdot d})$. Let $y = (y_1,\dots, y_s) \in \zo^k$.
\item For each $i \in B$, player $i$ computes $z_i = \QTExt_{\Raz}(x_i,y)$ and outputs $z_i$. The remaining players $i\notin B$ output $z_i = \bot$.
\end{step}
\end{description}
\end{protocol*}
\caption{Our GE-QR secure network extractor.}
\label{fig:GE-QR-net-ext}
\end{figure}

\begin{proof}{\bf (of Theorem~\ref{thm:QR-GEA-net-ext}; sketch)} We first note that the protocol has the same structure as our GE-IR secure network extractor constructed in Section \ref{subsec:CR-net-ext}, where a public high min-entropy source is published, and used to extract private randomness for the remaining players. Therefore, an analogous analysis proves that $\NetExt$ in Fig~\ref{fig:GE-QR-net-ext} is OA-IR secure with error $\eps' = s\eps_1 + \eps_2$ for players in set $B$. It follows by Theorem~\ref{thm:OA-GE-eq-for-net-ext} that $\NetExt$ is GE-IR secure with error $\eps'$ for players in set $B$. We next demonstrate how to apply Theorem~\ref{thm:CR_QR} to show that $\NetExt$ is GE-QR secure with error $2^{kt/s} \cdot \eps'$. Note that since $\eps_1 = 1/\poly(n)$, $2^{kt/s} \cdot \eps' > 1$ so the conclusion is not useful. Nevertheless, we discuss how to modify the proof to avoid the loss of $2^{kt/s} \cdot \eps_1$ afterword.

To apply Theorem~\ref{thm:CR_QR}, we need to argue that the premise of the theorem holds for some $\rho_{X'Y'\Adv'}$ system with rushing part $Y'_R$ and output $Z'$, described in Theorem~\ref{thm:CR_QR}. Let us consider a honest player $j\in B$. Let $A = \bigcup_{i} A_i$. We set $X' = X_j$, $Z'=Z_j$ $Y' = X_A$, and let $Y'_R$ be the components of $Y$ in the protocol that are subject to rushing. Note that while the components depends on the set $\Faulty$ of faulty players, but the length $|Y'_R|$ is always bounded by $kt/s$, since $t$ faulty players can only control up to $t$ groups. Finally, let $\Adv'$ be the remaining quantum system. Note that the GE-IR security of player $j$ with error $\eps'$ implies the premise of Theorem~\ref{thm:CR_QR} with error $\eps'$. Therefore, the conclusion of Theorem~\ref{thm:CR_QR} implies that player $j$ is GE-QR secure with error $2^{kt/s} \cdot \eps'$.

As mentioned, $2^{kt/s} \cdot s\eps_1 > 1$ so the conclusion is not useful. Note, however, the $s\eps_1$ error comes from application of $\IExt$, and we only need to pay the error for the honest groups. To avoid paying $2^{kt/s} \cdot s\eps_1$, we can first switch to a hybrid input distribution $X'$ such that the application of $\IExt$ to the honest groups produce perfectly uniform output. Then, it can be shown by similar steps as before that a honest player $j \in B$ is GE-IR secure with error $\eps_2$. We can then apply Theorem~\ref{thm:CR_QR} as before to show that player $j$ is GE-QR secure with error $2^{kt/s} \cdot \eps_2 \leq 2^{\alpha k/2}$. Finally, we can switch back to the real experiment, and conclude that player $j$ is GE-QR secure with error $2^{\alpha k/2} + s\eps_1$.

We defer a full proof to the full version of this paper.
\end{proof}

\paragraph{Sketch of handling $k<t$.} When $k < t$, the above approach fails because we could have $t$ faulty players in $B$, which makes $|Y_{\Bad}|>t$ while the error of the extractor is always at most $2^{-k}$. To deal with this, we have to reduce the size of $Y_{\Bad}$. In other words, we need to somehow be able to select a small subset from $B$ that roughly contains the same fraction of honest players. One natural way to do this is to sample a random subset of $B$. However, this is problematic because we need private uniform random bits to sample, which we do not have (in fact, this is our goal). Fortunately, we can use other combinatorial tools to do this step. 

Specifically, here we will use an extractor graph. An $[N, M, K, D, \eps]$ extractor graph is a bipartite graph with left vertex set $[N]$, right vertex set $[M]$, left degree $D$. It has the property that  for any subset $T \subset [M]$ with $|T|=\alpha M$, all but $K$ vertices in $[N]$ have roughly $\alpha$ fraction of neighbors in $T$ (with a deviation of at most $\eps$).  Non-constructively, $\forall N>K>0, \eps>0$ such graphs exist with $D=O(\log(N/K)/\eps^2)$ and $M=\Omega(KD\eps^2)$.\footnote{We also have explicit constructions, such as \cite{GuruswamiUV09}.} To apply an extractor graph here, we can identify the set $B$ with $[M]$ and identify the set $C$ of remaining players with $[N]$. We will then have each player in $C$ choose its neighbors in $B$ as a set $S$, and use $Y_S$ as the random string to apply $ \QTExt_{\Raz}$. This will ensure that most of $Y_S$ will roughly have the same fraction of entropy rate as $Y$. Note here we can choose $\eps$ to be a small enough constant and choose $K=o(t)$. Thus the degree $D=O(\log N)=O(\log p)$. By our assumption that $k> C \log p$ for some big enough constant $C>1$, this will ensure that $k>D$ and thus we can afford to use $Y_S$ in $ \QTExt_{\Raz}$ for quantum rushing. Note that in this way we only lose $o(t)$ honest players in $C$. However, one slight drawback is that the honest players do not know if they have obtained private uniform random bits in the end, as they do not know if they are the $K$ unlucky players given by the extractor graph.

\bibliographystyle{}

\bibliography{}

\begin{thebibliography}{10}
\bibitem{BarakIW04}
B.~Barak, R.~Impagliazzo, and A.~Wigderson.
\newblock Extracting randomness using few independent sources.
\newblock In {\em Proceedings of the 45th Annual IEEE Symposium on Foundations
  of Computer Science}, pages 384--393, 2004.

\bibitem{BarakKSSW05}
B.~Barak, G.~Kindler, R.~Shaltiel, B.~Sudakov, and A.~Wigderson.
\newblock Simulating independence: New constructions of condensers, {R}amsey
  graphs, dispersers, and extractors.
\newblock In {\em Proceedings of the 37th Annual ACM Symposium on Theory of
  Computing}, pages 1--10, 2005.

\bibitem{BarakRSW06}
B.~Barak, A.~Rao, R.~Shaltiel, and A.~Wigderson.
\newblock 2 source dispersers for {$n^{o(1)}$} entropy and {R}amsey graphs
  beating the {F}rankl-{W}ilson construction.
\newblock In {\em Proceedings of the 38th Annual ACM Symposium on Theory of
  Computing}, 2006.
  
\bibitem{Bell64}
J.~S. Bell.
\newblock On the {E}instein-{P}odolsky-{R}osen paradox.
\newblock {\em Physics}, 1(3):195--290, 1964.

\bibitem{BW92}
C.~Bennett, S.~Wiesner.
\newblock{Communication via one- and two- particle operators on Einstein-Podolsky-Rosen states.}
\newblock{Phys. Rev. Lett., 69(20):2881--2884, 1992.}

\bibitem{Bour05}
J.~Bourgain.
\newblock{More on the sum-product phenomenon in prime fields and its applications.}
\newblock{Internat. J. Number Theory, 1(1):1--32, 2005.}

\bibitem{CG88}
B.~Chor, O.~Goldreich.
\newblock{Unbiased bits from sources of weak randomness and
probabilistic communication complexity.}
\newblock{SIAM J. Comput., 17(2):230--261, 1988.}

\bibitem{DPVR12}
A.~De, C.~Portmann, T.~Vidick,  R.~Renner. 
\newblock{Trevisan's
extractor in the presence of quantum side information.}
\newblock{SIAM J. Comput., 41(4):915--940, 2012.}

\bibitem{DV10}
A.~De, T.~Vidick.
\newblock{Near-optimal extractors against quantum storage.}
\newblock{In Proc. 42nd STOC, pp. 161--170. ACM Press, 2010.}

\bibitem{DEOR04}
Y.~Dodis, A.~Elbaz, R.~Oliveira,  R.~Raz. 
\newblock{Improved randomness
extraction from two independent sources.}
\newblock{In Proc. 8th Internat. Workshop on Randomization and
Computation (RANDOM?04), pp. 334--344. Springer, 2004.}

\bibitem{DodisOPS04}
Yevgeniy Dodis, Shien~Jin Ong, Manoj Prabhakaran, and Amit Sahai.
\newblock On the (im)possibility of cryptography with imperfect randomness.
\newblock In {\em FOCS04}, pages 196--205, 2004.

\bibitem{FS08}
S.~Fehr, C.~Schaffner.
\newblock{Randomness extraction via d -biased masking in the
presence of a quantum attacker.}
\newblock{In 5th Theory of Cryptography Conf. (TCC?08), pp. 465--481, 2008.}

\bibitem{GKK+08}
D.~Gavinsky, J.~Kempe, I.~Kerenidis, R.~Raz, R.~de Wolf. 
\newblock{Exponential separation for one-way quantum communication complexity, with applications to cryptography.}
\newblock{SIAM J. Comput., 38(5):1695--1708, 2008.}

\bibitem{GoldwasserSV05}
S.~Goldwasser, M.~Sudan, and V.~Vaikuntanathan.
\newblock Distributed computing with imperfect randomness.
\newblock In {\em DISC 2005}, 2005.

\bibitem{GuruswamiUV09}
Venkatesan Guruswami, Christopher Umans, and Salil Vadhan.
\newblock Unbalanced expanders and randomness extractors from
  {P}arvaresh-{V}ardy codes.
\newblock {\em Journal of the ACM}, 56(4), 2009.

\bibitem{hill}
J.~H{\aa}stad, R.~Impagliazzo, L.~Levin, and M.~Luby.
\newblock A pseudorandom generator from any one-way function.
\newblock {\em SIAM Journal on Computing}, 28:1364--1396, 1999.

\bibitem{ILL89}
R.~Impagliazzo, L.~Levin, and M.~Luby.
\newblock Pseudo-random generation from one-way functions.
\newblock In {\em Proceedings of the 21st Annual ACM Symposium on Theory of
  Computing}, pages 12--24, 1989.

\bibitem{KLRZ}
Y.~Kalai, X.~Li, A.~Rao, and D.~Zuckerman.
\newblock Network extractor protocols.
\newblock In {\em Proceedings of the 49th Annual IEEE Symposium on Foundations
  of Computer Science}, pages 654--663, 2008.

\bibitem{KK12}
R.~Kasher, J.~Kempe. 
\newblock{Two-Source Extractors Secure Against Quantum Adversaries.}
\newblock{Theory of Computing, vol 8, pp. 461--486, 2012.}

\bibitem{KMR05}
R.~K$\ddot{o}$nig, U.~Maurer, R.~Renner.
\newblock{On the power of quantum memory.}
\newblock{IEEE Trans. Inform. Theory, 51(7):2391--2401, 2005.}

\bibitem{KRS09}
R.~K$\ddot{o}$nig, R.~Renner, C.~Schaffner.
\newblock{The operational meaning of min- and max- entropy}
\newblock{IEEE Trans. Inform. Theory, 55:4337--4347, 2009.}

\bibitem{KT08}
R.~T.~K$\ddot{o}$nig, B.~M.~Terhal.
\newblock{The bounded-storage model in the
presence of a quantum adversary.}
\newblock{IEEE Trans. Inform. Theory, 54(2):749--762, 2008.}

\bibitem{Li11}
X.~Li.
\newblock Improved constructions of three source extractors.
\newblock In {\em Proceedings of the 26th Annual IEEE Conference on
  Computational Complexity}, 2011.

\bibitem{Li13b}
X.~Li.
\newblock Extractors for a constant number of independent sources with
  polylogarithmic min-entropy.
\newblock In {\em Proceedings of the 54th Annual IEEE Symposium on Foundations
  of Computer Science}, 2013.

\bibitem{Li13a}
X.~Li.
\newblock New independent source extractors with exponential improvement.
\newblock In {\em Proceedings of the 45th Annual ACM Symposium on Theory of
  Computing}, 2013.

\bibitem{LubotzkyPS88}
A.~Lubotzky, R.~Phillips, and P.~Sarnak.
\newblock Ramanujan graphs.
\newblock {\em Combinatorica}, 8(3):261--277, 1988.

\bibitem{MW97}
Ueli~M. Maurer and Stefan Wolf.
\newblock Privacy amplification secure against active adversaries.
\newblock In {\em CRYPTO '97}, 1997.

\bibitem{NS06}
A.~Nayak, J.~Salzman.
\newblock{Limits on the ability of quantum states to convey
classical messages.}
\newblock{J. ACM, 53(1):184--206, 2006.}

\bibitem{Pippenger87}
N.~Pippenger.
\newblock Sorting and selecting in rounds.
\newblock {\em SIAM Journal on Computing}, 16(6):1032--1038, 1987.

\bibitem{Rao06}
A.~Rao.
\newblock Extractors for a constant number of polynomially small min-entropy
  independent sources.
\newblock In {\em Proceedings of the 38th Annual ACM Symposium on Theory of
  Computing}, 2006.

\bibitem{Raz05}
R.~Raz.
\newblock Extractors with weak random seeds.
\newblock In {\em Proceedings of the 37th Annual ACM Symposium on Theory of
  Computing}, pages 11--20, 2005.

\bibitem{Renner05}
R.~Renner.
\newblock{Security of Quantum Key Distribution.}
\newblock{PhD thesis, ETH Zurich, 2005.}

\bibitem{RK05}
R.~Renner, R.~K$\ddot{o}$nig.
\newblock{Universally composable privacy amplification against quantum adversaries.}
\newblock{Proceedings of the 2nd Theory of Cryptography Conference (TCC), pp. 407--425. Springer, 2005.}

\bibitem{TaS11}
A.~Ta-Shma.
\newblock{Short seed extractors against quantum storage.}
\newblock{SIAM J. Comput., 40(3):664--677, 2011.}

\bibitem{TSSR11}
M.~Tomamichel, C.~Schaffner, A.~Smith, R.~Renner.  
\newblock{Leftover hashing against quantum side information.}
\newblock{IEEE Trans. Inform. Theory, 57(8):5524--5535, 2011.}

\bibitem{Tre01}
L.~Trevisan.
\newblock Extractors and pseudorandom generators.
\newblock {\em Journal of the ACM}, pages 860--879, 2001.

\bibitem{Vaz87}
U.~V.~Vazirani. 
\newblock{Strong communication complexity or generating quasirandom sequences
from two communicating semi-random sources.}
\newblock{Combinatorica, 7(4):375--392, 1987.}

\bibitem{Zuc07}
David Zuckerman.
\newblock Linear degree extractors and the inapproximability of max clique and
  chromatic number.
\newblock In {\em Theory of Computing}, pages 103--128, 2007.


\end{thebibliography}

\end{document}